\documentclass[a4paper, 11pt]{article}
\usepackage{calc}

\usepackage[utf8]{inputenc}
\usepackage[margin=1.25in,centering]{geometry}
\usepackage[round]{natbib}
\usepackage{graphicx}
\usepackage{mathtools}
\usepackage{bm} %for making bold greek symbols
\usepackage{latexsym}
\usepackage{amsfonts}
\usepackage{amssymb}
\usepackage{amsbsy} %for bold greek letters
\usepackage{amsthm}
\usepackage{amsmath}
\usepackage{color}
\usepackage{todonotes}
\usepackage{subcaption}
\usepackage[font={small,sc}]{caption}
\usepackage{enumitem}
\usepackage{verbatim}
\usepackage[normalem]{ulem}
\theoremstyle{definition}

\newtheorem{corollary}{Corollary}
\newtheorem{lemma}{Lemma}
\newtheorem{proposition}{Proposition}
\newtheorem{theorem}{Theorem}

\newtheorem*{definition*}{Definition}

\DeclareMathOperator{\cov}{cov}

\usepackage{authblk} %affiliations in title

\usepackage{fancyhdr}
\usepackage{hyperref}
\usepackage{setspace}
\hypersetup{breaklinks,
	colorlinks=false,
	pdfborder={0 0 0},}
\usepackage[toc,page]{appendix} %appendices

\usepackage{tikz}
\usepackage{units}
\usepackage{graphicx}
\usepackage{color}
\usepackage{pgfplots}
\usetikzlibrary{calc,patterns}

\graphicspath{{pics/}} %path to the folder with pictures

\begin{document}
	%\clubpenalty 9999%not so many orphants
	%\widowpenalty 9999%not so many widows
	
	%-----<<< TITLE PAGE >>>-----
	
	\title{\textbf{Strategic Attribute Learning\footnote{For valuable comments, we thank: Arjada Bardhi, Christoph Carnehl, Kfir Eliaz, Nathan Hancart, Toomas Hinnosaar, Alessandro Ispano, Igor Letina, Antoine Loeper, Marc M\"{o}ller, Nick Netzer, Francisco Poggi, Johannes Schneider, Armin Schmutzler, Jakub Steiner, Peter Norman S{\o}rensen, Dezsö Szalay, and the audiences at the Universities of Alicante, Bern, Copenhagen, Essex, Madrid Carlos III, Mannheim, Zurich, and CERGE-EI, as well as the Nordic Theory Meetings, the European Economic Theory Conference, the Society for the Advancement of Economic Theory Conference (SAET), the Annual Meeting of European Association of Young Economists, the Lisbon Meetings in Game Theory and Applications, the Barcelona Summer Forum, EEA-ESEM, the Symposium of the Spanish Economic Association (SAEe), and the Annual Meeting of the VfS. Matyskov\'{a} gratefully acknowledges funding from the Generalitat Valenciana (Prometeo/2021/073) and from the grant PID2022-142356NB-I00 financed by MICIU/AEI /10.13039/501100011033 and FEDER, UE.}}}
	
	\author{Jean-Michel Benkert, Ludmila Matyskov\'{a}, and Egor Starkov\footnote{Benkert: Department of Economics, University of Bern. Matyskov\'{a}: Department of Economics, University of Alicante.  Starkov: Department of Economics, University of Copenhagen. Email: jean-michel.benkert@unibe.ch, ludmila.matyskova@ua.es, egor.starkov@econ.ku.dk.}
}
	
	\setcounter{Maxaffil}{0}
	\renewcommand\Affilfont{\itshape\small} 
	\maketitle
	\vspace{-5ex}
	\begin{abstract}
A researcher allocates a budget of informative tests across multiple unknown attributes to influence a decision-maker. We derive the researcher's equilibrium learning strategy by solving an auxiliary single-player problem. The attribute weights in this problem depend on how much the researcher and the decision-maker disagree. If the researcher expects an excessive response to new information, she forgoes learning altogether. In an organizational context, we show that a manager favors more diverse analysts as the hierarchical distance grows. In another application, we show how an appropriately opposed advisor can constrain a discriminatory politician, and identify the welfare-inequality Pareto frontier of researchers.
	\end{abstract}
	
	\setlength{\parindent}{0cm}

	\setlength{\parindent}{0cm}
	
	\newpage
	
	\spacing{1.45}

	%================================
	\section{Introduction}
	%================================

	Decision-making under uncertainty often involves multiple dimensions of potentially unequal importance. For example, when designing a new product, a company may consider the product's uncertain reception across various market segments, some of which may be more important than others. Similarly, a politician's optimal policy may need to address the diverse and uncertain needs of distinct social groups, with some groups potentially carrying more weight than others do.

	In both examples, decision-makers must learn about the various dimensions, or attributes, influencing their decisions but often lack the resources for a thorough investigation. Consequently, a specialized researcher is often in charge of the learning task. However, the researcher may prioritize attributes differently from the decision-maker and strategically procure information about the unknown attributes to influence the final decision.

	In this paper, we explore questions pertaining to such strategic concerns in learning about complex decisions. First, what kinds of biases may arise when multiple attributes are at play, and how do they interact? For instance, in an organizational application, we explore two orthogonal biases---hierarchical distance and diversity---and analyze when a manager (decision-maker) prefers diverse analysts (researchers) as a function of hierarchical distance. Second, how can biases mitigate socially undesirable decisions? For instance, in the context of political discrimination, we show how an appropriately opposed advisor (researcher) can constrain a discriminatory politician whose decision affects well-being of various social groups (attributes). More broadly, how does preference misalignment in a multiattribute environment affect learning, when researchers shape not only the \emph{extent} of learning but also its \emph{direction}?

	To address these questions, we develop a framework that captures the strategic interaction between a \emph{researcher}, who learns about an unknown state of the world, and a \emph{decision-maker}, who acts based on the resulting information. Crucially, we assume that the state of the world consists of multiple \emph{attributes} and capture the player's preference misalignment by allowing for differing weights on the attributes. As such, our main contribution is a tractable model for studying how preference misalignment affects multi-attribute learning and decision-making.

	In Section \ref{sec:model}, we formally introduce our framework. In our model, the attributes are independently and normally distributed and jointly determine the state of the world and thus the players' preferred decisions. The state can be imperfectly learned by allocating a given budget of \emph{tests} across different attributes. The more tests are allocated to an attribute, the more informative a signal about it becomes. Players aim to minimize the quadratic loss between the final decision and their \emph{bliss point}, which is a weighted sum of attributes.

We begin by analyzing a helpful benchmark in Section \ref{sec:single_player}: a situation in which a single player controls both the decision and the learning, where the optimal learning strategy is the primary focus of our analysis. Theorem \ref{thm:benchmark} shows that---absent strategic motives---the player chooses the test allocation which achieves the highest reduction in residual uncertainty regarding her optimal decision. In particular, optimal learning reflects both the attribute's relevance to the agent and its prior uncertainty. Put differently, the more important and more uncertain an attribute is, the more tests are allocated to it. Moreover, the budget size determines how many attributes the agent learns about.

In Section \ref{sec:strategic_players}, we return to the primary framework, in which a researcher controls the learning and a decision-maker makes the decision. Our main result (Theorem \ref{thm:strategic_solution}) shows that the researcher's equilibrium test allocation then coincides with the single-player solution with \emph{auxiliary weights}, which are different from the actual weights of either player. We show that these auxiliary weights decrease as the misalignment between players grows. If the researcher values an attribute less (or more) than the decision-maker does, he views the decision-maker's reaction to any information about that attribute as excessive (or insufficient). When the misalignment results in overreaction, it can lead the researcher to abstain from learning about the attribute entirely. Notably, when all of the researcher's weights are sufficiently low relative to the decision-maker's weights, the researcher entirely forgoes learning---contrasting with the single-player case, where learning generically takes place. The fact that the researcher's equilibrium test allocation coincides with the single-player solution with auxiliary weights clarifies how the strategic motive influences the researcher's learning, reducing its impact to a shift in attribute weights. Thus, by comparing these auxiliary weights to the researcher's original weights, we can pinpoint how misalignment drives adjustments in the researcher's learning strategy---insight that would be harder to achieve if the strategic solution deviated fundamentally from the single-player case.

We explore two economical applications in Section \ref{sec:applications}, each with a tailored parametrization of player misalignment. These applications demonstrate the versatility of our analytical tools in studying misalignment and yield new insights within each specific context. 
	In our first application, we examine how biases across layers of hierarchy influence learning in organizations. Here, an analyst (researcher)  gathers information, and a manager (decision-maker) acts on it. While the manager trivially prefers an analyst whose preferences match hers, we examine the preferred analyst when full alignment is unfeasible.
	We parametrize misalignment by decomposing the analyst's preferences into \emph{sensitivity} (how closely the players align in absolute terms) and \emph{distortion} (how much the players disagree in relative terms). Proposition~\ref{prop:CS_manager_gamma} shows that the manager may prefer analysts with high distortion when the analysts have low sensitivity. An analyst with low sensitivity views the manager's reaction to any information as excessive and thus may completely forgo learning. In such cases, the manager then prefers an analyst who is very distorted toward some attribute, as then the manager's reaction to information about that attribute is no longer viewed as excessive, resulting in some learning. As any learning is better for the manager than no learning, she prefers a strongly distorted analyst over one with little or no distortion. Assuming that the more hierarchical the organization, the bigger the difference in the sensitivity between the managers and the analysts, our results suggest balancing diversity and uniformity (reflected in distortion) in organizations as a function of hierarchical distance.

	In our second application, we explore the issue of discrimination. A politician (decision-maker) chooses a policy to meet the uncertain needs of two a priori identical social groups (attributes) but favors one group over the other. With an unchecked politician---who controls both decision-making and learning---such favoritism results in inequality and reduces utilitarian welfare compared to the situation without favoritism. We examine how  delegating learning to an advisor (researcher) with different preferences can mitigate these negative effects.
	Proposition~\ref{prop:discrimination_frontier} demonstrates that there exists a welfare-inequality Pareto frontier of potential advisors that strictly dominates unchecked discrimination. At one end of this frontier is an impartial advisor, who---given the politician's favoritism---maximizes welfare, but leaves the groups with unequal outcomes (Lemma~\ref{prop:discrimination_impartial}). At the other end of the frontier is an advisor who is suitably partial towards the disadvantaged group. Such an advisor learns precisely enough about this group to counteract the politician's favoritism, and thus eliminates inequality when appointed (Lemma~\ref{prop:discrimination_equalizing_bias}). However, it comes at a cost to welfare, as the politician then relies less on the advisor's information, and information is thus ``wasted.'' Notably, the level of advisor's partiality required to eliminate inequality does not increase with the politician's favoritism but exhibits a \emph{non-monotone} relationship.

	In Section \ref{sec:extensions}, we outline two extensions developed in detail in the online appendix. The first examines multiple researchers in the context of media markets, showing that while competition between media outlets (researchers) leads to polarization in equilibrium, this outcome is actually beneficial for a voter (decision-maker). The second extension addresses uncertainty about the decision-maker's preferences. We study this within a dual-self model, where a single agent may be either sophisticated or na\"ive about potential changes in her preferences between the learning and the decision-making stages, finding that na\"ivete can, in fact, be advantageous. 

	We conclude in Section \ref{sec:conclusion} by discussing our results and prospective avenues for future research.
	
	\paragraph*{Related literature.} 

    Our work is closely related to \citet{bardhi_attributes:_2024}, who examines a strategic multi-attribute learning problem where a project's payoff is composed of correlated attributes, weighted differently by a decision-maker and a researcher. The researcher samples a finite number of attributes and uses the observed realizations to draw inferences about other attributes. In this delegated learning problem, \citet{bardhi_attributes:_2024} explores how the researcher can optimally leverage correlation between the attributes to determine \emph{which} attributes to learn about.\footnote{More broadly, the framework of \citet{bardhi_attributes:_2024} is related to the Gaussian sampling literature such as \citet{bardhi_local_2023}, \citet{callander_searching_2011}, and \citet{carnehl_quest_2024}.} Importantly, this learning decision is binary for each attribute. In contrast, our study encompasses both the extensive and the intensive margin, so that we can derive which attributes the research learns about and in what intensity. Moreover, abstracting from correlation between attributes allows us to obtain closed-form solutions, giving a clear view on the effects of preference misalignment and enabling the exploration of various applications.\footnote{We discuss the case of correlated attributes in more detail in Section \ref{sec:correlation}.}

	\citet{kirneva_informing_2023}, building on \citet{tamura2018bayesian}, considers a related multi-dimensional setting, where the decision-maker can acquire additional costly information on top of what is provided by the researcher. The need to influence the decision-maker's learning strategy (and not only the decision) then shapes the researcher's learning strategy. The researcher may then provide partial information about the attribute on which the players' preferences are misaligned in order to divert the decision-maker's learning away from it.\footnote{Similar insights---namely, that when the receiver can acquire information, the sender's choices are primarily driven by a desire to influence the learning process rather than the final decision---have also been obtained by \citet{matveenko_sparking_2023}.} In contrast, we do not allow for independent learning by the decision-maker, creating a distinct strategic environment and shutting down the ``attention diversion'' channel.

	More broadly, our paper is related to the literature on strategic learning in Gaussian environments, see \cite{veldkamp_information_2011} for an overview. While there is extensive literature in macroeconomics in finance using such models, we are not aware of any models of delegated learning using this framework. Related also is the literature on information design and Bayesian Persuasion in multi-dimensional Gaussian settings, see \citet{tamura2018bayesian}, \citet{dworczak_persuasion_2024}, and \citet{miyashita_lqg_2025} for recent contributions. Unlike these contributions, we primarily focus on the effects of misalignment between the researcher and the decision-maker.

	Our framework is related to \citet{liang_dynamically_2022}, who study a single agent's dynamic learning under attribute correlation. In contrast, we consider a static framework and  focus on strategic learning, abstracting  from correlation to highlight the role of preference misalignment between players. Notably, there is a key parallel in our findings. Both papers show that a ``greedy'' learning strategy, which achieves the highest reduction in residual uncertainty and is optimal in a single-player model without correlation, remains optimal in the presence of correlation \citep{liang_dynamically_2022} and in the presence of strategic motives (our paper), respectively.

	Our first application, on diversity in organization, connects to the literature on delegated expertise, where a decision-maker delegates learning about the state of interest to a biased researcher.\footnote{Pioneered by \citet{demski_delegated_1987}, this literature has seen renewed interest, for instance, see \citet{deimen_delegated_2019} and \citet{lindbeck_delegation_2020}.} We explore the effects of different biases between the researcher and the decision-maker, which is related to studies by \citet{ball_benefitting_2021}, \citet{che_opinions_2009}, and \citet{ilinov_optimally_2022}.  These papers show that delegating learning to a researcher with some preference misalignment can be optimal in single-dimensional settings, as misalignment encourages the acquisition of more costly information. In contrast, we explore a multi-dimensional setting, where the researcher decides \emph{which attributes} to learn about and to what extent. Our findings in Section \ref{sec:organizations} reveal that higher misalignment on a certain type of bias can be beneficial, as it may prevent the breakdown of learning caused by misalignment on another, orthogonal, type of bias.
    
	Our second application, on discrimination in policymaking, relates to \cite{fosgerau_equilibrium_2023} and \cite{echenique_rationally_2023}, who explore how decision-makers' strategic learning choices reinforce discrimination.\footnote{See \citet{Onuchic2024} for a recent overview of theories of discrimination.} They show that employers' discriminatory beliefs shape job candidates' incentives to invest in their skills, potentially leading to a self-sustaining discriminatory equilibrium. In contrast, we focus on countering discriminatory tendencies by separating learning from decision-making and delegating it to independent advisors. Further, our results in Section \ref{sec:discrimination} complement \citet{liang_algorithm_2024}. They study the fairness-accuracy Pareto frontier in the context of an interaction between an egalitarian researcher and a utilitarian decision-maker, where the researcher may coarsen or ban information about certain attributes. In contrast, we examine the welfare-inequality frontier in a setting where the decision-maker is inherently unfair but can be paired with a researcher holding different preferences, who can flexibly allocate informative tests across different attributes.

	%=================================
	%================================
	\section{Model}\label{sec:model}
	%================================
	%==================================
	
	%==============================
	\subsection{Setup}
	%==============================

	\paragraph*{Players.} There are two players, a decision-maker ($D$, she) and a researcher ($R$, he). First, the researcher chooses the learning strategy to examine payoff-relevant attributes of the state of the world by choosing how to allocate a budget of tests $T$. Second, the decision-maker observes the results of these tests and makes a decision. 
	
	\paragraph*{Attributes.} There is an unknown state of the world $\tilde{\theta}= (\tilde{\theta}_1,\hdots,\tilde{\theta}_K)$ consisting of $K\geq 2$ attributes.\footnote{We denote a random variable and its realization by $\tilde{x}$ and $x$, respectively.}	All attributes $\tilde{\theta}_k$ for $k = 1,\ldots,K$ are jointly normally distributed with commonly known prior means $\mu_k^0 \in \mathbb{R}$ and prior variances $\Sigma^0_k>0$; that is, $\tilde{\theta}_k \sim \mathcal{N} \left( \mu_k^0, \Sigma^0_k  \right).$ Moreover, the attributes are independent: $\tilde{\theta}_l \perp \tilde{\theta}_j$ for $l\neq j$.
	
	\paragraph*{Actions.} The decision-maker chooses a decision $d\in \mathbb{R}$. The researcher, given an exogenous budget of tests $T>0$, chooses a test allocation $\tau \in \mathcal{T} := \{\tau \in \mathbb{R}^K_+: \tau_1 + \hdots + \tau_K \leq T\}$, where $\tau_k$ is the amount of tests allocated to learn about attribute $\tilde{\theta}_k$. Allocating $\tau_k$ to attribute $\tilde{\theta}_k$ yields a single observation of a signal,
	\begin{eqnarray} \label{eq:signal_defn}
		\tilde{s}_k(\tau) = \tilde{\theta}_k + \tilde{\varepsilon}_k,
	\end{eqnarray}
	where $\tilde{\varepsilon}_k \sim \mathcal{N}(0,1/\tau_k)$, and where $\tilde{\varepsilon}_k$ is independent of all other random variables.\footnote{When $\tau_k = 0$, the signal $\tilde{s}_k(\tau)$ is considered completely uninformative about $\tilde{\theta}_k$.}$^,$\footnote{Put differently, the researcher chooses how to allocate a budget of precision $T$ across different attributes.  An equivalent interpretation, abstracting from integer constraints, is that allocating $\tau_k$ tests to attribute $k$ produces $\tau_k$ independent realizations of the standard normal signal $\theta_k + \mathcal{N}(0,1)$.}
	When the researcher chooses $\tau=(0,\hdots,0)$, we say \emph{the researcher abstains from learning}. The test allocation $\tau = (\tau_1, \tau_2, \hdots, \tau_K)$ and the subsequent signal realizations $s = (s_1, s_2, \hdots, s_K)$ are publicly observed, and the players are symmetrically informed throughout the game.

	\paragraph*{Payoffs.}
    Each player $i=D,R$ has a bliss point, which is
%	The players disagree on the importance of the attributes, which is captured by having different weights assigned to each attribute in their respective bliss point. The bliss point of player $i\in \{D,R\}$ is 
    a linear combination of the realized attribute values:
	\begin{align*}
		b^i(\theta) :=  \alpha_1^i\theta_1 + \hdots + \alpha_K^i\theta_K
	\end{align*}
	for commonly known weights $\alpha^i = (\alpha_1^i,\hdots, \alpha_K^i) \in \mathbb{R}^K_+$.\footnote{\label{foot:alpha_positive}The non-negativity assumption $\alpha^i_k \geq 0$ simplifies the exposition and is without loss of generality. If the players' weights $\alpha^{D}_k$ and $\alpha^R_k$ had opposite signs, the researcher would never learn about attribute $\tilde{\theta}_k$, effectively equivalent to setting $\alpha^R_k = 0$. Additionally, since any attribute can be replaced by its negative, if $\alpha^{D}_k, \alpha^R_k < 0$ for some $k$, we can instead consider $-\tilde{\theta}_k$ with positive weights $-\alpha^{D}_k, -\alpha^R_k > 0$.} 
	The players may weigh attributes differently, which is the source of their conflict. When the decision-maker takes a decision $d\in \mathbb{R}$ and the realized state is $\theta$, player $i$ obtains ex post utility
	\begin{eqnarray} \nonumber
		u^i(d,\theta) = - (d-b^i(\theta))^2.
	\end{eqnarray}

	From the ex ante perspective, $i$'s bliss point $b^i(\theta)$ is distributed according to
	\begin{equation} %\label{eq:bliss_prior}
	\nonumber \begin{aligned}
		\tilde{b}^i &:= b^i(\tilde{\theta}) \sim 
		\mathcal{N}
		\left(
		b_0^i, \sigma^{2,i}_0 
		\right),
	\end{aligned}
	\end{equation}
	where $b_0^i = \alpha_1^i \mu_1^0 + \hdots + \alpha_K^i \mu_K^0$ is player $i$'s ex ante expectation of their bliss point, and $\sigma^{2,i}_0 = \left(\alpha_1^i\right)^2 \Sigma^0_1 + \hdots + \left(\alpha_K^i\right)^2 \Sigma^0_K$ is the prior variance of player $i$'s bliss point. 

	\paragraph*{Timing.} First, at the ex ante stage, the researcher chooses a publicly observable test allocation $\tau$ to maximize his ex ante expected utility
	\begin{align} \label{eq:util_exante_R}
		V^R(\tau) := \mathbb{E} \left[ u^R \left( d \left( \tilde{s}, \tau \right), \tilde{\theta} \right) \right],
	\end{align} 
	where, given some anticipated decision strategy $d(s,\tau)$, the expectation is taken over the induced distribution of the decisions and the state of the world.
	Next, nature draws the state realization $\theta$ and signal realizations $s$. The players observe $s$ and update their beliefs about the state (and bliss points). Then, at the interim stage, the decision-maker chooses decision $d$ to maximize her interim expected utility 
	\begin{align} \label{eq:util_interim_DM}
		U^{D}(d, s, \tau) := \mathbb{E} \left[ u^{D} \left( d, \tilde{\theta} \right) \right],
	\end{align}
	where the expectation is taken over the state.
	Finally, at the ex post stage, payoffs are realized. 

	\paragraph*{Equilibrium.} A weak Perfect Bayesian Equilibrium of the game (henceforth, ``equilibrium'') is a triple $(\tau^*, d^*(s, \tau), \tilde{\theta}(s,\tau))$ such that
	\begin{enumerate}
		\item the researcher's test allocation strategy $\tau^*\in \mathcal{T}$ maximizes his expected utility \eqref{eq:util_exante_R} given the decision-maker's strategy $d^*(s,\tau)$ and her prior belief; whenever the researcher is indifferent between some $\tau$ and abstaining from learning, he chooses the latter;\footnote{The tie-breaking rule resolves a possible multiplicity of equilibria, which only arises in non-generic cases.}
		\item the decision-maker's decision strategy $d^*(s,\tau): \mathbb{R}^K \times \mathcal{T} \to \mathbb{R}$ maximizes her expected utility \eqref{eq:util_interim_DM} given her posterior beliefs $\tilde{\theta}(s,\tau)$; and
		\item the decision-maker's posterior beliefs $\tilde{\theta}(s,\tau): \mathbb{R}^K \times \mathcal{T} \to \varDelta(\mathbb{R}^K)$ are obtained via Bayes' rule  for all possible signal realizations $s$, given the researcher's choice $\tau$.\footnote{We pin down the off-equilibrium-path beliefs using Bayes' rule as well. Given the researcher's test allocation strategy $\tau$ is observable, and he has no private information when making that choice, this restriction appears reasonable. The same off-path beliefs would be uniquely pinned down by stronger concepts like Sequential Equilibrium, although there are known technical issues associated with finding Sequential Equilibria in games with infinitely many actions and signals \citep{myerson2020perfect}, which motivates us to use weak PBE as our equilibrium concept.}
	\end{enumerate}

	%=======================================
	\subsection{Preliminary analysis: belief updating}\label{sec:prelim}
	%=======================================

	Given test allocation $\tau=(\tau_1,\hdots,\tau_K)$ and signal realizations $s=(s_1,\hdots,s_K)$, let $\tilde{\theta}_k(s_k,\tau_k)$ and $\tilde{b}^i(s,\tau)$ denote the posterior beliefs about attribute $\tilde{\theta}_k$ and about player $i$'s bliss point, respectively. For every $k$, we get	\begin{eqnarray}\label{eq:posterior_attributes} 
		\tilde{\theta}_k(s_k,\tau_k)
		\sim \mathcal{N} \left( \hat{\mu}(s_k,\tau_k), \hat{\Sigma}_k(\tau_k) \right),
	\end{eqnarray}
	where	\begin{eqnarray}\label{eq:posterior_mean_attributes} 
		\hat{\mu}_k(s_k,\tau_k) &=& \frac{1}{1+\tau_k\Sigma^0_k} \mu^0_k + \frac{\tau_k \Sigma^0_k}{1 + \tau_k\Sigma^0_k} s_k, 
		\\ 
		\label{eq:posterior_var_attributes}
		\hat{\Sigma}_k(\tau_k) &=& \frac{\Sigma^0_k}{1+\tau_k \Sigma^0_k}.
	\end{eqnarray}
	Since $b^i(\theta) = \sum_k \alpha^i_k \theta_k$ and attributes are independent, we get
	\begin{eqnarray}\label{eq:bliss_posterior}
		\tilde{b}^i(s,\tau) \sim \mathcal{N}\left( \sum_k \alpha^i_k \hat{\mu}_k(s_k,\tau_k), \hat{\sigma}^{2,i}(\tau) \right),
	\end{eqnarray}
	where
	\begin{eqnarray}\label{eq:bliss_posterior_var}
		\hat{\sigma}^{2,i}(\tau) = (\alpha^i_1)^2 \hat{\Sigma}_1(\tau_1) + \hdots + (\alpha^i_K)^2 \hat{\Sigma}_K(\tau_K).
	\end{eqnarray}
	Due to the properties of normal distribution, the variances $\hat{\Sigma}_k(\tau_k)$, and hence also the variance of $\tilde{b}^i(s,\tau)$,  do not depend on signal realizations $s$.

	%=======================================
	%=======================================
	\section{The single-player case}\label{sec:single_player}
	%=======================================
	%=======================================
	
	We begin by analyzing the scenario in which the decision-maker and the researcher share identical preferences, referred to as the single-player case. Here, we omit the players' superscripts $i =D,R$ and set 
    $\alpha := \alpha^{R} = \alpha^{D}$, $b(\theta) := b^{R}(\theta) = b^{D}(\theta)$, and so forth. For simplicity, we refer to this single player as the \emph{agent} in this section. In the learning stage, the agent allocates tests to learn about the attributes, and in the decision stage, after observing the signal realizations, the agent takes a decision. 
	
	The optimal decision strategy is straightforward.\footnote{We use ``optimal'' strategy to denote the equilibrium solution defined in Section \ref{sec:model} when the researcher's and decision-maker's weights coincide.} Given a test allocation $\tau$ and signal realizations $s$, the expected utility \eqref{eq:util_interim_DM} is maximized by
	\begin{eqnarray}\label{eq:optimal_decision}
		d^*(s,\tau) = \mathbb{E}[\tilde{b}(s,\tau)] = \alpha_1 \hat{\mu}_1(s_1,\tau_1) + \hdots + \alpha_K \hat{\mu}_K (s_K,\tau_K),
	\end{eqnarray}
	where $\hat{\mu}_k(s_k,\tau_k)$ for every $k=1,\hdots, K$ is given by equation \eqref{eq:posterior_mean_attributes}. Conditional on $\tau$, the optimal decision is ex ante distributed as
	\begin{eqnarray}\label{eq:distr_posterior_mean}
		\tilde{d}^\ast(\tau):= {d}^*(\tilde{s},\tau) \sim \mathcal{N} \left( b_0, \psi(\tau) \right),
	\end{eqnarray}
	where the variance is given by $\psi(\tau) := \sigma^{2}_0 - \hat{\sigma}^{2}(\tau)$.
	Anticipating the optimal decision strategy, the problem of choosing a test allocation can be reformulated as follows.
	\begin{lemma}\label{lemma:single_player_maximization}
		In the single-player case, the optimal test allocation is given by	\begin{eqnarray}\label{eq:single_player_maximization2}
			\arg \max_{\tau \in \mathcal{T}}\ V(\tau)= \arg \min_{\tau \in \mathcal{T}} \ \hat{\sigma}^2(\tau).
		\end{eqnarray}
	\end{lemma}

	Lemma \ref{lemma:single_player_maximization} establishes that the agent chooses $\tau$ to minimize the residual uncertainty $\hat{\sigma}^2$ regarding the bliss point $\tilde{b}$ (i.e., the residual uncertainty concerning the optimal decision) that remains at the posterior beliefs. 
    Our first main result below characterizes the test allocation that accomplishes this objective. For clarity, we present the characterization in the case of two attributes, with the general form for $K\geq 2$ available in the appendix.

    \begin{theorem}\label{thm:benchmark}
        Consider the single-player case and let $K=2$. Without loss of generality, let $\alpha_1 \Sigma^0_{1} \geq \alpha_2 \Sigma^0_{2}$.
		Then, the unique optimal test allocation $\tau^*=(\tau^*_1,\tau^*_2)$ is as follows: if $\alpha_1 = \alpha_2=0$, then $\tau^* = (0, 0)$; 
		otherwise,
		\begin{align*}
			\tau^*_1 = 
			\begin{cases}
				T & \text{ if }\ T\leq \bar{T} \\
				\bar{T} + \frac{\alpha_1}{\alpha_1+\alpha_2}  (T-\bar{T}) & \text{ if }\ T>\bar{T}
			\end{cases}
		\end{align*}
		and $\tau^*_2=T-\tau^*_1$, where $\bar{T} = \frac{\alpha_1 }{\alpha_2\Sigma^0_2} - \frac{1}{\Sigma^0_1}$.
    \end{theorem}

\begin{comment}
 \begin{theorem}\label{thm:benchmark}
		In the single-player case, without loss of generality, let $\alpha_1 \Sigma^0_{1} \geq \hdots \geq  \alpha_K \Sigma^0_{K}$.
		Then, the unique optimal test allocation $\tau^*$ is as follows:
		if $\alpha_k = 0$ for all $k$, then $\tau^* = (0, \ldots, 0)$;
		otherwise,\footnote{We apply the convention $\sum_{j=k}^{l} x_j = 0$ when $l<k$.}
		\begin{align*}
			\tau_k^* = \left( T - \bar{T}_{J} \right) \frac{\alpha_k}{\sum_{l=1}^{J} \alpha_l} + \sum_{j=k}^{J-1} \left( \left( \bar{T}_{j+1} - \bar{T}_{j} \right) \frac{\alpha_k}{\sum_{l=1}^{j} \alpha_l} \right)
		\end{align*}
		for all $k \in \{1, \ldots, J\}$ and $\tau^*_k = 0$ for all $k \in \{J+1, \ldots, K\}$,
		where the budget thresholds $\bar{T}_1 = 0 \leq \bar{T}_2 \leq \ldots \leq \bar{T}_{K+1}$ are given by
		\begin{align*}
			\bar{T}_k &:= \frac{\sum_{l=1}^{k-1}\alpha_l}{\alpha_{k}\Sigma_{k}^0} - \sum_{l=1}^{k-1} \frac{1}{\Sigma_l^{0}}
		\end{align*} 
		for all $k\in \{1, \ldots, K\}$, $\bar{T}_{K+1} := +\infty$,
		and $J \in \{1, \ldots, K+1\}$ is such that $T \in \left[\right. \bar{T}_{J}, \bar{T}_{J+1} \left.\right)$.
	\end{theorem}
 \end{comment}
 
%	The proof involves two steps. First, we apply the Karush-Kuhn-Tucker conditions to establish the necessary and sufficient conditions for a unique solution. Then, we use the fundamental theorem of calculus to determine the multipliers and confirm the solution.

    Theorem \ref{thm:benchmark} shows that
    the agent's optimal test allocation depends on the weights $\alpha_k$ solely through their \emph{ratios}. The ex ante marginal value of learning about attribute $\tilde{\theta}_k$ is $\frac{\partial V(\tau)}{\partial \tau_k}|_{\tau=(0,\hdots,0)} = \left(\alpha_k \Sigma^0_k\right)^2$. The term $\alpha_k \Sigma^0_k$ reflects  both the attribute's relevance to the agent (weight $\alpha_k$) and its prior uncertainty (prior variance $\Sigma_{k}^0$). The budget size $T$ determines whether the agent learns exclusively about the attribute with the higher ex ante marginal value of learning $(T\leq \bar{T})$ or if he learns about both attributes $(T>\bar{T})$. Note that allocating $\tau_k$ tests to learn about attribute $\tilde{\theta}_k$ decreases its interim marginal value of learning, which is given by $\frac{\partial V(\tau)}{\partial \tau_k} =$ $\left(\frac{\alpha_k \Sigma_k^0}{1 + \tau_k \Sigma_k^0}\right)^2$.  %$\left(\frac{\alpha_k \Sigma_k^0}{1 + \tau_k \Sigma_k^0}\right)^2$. 
    Evaluating the interim marginal value of learning for attribute $\tilde{\theta}_1$ at $\tau_1 = \bar{T}$, we obtain $(\alpha_2 \Sigma_2^0)^2$, which matches the ex ante marginal value of learning  about attribute $\tilde{\theta}_2$. Hence, with $\bar{T}$ tests allocated to attribute $\tilde{\theta}_1$ and none to $\tilde{\theta}_2$, the marginal values of learning about the two attributes are equalized. Therefore, any test budget $T\leq \bar{T}$ is allocated exclusively to attribute $\tilde{\theta}_1$, and the agent learns about attribute $\tilde{\theta}_2$ only when more tests are available. These additional tests are allocated between $\tilde{\theta}_1$ and $\tilde{\theta}_2$ in proportion to weights $\alpha_1$ and $\alpha_2$, keeping the interim marginal values of learning about the two attributes equal.\footnote{
    %    	As noted in the literature review, \citet{liang_dynamically_2022} investigate an agent's optimal dynamic information acquisition problem in a framework similar to ours but which allows for correlated attributes. Our single-agent solution can also be obtained using their results with suitable modifications.
   We can connect this result in two ways to \citet{liang_dynamically_2022}. First, Theorem \ref{thm:benchmark}---which abstracts from strategic motives---can be obtained from their results by considering a discrete time version of their model without correlation. Second, it thus immediately follows from their results that our 
    Theorem \ref{thm:benchmark} above (and also Theorem \ref{thm:strategic_solution} below) are robust to allowing for dynamic learning----as in their paper---and do not hinge on the static choice of the test allocation.
    }
     
	Note that the agent exhausts the entire budget $T$, except in the degenerate case where neither attribute holds any importance $(\alpha_1=\alpha_2=0)$. In this scenario, consistent with our equilibrium selection rule, the agent abstains from learning.

	\section{Strategic players}\label{sec:strategic_players}
	%=======================================
	%=======================================
	
	We now turn to the general case where the players have different preferences, leading to a conflict of interest. We proceed by backward induction, first solving the decision-maker's problem and then addressing the researcher's learning problem.
		
	The decision-maker's problem is analogous to the decision stage in the single-player problem with $\alpha=\alpha^{D}$. In particular, maximizing the expected utility \eqref{eq:util_interim_DM} over $d$ given some test allocation $\tau$ yields the equilibrium decision strategy $d^*(s,\tau) = \mathbb{E}[\tilde{b}^{D}(s,\tau)] = \alpha^{D}_1 \hat{\mu}(s_1,\tau_1) + \hdots + \alpha^{D}_K \hat{\mu}(s_K,\tau_K)$. Thus, for a given $\tau$, the equilibrium decision is ex ante distributed as $\tilde{d}^\ast(\tau):= d^*(\tilde{s},\tau)\sim N(b_0^{D}, \psi^{D}(\tau))$, where $\psi^{D}(\tau) := \sigma^{2,D}_0 - \hat{\sigma}^{2,D}(\tau)$	and $\hat{\sigma}^{2,D}(\tau)$ is given by \eqref{eq:bliss_posterior_var}.
	
	Turning to the researcher's learning problem, Lemma \ref{lemma:single_player_maximization} shows that if the researcher's optimal decisions were implemented, he would choose $\tau$ to minimize $\hat{\sigma}^{2,R}(\tau)$, the posterior uncertainty about his optimal decision. In a strategic setting, however, the researcher  prioritizes attributes where his weights align closely with those of the decision-maker. %he and the decision-maker assign similar weights. 
    The following lemma captures the essence of the researcher's maximization problem.	

	\begin{lemma}\label{lemma:strategic_player_maximization}
		The researcher's equilibrium test allocation solves
		\begin{eqnarray}\label{eq:strategic_player_maximization}
			\max_{\tau \in \mathcal{T}} V^R(\tau) = \max_{\tau \in \mathcal{T}}
			\Bigg\{ \underbrace{2 \sum_k \alpha_k^{D} \alpha_k^R \Sigma^0_k \hat{\Sigma}_k(\tau_k) \tau_k}_{2\cov(\tilde{d}^*(\tau),\tilde{b}^R)}
			-\underbrace{ \left(\sigma^{2,D}_0 - \hat{\sigma}^{2,D}(\tau) \right)}_{\psi^{D}(\tau)} + V^R(\emptyset) \Bigg\},
		\end{eqnarray}
		where $\hat{\Sigma}_k(\tau_k)$ for every $k=1,\hdots,K$ are given by equation \eqref{eq:posterior_var_attributes}, $\hat{\sigma}^{2,D}(\tau)$ by equation \eqref{eq:bliss_posterior_var}, and $V^R(\emptyset)$ is the researcher's expected utility under no learning (constant in $\tau$).
	\end{lemma}
	Lemma \ref{lemma:strategic_player_maximization} shows that the researcher pursues two interdependent objectives. First, he seeks to align the final decision with his own bliss point: $\uparrow \cov \left( \tilde{d}^*(\tau),\tilde{b}^R \right)$. Second, he aims to prevent the decision-maker from overreacting to new information relative to the researcher's bliss point: $\downarrow \left(\psi^{D}(\tau)-\cov \left( \tilde{d}^*(\tau),\tilde{b}^R \right) \right)$. Together, these effects constitute the added value of learning.\footnote{
		The trade-off illustrated by Lemma~\ref{lemma:strategic_player_maximization} also arises in Theorem 2 in \citet{bardhi_attributes:_2024}. However, the specifics of $\psi^{D}$ and $\cov \left( \tilde{d}^*(\tau),\tilde{b}^{R} \right)$ differ between the two models due to the distinct ways in which uncertainty and player payoffs are modeled.}

	To concisely formulate our main result---the solution to the researcher's problem---we introduce a final piece of notation. For all $k=1,\hdots,K$, let
	\begin{equation} \label{eq:alpha_n_lambda}
		\begin{aligned} 
		\Delta_k &:= |\alpha_k^{R}-\alpha_k^{D}|,
		\\
		\lambda_k &:= 2\alpha_k^{R}\alpha_k^{D}-\left(\alpha^{D}_k\right)^2 =  (\alpha_k^{R})^2 - (\Delta_k)^2.
		\end{aligned}
	\end{equation}
	Hence, $\Delta_k$ is the distance between the players' weights on that attribute, and $\lambda_k$ is a parameter that reflects the aforementioned trade-off. %, discussed in more detail below.
	
	\begin{theorem}\label{thm:strategic_solution}
		The equilibrium test allocation $\tau^*$ coincides with the solution of the single-player problem with \emph{auxiliary weights}
		\begin{eqnarray}\label{eq:auxiliary_weights}
			\hat{\alpha}_k := \sqrt{ \max \{ \lambda_k, 0 \} }
			\quad \text{for all } k=1,\hdots,K.
		\end{eqnarray} 
	\end{theorem}
	
	%	The proof of the theorem unfolds in two steps. First, we compute the first-order conditions (FOC) of the researcher's problem in the strategic-player scenario. Second, we establish that these FOCs coincide with those in the single-player problem with auxiliary weights as in equation \eqref{eq:auxiliary_weights} and the two maximization problems yield the same solution.
		
	Theorem \ref{thm:strategic_solution} reveals a key insight: even when the researcher and the decision-maker have different weights, the researcher's equilibrium test allocation coincides with a solution to the single-player problem with some auxiliary weights. Thus, the effect of the strategic motive on the researcher's learning reduces to a shift in attribute weights. Therefore, by comparing these auxiliary weights to the researcher's original weights, we can pinpoint how misalignment drives adjustments in the researcher's learning strategy---insight that would be harder to achieve if the strategic solution deviated fundamentally from the single-player case.
	
	From the researcher's perspective, the ex ante marginal value of learning about attribute $\tilde{\theta}_k$ is proportional to the parameter $\lambda_k$, which---unlike in the single-player case---can be negative. Theorem \ref{thm:strategic_solution} explains the link between the auxiliary weights $\hat{\alpha}_k$ and the parameters $\lambda_k$. Note that $\lambda_k$ (and thus $\hat{\alpha}_k$) is decreasing in the misalignment between the players' weights $\Delta_k$. When $\Delta_k>0$, the researcher views the decision-maker's reaction to information about $\tilde{\theta}_k$ as either excessive ($\alpha^{D}_k> \alpha^R_k$) or insufficient $(\alpha^{D}_k<\alpha^R_k)$, thereby reducing the researcher's incentives to learn about that attribute. If the reaction is too strong ($\alpha_k^{D} \geq 2 \alpha_k^R$) or too weak ($\alpha^{D}_k=0$), the researcher avoids learning about that attribute entirely. If $\lambda_k \leq 0$ for all $k$, the researcher abstains from learning, as any information would lead to undesirable overreaction (or no reaction) by the decision-maker, making the status quo decision the researcher's preferred outcome.\footnote{This observation is related to the fundamental lesson from strategic communication models: if the misalignment between the sender and the receiver becomes too large, then no information is transmitted (c.f. \citealp{crawford_strategic_1982}; or \citealp[Section VA]{kamenica_bayesian_2011}). Our results add nuance to this observation: when misalignment related to one attribute grows large, no information about that attribute is transmitted, but more information about other attributes may be provided. Only if misalignment is large for all attributes does information transmission break down completely.}

	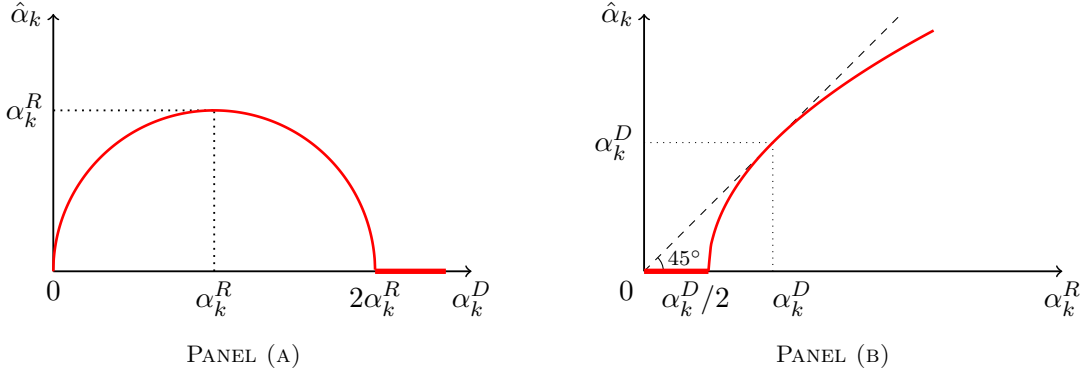
\begin{figure}
		\begin{center}
  			\caption{The researcher's solution: Auxiliary single-player weights}
			\label{fig:alpha_tildes}
			\subfloat[][Panel (a)]{
				\begin{tikzpicture}[scale=0.85]
					% Axes
					\draw (0,0) node[below]{$0$};
					\draw (5,0) node[below]{$2\alpha^R_k$};
					\draw[->, line width=0.7pt] (0,0) -- (0,4) node[left] {$\hat{\alpha}_k$};
					\draw[->, line width=0.7pt] (-0,0) -- (6.5,0) node[below] {$\alpha_k^{D}$};
					
					% Lines
					\draw[dotted, line width=0.7pt]  (2.5,0) -- (2.5,2.5);
					\draw (2.5,0)
					node[below]{$\alpha_k^R$};
					\draw (0,2.5)
					node[left]{$\alpha^R_k$};
					\draw[dotted, line width=0.7pt] (0,2.5) -- (2.5,2.5);
					%	\draw[dashed] (0,0) -- (4,4);
					%\draw (0.3,0) arc (0:45:0.3); %node[right]{\scriptsize$45^{\circ}$};
					\draw[smooth, samples=1000, domain=0:5, line width=1pt, red] plot({\x}, {sqrt(\x*(5-\x))});
					\draw[line width=2pt, red] (5,0) -- (6.1,0);
			\end{tikzpicture}}
			\hfill
			\subfloat[][Panel (b)]{
				\begin{tikzpicture}[scale=0.85]
					% Axes
					\draw (0,0) node[below left]{$0$};
					\draw[->, line width=0.7pt] (0,0) -- (0,4) node[left] {$\hat{\alpha}_k$};
					\draw[->, line width=0.7pt] (0,0) -- (6.5,0) node[below] {$\alpha^{R}_k$};
					
					% Lines
					%\draw[dotted] (1,4) -- (1,0) 
					\draw (0.8,0) node[below]{$\alpha^{D}_k/2$};
					\draw[dotted] (2,2) -- (2,0);
                        \draw (2.3, 0) node[below]{$\alpha^{D}_k$};
					\draw[dotted] (2,2) -- (0,2) node[left]{$\alpha^{D}_k$};
					\draw[dashed] (0,0) -- (4,4);
					\draw (0.3,0) arc (0:45:0.3) node[right]{\scriptsize$45^{\circ}$};
					
					\draw[smooth, samples=100, domain=1:4.5, line width=1pt, red] plot({\x}, {sqrt(2*(2*\x-2))});
					\draw[line width=2pt, red] (0,0) -- (1,0);
				\end{tikzpicture}
			}
		\end{center}
		\footnotesize \emph{Notes:} The figures depict $\hat{\alpha}_k$ as a function of the decision-maker's weight, $\alpha^{D}_k$, in panel (a), and as a function of the researcher's weight, $\alpha^{R}_k$, in panel (b), keeping the weight of the other player constant.
	\end{figure}

	When players disagree on the weight of an attribute, the auxiliary weight differs from the weight of either player. Figure \ref{fig:alpha_tildes} illustrates how $\hat{\alpha}_k$ depends on $\alpha^{R}_k$ and $\alpha^{D}_k$, keeping the weight of the other player fixed. Panel (a) shows that $\hat{\alpha}_k$ always satisfies $\hat{\alpha}_k \leq \alpha^{R}_k$: the misalignment effectively reduces the importance the researcher assigns to the attribute. Nevertheless, the extent and direction of any distortion in the equilibrium test allocation---compared to the researcher's non-strategic optimum---then depend    
    on how the \emph{ratios} of the auxiliary weights differ from the ratios of the researcher's weights.  

    %========================================
    \subsection{Correlated attributes} \label{sec:correlation}

 In the model, we assume that the attributes are independent. In this subsection, we discuss how the results change when the attributes are jointly normal but correlated---that is, when the variance–covariance matrix is no longer diagonal.
    
Theorem \ref{thm:strategic_solution} characterizes the main, \textit{direct} effects of misalignment on multi-attribute learning. Introducing correlation into the model adds additional, \textit{indirect} effects.
The marginal value of learning about any given attribute now consists of both a direct component---reflecting the value of information about that attribute itself---and indirect effects---capturing the informational spillovers from correlated attributes. When deciding which attributes (and how much) to learn about, the researcher bases his choice on a specific \textit{projection value}, which aggregates these direct and indirect effects. In this case, misalignment can give rise to effects that cannot arise in single-agent framework, and thus the equivalence established in Theorem \ref{thm:strategic_solution} generally breaks down.\footnote{For instance, imagine a researcher who, in the absence of correlation, never wants to learn about attribute $k$ (i.e., $\lambda_k<0$), but who derives positive marginal value of learning about another attribute $j\neq k$ (i.e., $\lambda_j>0$). If these two attributes are correlated, learning about $j$ provides a direct benefit to the researcher but also an indirect cost, since the decision-maker—through correlation—also gains information about $k$. Such a mechanism cannot arise in a single-agent setting, where learning about any attribute (in the absence of correlation) always has a \textit{positive} marginal value.}

Nevertheless, as long as the correlation is not too strong, focusing on the direct effects provides a useful approximation of the researcher's incentives. Furthermore, the qualitative insights from Theorem \ref{thm:strategic_solution} remain robust even in the presence of correlation: the marginal value of learning about an attribute declines under misalignment about that attribute, as the researcher perceives the information about that attribute to be either insufficiently utilized or excessively acted upon. Consequently, the researcher may choose to reduce learning about, or even completely abstain from learning about, certain attributes when the perceived overreaction is sufficiently strong.

	%==========================
    \subsection{Equivalent payoff specifications} \label{sec:equivalent_models} 

    This section introduces alternative payoff structures that result in the same equilibrium test allocation as the baseline model. The objective of presenting these frameworks is to offer alternatives that may be better tailored to particular economic applications, thus demonstrating the adaptability of our model and the analytical tools detailed in Sections~\ref{sec:single_player} and \ref{sec:strategic_players}.

	\paragraph{Baseline model.} Introduced in Section~\ref{sec:model}, the baseline model features a decision-maker taking a single decision, $d\in \mathbb{R}$. Given a decision $d$ and a realized state $\theta=(\theta_1,\hdots,\theta_K)$, the utility of player $i=R,D$ is
	\begin{eqnarray} \label{eq:util_spec_baseline}
		u^i(d,\theta) = -(d-b^i(\theta))^2 = - \left( d - \sum_k \alpha^i_k \theta_k  \right)^2.
	\end{eqnarray}
	
	\paragraph{Framework A.} In this framework, the decision-maker also takes a single decision, $d\in \mathbb{R}$. Given a decision $d$ and a realized state $\theta=(\theta_1,\hdots,\theta_K)$, the utility of player $i=R,D$ is
	\begin{eqnarray} \label{eq:util_spec_weighted}
		u^i_A(d,\theta) =  - \sum_k \alpha^i_k \left( d - \theta_k  \right)^2,
	\end{eqnarray}
	where $\sum_k \alpha^i_k = 1$ for both players.
	
	\paragraph{Framework B.} This framework features the decision-maker simultaneously taking $K$ distinct decisions $d_1,\hdots,d_K \in \mathbb{R}$.  Given decisions $d=(d_1,\hdots,d_K)$ and a realized state $\theta=(\theta_1,\hdots,\theta_K)$, the utility of player $i=R,D$ is
	\begin{eqnarray} \label{eq:util_spec_multiaction}
		u^i_B(d,\theta) = -\sum_k (d_k -\alpha^i_k \theta_k)^2.
	\end{eqnarray}

	In all three models: (i) players aim to minimize a loss given by a quadratic distance, and (ii) the weights $\alpha^i$ determine the relative importance of different attributes. Beyond these similarities, the frameworks differ in how the players aggregate losses across attributes and decisions. Despite these differences, the following proposition (proved in Online Appendix \ref{OA:frameworks}) shows that Frameworks A and B are equivalent to the baseline model in terms of the equilibrium test allocation.
	\begin{proposition}\label{prop:equivalent_models}
		Given the decision-maker's weights $\alpha^{D}$, the researcher's weights $\alpha^R$, the test budget $T$, and the prior distribution of the state $\tilde{\theta}$, the researcher's equilibrium test allocations in the baseline model, Framework A, and Framework B are identical.
	\end{proposition}
	The displayed flexibility allows us to study a rich set of economic problems. For instance, we use Framework~A  to study discrimination (Section~\ref{sec:discrimination}) and media polarization (Online Appendix~\ref{sec:media}). To illustrate an application of Framework~B, consider
    a portfolio choice problem, where an investor makes decisions $d_k$ on how much to invest in each asset $k=1,\hdots,K$. Here, $\tilde{\theta}_k$ represents the uncertain future return on asset $k$. The investor evaluates the future returns objectively ($\alpha^{D}_k = 1$ for all $k$), while an advisor may have an incentive to steer the investor towards some assets and away from others: $\alpha^{R}_k \neq 1$ for some $k$. 
	%==========================
	%==========================
	\section{Applications}\label{sec:applications}
	%==========================
	%==========================
		
	In this section, we illustrate how our framework can be used to study preference misalignment in two applications, each with a tailored parametrization of the misalignment. First, we adopt an organizational setting to showcase how biases across hierarchical layers impact learning in organizations. Next, we analyze a policy-making model with a discriminatory politician and explore how introducing independent researchers to inform policy decisions can mitigate negative effects of discrimination on welfare and inequality.  
 
	In the applications, we adopt the simplest specification, a two-attribute case, which is already sufficiently rich to demonstrate the tension between the players arising in the multi-attribute context. Throughout, we assume that both weights of the decision-maker are strictly positive, $\alpha^{D}_k>0$ for $k=1,2$.

	%==========================
	\subsection{Diversity in organizations}\label{sec:organizations}
	%==========================
		
	In this application, we consider a stylized   model of an organization, comprising a manager (decision-maker) and an analyst (researcher). This setup reflects an organizational hierarchy where the manager holds authority over the analyst and takes decisions based on the analyst's information. Given a diverse pool of analysts with varying weights, the manager selects one to provide information. The key question is which analyst the manager will choose.

    First, it is immediate that the manager's ideal choice would be an analyst whose weights match hers. However, our focus is on identifying the type of analyst the manager would prefer when, due to an inherent hierarchical structure, perfect alignment is unachievable. To analyze this question, we begin by establishing a suitable parametrization of the misalignment between the players' weights. 

 \begin{figure}
		\begin{center}
  			\caption{The bias decomposition of an analyst}
			\label{fig:analyst_bias_decomposition}
			\begin{tikzpicture}[scale=1.2]
				% Area colors
				%				\filldraw[color=cyan!10!white] (1,0.5) -- (4,4) -- (5,4) |- (1,0);
				%				\filldraw[color=red!10!white] (1,0.5) -- (1.5,4) -| (0,0.5);
				%				\filldraw[color=magenta!10!white] (0,0) |- (1,0.5) |- (0,0);
				
				% Axes
				\draw[->, line width=0.7pt] (0,0) -- (0,4.2) node[left] {$\alpha^{R}_2$};
				\draw[->, line width=0.7pt] (0,0) -- (5.2,0) node[below] {$\alpha^{R}_1$};
				
				%				% Thresholds
				%				\draw[color=black!50!white, dashed] (1,0.5) -- (1,0) node[below]{$\frac{\alpha^{D}_1}{2}$};
				%				\draw[color=black!50!white, dashed] (1,0.5) -- (0,0.5) node[left]{$\frac{\alpha^{D}_2}{2}$};
				
				%				\draw[color=black!30!white, line width=1pt] (1,0.5) -- (4,4); 
				%				\draw[color=black!30!white, line width=1pt] (1,0.5) -- (1.5,4);
				
				% Bias lines
				\draw[->] (0,0) -- (5,2.5);
				\draw[] (5,2.6) node[above,left]{$\beta$};
				\draw[->] (2.5,0) -- (0.5,4) node[right]{$\gamma$};
				\draw[dashed] (0,1) -- (5,3.5);
				\filldraw[color=black] (1.6,1.8) circle(0.03);
				\draw[] (1.4,1.8) node[above,left]{$\gamma>0$};
				\filldraw[color=red] (2,1) circle(0.05);
				\draw[color=red] (1.8,1.4) node[above,right]{$\alpha^{D}$};
				
				%				% Area labels
				%				\draw (4,1) node[text=gray]{$L_1$};
				%				\draw (2/3,3) node[text=gray]{$L_2$};
				%				\draw (2.8,3.5) node[text=gray]{$L_{12}$};
				%				\draw (0.5,0.2) node[text=gray]{$L_0$};
			\end{tikzpicture}
		\end{center}
		\footnotesize \emph{Notes:} The dashed line represents different $\alpha^{R}$ values where $\gamma$ is fixed at a positive value and various values of $\beta$ are considered.
	\end{figure}
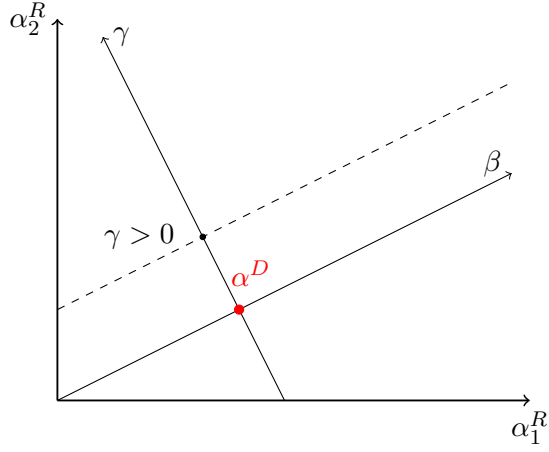
		
	We fix the manager's weight vector, $\alpha^{D}=(\alpha^{D}_1,\alpha^{D}_2)$, and define an orthogonal vector $\bar{\alpha}^{D} := (-\alpha_2^{D},\alpha_1^{D})$.
	Next, we express the analyst's weight vector $\alpha^R$ as a linear combination of the manager's weight vector and its orthogonal counterpart as 
	\begin{eqnarray}\label{eq:decomposition_binary}
		\alpha^{R}(\beta,\gamma) := \beta \alpha^{D} + \gamma \bar{\alpha}^{D}
	\end{eqnarray}
	for some $\beta \in \mathbb{R}_+$ and $\gamma \in \Gamma(\beta) := \left[ -\beta \frac{\alpha_2^{D}}{\alpha_1^{D}}, \beta \frac{\alpha_1^{D}}{\alpha_2^{D}} \right]$.\footnote{
		The constraint $\gamma \in \Gamma(\beta)$ follows from the requirement that $\alpha_k^{R}\geq 0$ for both $k=1,2$. As noted in Footnote~\ref{foot:alpha_positive}, this restriction is without loss and only for ease of exposition.}
	Here, $\beta$ represents the absolute bias, referred to as \emph{sensitivity} (to new information), and $\gamma$ represents the relative bias, referred to as \emph{distortion}. We say that \emph{the analyst becomes more sensitive} when $\beta$ increases.\footnote{Note that when $\beta=1$, the analyst has the same sensitivity as the manager.} We further say \emph{the analyst becomes more distorted} when $\gamma\geq 0$ increases or $\gamma\leq 0$ decreases. Specifically, when $\gamma >0$ ($\gamma<0$), the analyst exhibits bias towards attribute $\tilde{\theta}_2$ (attribute $\tilde{\theta}_1$). When $\gamma=0$, we call the analyst \emph{undistorted}. Figure \ref{fig:analyst_bias_decomposition} visualizes this bias decomposition.\footnote{
		Alternative decompositions of the analyst's weights in terms of ``absolute" and ``relative" bias exist. For instance, one could present both $\alpha^{R}$ and $\alpha^{D}$ in polar coordinates and let $\beta$ be the difference in distances from the origin and $\gamma$ be the difference in angles. The insights from our Propositions \ref{prop:CS_manager_gamma} and \ref{prop:CS_manager_beta} would extend to such alternative representations. 
		The same applies to settings with more than two attributes, where $\gamma$ could be a multiplier on any vector orthogonal to $\alpha^D$.
		}  
    In the context of our organizational application, we posit that analysts who are separated by more layers of hierarchy from the manager are less sensitive.

	First, we examine how the manager's payoff changes with the analyst's distortion $\gamma$. Proposition \ref{prop:CS_manager_gamma} below shows that the manager is worse off as the analyst becomes more distorted, provided the analyst is sufficiently sensitive (high $\beta$). However, the converse is true when the analyst is too insensitive (low $\beta$).
	
	\begin{figure}
		\centering
  			\caption{Equilibrium test allocation}			\label{fig:organization_distortion}
		\begin{subfigure}[t]{0.48\textwidth}
        \centering
        \caption{Panel (a)}
        \begin{tikzpicture}[scale=1.45]
            % Axes
            \draw[->, thick] (0,0) -- (0,3.5) node[left] {$\alpha^R_2$};
            \draw[->, thick] (0,0) -- (4,0) node[below] {$\alpha^R_1$};
            
            % Thresholds
            \draw (1.5,0.1) -- (1.5,-0.1) node[below]{$\frac{\alpha^{D}_1}{2}$};
            \draw (-0.1,0.75) -- (0.1,0.75) node[left]{$\frac{\alpha^{D}_2}{2}\ \ $};

            \draw[black, dashed, opacity=0.5] (1.5,0.0) -- (1.5,0.75);
            \draw[black, dashed, opacity=0.5] (0,0.75) -- (1.5,0.75);

            % Emphasis lines
            \draw[thick, red] (0,0) -- (1.5,0.75);
            \draw[->, thick, blue] (1.5,0.75) -- (4,2);
            \node[below right] at (4,2) {$\beta$};
            
            % Curved brackets with labels, adjusted positioning and line breaks
            \draw[decorate, decoration={brace, amplitude=13pt}, red] (0,0) -- (1.5,0.75) node[midway, yshift=27pt, red, font=\footnotesize, align=center] {no testing\\ for $\beta\leq 1/2$};
            \draw[decorate, decoration={brace, amplitude=13pt}, blue] (1.5,0.75) -- (3.9,1.95) node[midway, yshift=27pt, xshift=-10pt, blue, font=\footnotesize, align=center] {manager's first\\  best for $\beta>1/2$};

    % Points
            \filldraw (3,1.5) circle(1.5pt);
            \node[below right] at (3,1.7) {$\alpha^{D}$};
                   \end{tikzpicture}
    \end{subfigure}
    \hfill
    \begin{subfigure}[t]{0.48\textwidth}
        \centering
        \caption{Panel (b)}
        \begin{tikzpicture}[scale=1.45]
            % Area colors
            \filldraw[color=blue!15!white] (1.5,0.75) -- (4,1.52) -- (4,0) -- (1.5,0);  % L_1
            \filldraw[color=yellow!30!white] (1.5,0.75) -- (1.98,3.5) -| (0,0.75);       % L_2
            \filldraw[color=green!20!white] (1.5,0.75) -- (4,1.52) -- (4,3.5) -- (1.98,3.5);  % L_{12}
            %\fill[left color=yellow!30!white, right color=blue!15!white, shading angle=45] (1.5,0.75) -- (4,1.52) -- (4,3.5) -- (1.98,3.5);
            \filldraw[color=gray!30!white] (0,0) -- (0,0.75) -- (1.5,0.75) -- (1.5,0);  % L_0
            
            % Axes
            \draw[->, line width=0.7pt] (0,0) -- (0,3.5) node[left] {$\alpha^R_2$};
            \draw[->, line width=0.7pt] (0,0) -- (4,0) node[below] {$\alpha^R_1$};
            
            % Thresholds
            \draw[black] (1.5,0.04) -- (1.5,-0.04) node[below]{$\frac{\alpha^{D}_1}{2}$};
            \draw[black] (-0.04,0.75) -- (0.04,0.75) node[left]{$\frac{\alpha^{D}_2}{2}\ $};
            
            \draw[color=black!80!white, line width=0.3pt] (1.5,0.75) -- (4,1.52); 
            \draw[color=black!80!white, line width=0.3pt] (1.5,0.75) -- (1.98,3.5);
            \draw[color=black!80!white, line width=0.3pt] (1.5,0.75) -- (1.5,0);
            \draw[color=black!80!white, line width=0.3pt] (1.5,0.75) -- (0,0.75);
            
            % Bias lines
            \draw[-,dashed,opacity=0.5] (0,0) -- (4,2);
            \filldraw[color=black] (3,1.5) circle(0.05);
            \draw (3,1.8) node{$\alpha^{D}$};
            
            % Area labels
            \draw (3.7,0.2) node[text=black,font=\footnotesize]{$L_1$};
            \draw (0.9,3.2) node[text=black,font=\footnotesize]{$L_2$};
            \draw (3.7,3.2) node[text=black,font=\footnotesize]{$L_{12}$};
            \draw (0.9,0.2) node[text=black,font=\footnotesize]{$L_0$};
            
            % Gamma vectors
            \draw[-> ,thick,black] (2.4,1.2) -- (1.25,3.5);
            \draw (1.5,3.4) node{$\mathbf{\gamma}$};
            \draw (2.5,2.3) node{\small{$\tau_1^* \downarrow, \tau_2^* \uparrow$}};
            \draw[-> , thick, black] (4/5,2/5) -- (0,2) node[right]{$\gamma$};
        \end{tikzpicture}
    \end{subfigure}
	
		\footnotesize \flushleft \emph{Notes:} Panel (a) depicts the undistorted analyst's equilibrium test allocation:  %based on sensitivity $\beta$:
        no testing if $\beta\leq 1/2$, and manager's first best if $\beta >1/2$. Panel (b) details the equilibrium test allocation for different $(\beta,\gamma)$ pairs: the analyst learns (i) only about attribute $\tilde{\theta}_1$ if $\alpha^R \in L_1$; (ii) only about attribute $\tilde{\theta}_2$ if $\alpha^R\in L_2$; (iii) about both attributes if $\alpha^R\in L_{12}$; (iv) about no attribute if $\alpha^R\in L_0$. 
        The black thick lines captures analysts with increasing distortion towards attribute $\tilde{\theta}_2$ from an undistorted analyst (dashed line) for fixed sensitivity at $\beta<1/2$ (line closer to the origin) and $\beta>1/2$ (line further away from the origin).
	\end{figure}
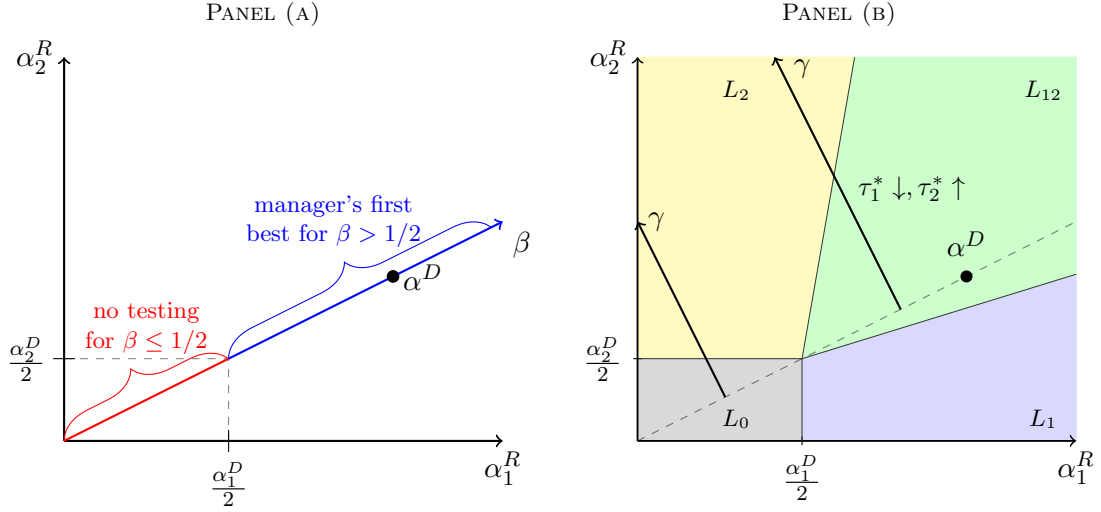
	
	\begin{proposition}\label{prop:CS_manager_gamma} Fix the analyst's level of sensitivity $\beta \in \mathbb{R}_+$. \nopagebreak
	\begin{enumerate}[label=(\roman{enumi})]
		\item If $\beta \leq 1/2$ (the analyst is too insensitive), then the manager's expected equilibrium payoff weakly \emph{increases} as the analyst becomes more distorted. 
		\item If $\beta > 1/2$ (the analyst is sufficiently sensitive), then the manager's expected equilibrium payoff weakly \emph{decreases} as the analyst becomes more distorted.
	\end{enumerate}
	\end{proposition}

	The manager and an undistorted analyst share the same weight ratios, so the analyst's equilibrium learning strategy aligns with the manager's first-best solution, provided the analyst's sensitivity is high enough to motivate learning ($\beta>1/2$). In this case, if the analyst's distortion increases, his equilibrium learning strategy deviates further from the manager's first best, making the manager worse off.
    Conversely, when $\beta \leq 1/2$, an undistorted analyst abstains from learning, perceiving the manager's reaction to any information as excessive. Then, if the analyst's distortion increases, his preferences align more closely with the manager's on one particular attribute, potentially prompting him to learn about it. While the added distortion does increase misalignment on the other attribute, it does not negatively impact learning since the initial misalignment was already sufficient to deter the analyst from learning about it. Overall, since any acquired information is preferable to none, the manager ultimately benefits.
    
	The results, illustrated in Figure \ref{fig:organization_distortion}, highlight the importance of both diversity and uniformity within organizations. The State of the Global Workplace report \citep{gallup2025_state_global_workplace} documents a substantial engagement gap between managers and individual workers. One plausible interpretation is that lower engagement among individuals further down the hierarchy makes them less sensitive to new information. In organizations with many hierarchical layers---and therefore large differences in sensitivity---our findings suggest that having substantial diversity at the lower levels, reflected in relatively high degrees of distortion $|\gamma|$ relative to managers, can be beneficial. Conversely, in smaller, less hierarchical organizations, a more uniform workforce with minimal distortion is preferable, as variation in engagement (and sensitivity) between managers and analysts is less problematic in these settings.
 
    The above underscores the nuanced role of distortion in shaping the analyst's alignment with the manager's objectives. However, when we shift focus from distortion to sensitivity, the effects become more straightforward.
	\begin{proposition}\label{prop:CS_manager_beta}
		Fix the analyst's distortion $\gamma\in \mathbb{R}$. Then, the manager's expected equilibrium payoff is weakly increasing in the analyst's sensitivity $\beta$. 
	\end{proposition}
	
	The manager always prefers a more sensitive analyst, as this leads to an equilibrium test allocation that aligns more closely with her first-best solution. In other words, the manager seeks highly engaged employees who are sensitive to new information relevant to the organization. To prove this result, we show that for any distortion $\gamma$, the equilibrium test allocation $\tau^*$ increasingly aligns with the manager's first best as $\beta$ increases, effectively neutralizing the impact of distortion.

	%==========================
	\subsection{Discrimination, welfare and inequality}\label{sec:discrimination}
	%==========================
	
	In this application, we explore a scenario where a politician decides on a policy affecting two social groups and favors one group. We investigate how appointing an advisor with different preferences, who strategically curates the information provided to the politician, can mitigate the negative impact of this type of discrimination on utilitarian welfare and inequality.

	Suppose there are two social groups $k=1,2$. Each group $k$ has an unknown optimal policy $\tilde{\theta}_k \sim \mathcal{N}(0,1)$, and so the two groups are a priori identical.	A politician (decision-maker) decides on a common policy, $d\in \mathbb{R}$. The utility of group $k$ is given by $u_k(d,\theta_k) = - (d-\theta_k)^2$.	A budget $T=1$ of tests is available to inquire about the optimal policies of both groups. An advisor (researcher) chooses a test allocation $\tau=(\tau_1,\tau_2) \in \mathcal{T}$. The learning process and notation are the same as in the baseline model.

	The politician and the advisor are utilitarian with particular weights. The politician cares about both groups but favors group $k=1$, her payoff being
	\begin{eqnarray*}
		u^{D}(d, \theta_1, \theta_2; \delta) = \underbrace{\frac{1+\delta}{2}}_{\alpha^{D}_1(\delta):=} u_1(d,\theta_1) + \underbrace{\frac{1-\delta}{2}}_{\alpha^{D}_2(\delta):=}u_2(d,\theta_2)
	\end{eqnarray*}
    for a given level of discrimination $\delta \in (0,1)$. On the other hand, the advisor (possibly) favors group $k=2$, his payoff being
	\begin{eqnarray*}
		u^{R}(d,\theta_1,\theta_2;p) = \underbrace{\frac{1-p}{2}}_{\alpha^{R}_1(p):=}u_1(d,\theta_1) + \underbrace{\frac{1+p}{2}}_{\alpha^{R}_2(p):=} u_2(d,\theta_2)
	\end{eqnarray*}
    for a given level of partiality $p\in [0,1)$. When $p=0$, indicating the advisor cares equally about both groups, we call the advisor \emph{impartial}.\footnote{For expositional clarity, we focus on $p\in[0,1]$. However, all results continue to hold for $p\in [-1,1]$, i.e., also when including advisors who favor the same group as the politician.}

	Our focus in this application is on \emph{welfare} $W(p, \delta)$, defined as the sum of ex ante equilibrium expected payoffs of both groups (and is hence utilitarian with equal weights), and \emph{inequality} $I(p, \delta)$, defined as their difference (and is hence an egalitarian measure):
	\begin{eqnarray} \label{eq:def_welfare}
		W(p, \delta) &:= & \mathbb{E}\left[u_1\left(\tilde{d}^{\ast}(\tau^*),\tilde{\theta}_1\right) \right]+\mathbb{E}\left[u_2\left(\tilde{d}^{\ast}(\tau^*),\tilde{\theta}_2\right) \right],
		\\ \label{eq:def_inequality}
		I(p, \delta) &:= & \left|\mathbb{E}\left[u_1\left(\tilde{d}^{\ast}(\tau^*),\tilde{\theta}_1\right) \right]-\mathbb{E}\left[u_2\left(\tilde{d}^{\ast}(\tau^*),\tilde{\theta}_2\right) \right]\right|.
		\label{eq:inequality}
	\end{eqnarray}
	In the expressions above, the right-hand side depends on $p$ and $\delta$ via the advisor's choice of test allocation $\tau^*$ and the politician's policy choice $\tilde{d}^{*}(\tau^*)$, which can be more fully described as $\tilde{d}^{*} \left( \tau^*(p,\delta), \delta \right)$.
	When $I(p, \delta)=0$, we say there is \emph{equality} in equilibrium.\footnote{This model fits Framework A in Section~\ref{sec:equivalent_models}. By Proposition~\ref{prop:equivalent_models}, Theorem~\ref{thm:strategic_solution} can be used to find the equilibrium test allocation.}

    As a benchmark, we consider the case of \emph{unchecked discrimination}, when the politician controls both learning and decision-making. Here, welfare and inequality are defined analogously to equations \eqref{eq:def_welfare} and \eqref{eq:def_inequality}, with the politician selecting the optimal test allocation instead of an advisor doing so. In this context, increased discrimination negatively impacts both welfare and inequality: welfare declines while inequality rises. We aim to understand how appointing an advisor can mitigate these adverse effects. We analyze a scenario where welfare and inequality are influenced solely through the advisor's role in the learning stage, while the politician retains full control over decision-making.
   
	First, we ask which advisor maximizes welfare. Since both welfare and the preferences of the impartial advisor are utilitarian with equal weights, it immediately follows that appointing an impartial advisor maximizes welfare, regardless of the politician's level of discrimination. In contrast, appointing a partial advisor or allowing unchecked discrimination yields lower welfare. The next lemma describes the (welfare-maximizing) learning strategy of the impartial advisor.

	\begin{lemma}\label{prop:discrimination_impartial}
    For every level of discrimination $\delta$, welfare is maximized by appointing the impartial advisor ($p=0$). Moreover, the impartial advisor's equilibrium test allocation $\tau^*(0,\delta)=(1/2,1/2)$ is \emph{independent} of $\delta$.
	\end{lemma}

	As the level of discrimination $\delta$ increases, two effects emerge. 
	From the advisor's perspective, information about group 1 is increasingly overemphasized in the politician's decision, while information about group 2 is increasingly underused. Each effect independently \emph{reduces} the advisor's incentive to learn about the respective group. However, for the impartial advisor, these effects are equal in magnitude and cancel each other out. This leads him to be unresponsive to changes in $\delta$ and always choose an equal test allocation. 
    Consequently, as $\delta$ increases, inequality increases under an impartial advisor.

	The above raises the question of whether a partial advisor can counterbalance the politician's discrimination and restore equality. If so, does the required level of partiality $p$ increase with the level of discrimination $\delta$? The next lemma addresses these questions.

	\begin{lemma}\label{prop:discrimination_equalizing_bias}
		For every politician's level of discrimination $\delta$, there exists a unique advisor's level of partiality $\hat{p}(\delta) > 0$ that ensures equality in equilibrium: $I(\hat{p}(\delta), \delta) = 0$. The function $\hat{p}(\delta)$ is continuous and \emph{non-monotone}:
		there exists a unique $\overline{\delta}\in (0,1)$ such that $\hat{p}(\delta)$ is strictly increasing for $\delta <\overline{\delta}$ and strictly decreasing for $\delta >\overline{\delta}$.
	\end{lemma}
 
	As noted earlier, increases in $\delta$ reduce the advisor's incentive to learn about both groups, as information about group 1 is increasingly overemphasized in the politician's decision, while information about group 2 is increasingly underused. For an advisor with partiality $p>0$, the former effect dominates the latter, prompting him to learn more about group 2 as $\delta$ increases. Moreover, the disparity between the two effects is greater for advisors with higher partiality $p$, as they prioritize group 2 over group 1 more. As such, advisors with higher partiality adjust their learning strategies more sharply in response to changes in $\delta$ than do those with lower partiality.
 
	On the other hand, as the level of discrimination $\delta$ increases,
	restoring equality requires more wasted information: the advisor must learn less about group~1 and more about group~2, even though the politician becomes less responsive to information about group~2. Hence, the politician's decision increasingly relies on her prior beliefs---where there is no disagreement between groups---rather than the advisor's information.  
	As the amount of ``effectively'' utilized information diminishes with higher $\delta$, the adjustments required in the learning strategy to restore equality become progressively smaller.
	
	In summary, two effects are at play: advisors with higher partiality $p$ respond more strongly to changes in $\delta$, while the degree of adjustment in tests to achieve equality diminishes with higher $\delta$. These dynamics then lead to the non-monotonicity result in Lemma~\ref{prop:discrimination_equalizing_bias}. Additionally, since restoring equality involves wasting more information as $\delta$ increases, it results in a welfare loss. Indeed, welfare $W(\hat{p}(\delta),\delta)$ decreases with $\delta$.

    As we have seen, welfare is maximized with an impartial advisor, but then inequality increases with $\delta$.
    Conversely, with partial advisors who restore equality, welfare then decreases with $\delta$. Hence, welfare maximization and inequality minimization are misaligned objectives, as captured in the following proposition. 
    \begin{proposition}\label{prop:discrimination_frontier}
		For any level of discrimination $\delta$, a welfare-inequality Pareto frontier is formed by advisors having partiality levels $p\in [0,\hat{p}(\delta)]$, where both welfare and inequality strictly decrease in $p$. Moreover, welfare with the equality-restoring advisor $\hat{p}(\delta)$ is strictly higher than under unchecked discrimination.
	\end{proposition}

	Proposition~\ref{prop:discrimination_frontier} illustrates a trade-off between welfare and inequality. Reducing inequality requires wasting more information, achieved by advisors with higher levels of partiality, while increasing welfare requires using information more efficiently, achieved by less partial advisors. This dynamic gives rise to the Pareto frontier formed by advisors ranging from impartial to those restoring equality. Moreover,
    any advisor along this frontier unambiguously improves both welfare and inequality compared to unchecked discrimination. Notably, even an advisor who restores equality results in higher welfare than does unchecked discrimination. 
	This is because welfare is concave in test allocation, and the learning strategy of an unchecked politician is heavily skewed towards learning about the needs of group 1, whereas an ``equalizing'' advisor learns more evenly about both groups (even though his learning is skewed towards group 2).

	%We prove the result in two steps. First, as the politician's level of discrimination towards group~1 increases, an advisor partial towards group 2 compensates by learning more about group~2. This adjustment occurs disproportionately strongly, as the cross derivative of the number of tests allocated to group 2 in equilibrium is positive, $\frac{\partial^2\tau_2^\ast}{\partial \delta \partial p}>0$. Hence, as $\delta$ increases, an advisor with higher partiality compensates for this increase by reallocating more tests. This has two implications: first, there always exists a unique partial advisor under whom there is no inequality in equilibrium. Second, the required level of partiality is an inverse U-shaped curve since the same-size compensation can be achieved by less partiality when the politician's level of discrimination increases. 

	%==========================
	%==========================
	\section{Extensions} \label{sec:extensions}
	%==========================
	%==========================

In this section, we outline two extensions, presented in more detail in Online Appendix \ref{sec:online_app_B}. The first examines media polarization by introducing two researchers into the model, representing media outlets competing to influence a voter (decision-maker). Each outlet and the voter seek implementation of their preferred policy mix on two policy issues. In the model, first, the voter allocates her attention between the two outlets. Second, the outlets simultaneously decide how much coverage to devote to the two policy issues. The combined attention and coverage provide the voter with signals about the two policy issues. Finally, the voter casts her ballot. We provide conditions under which media outlets become fully polarized, with each outlet covering only a single issue.\footnote{Research shows that individuals tend to seek information that confirms their prior beliefs, interpreting ambiguous information in a way that aligns with these beliefs (see \citealp{fryer_updating_2019,nimark_inattention_2019,olszewski_preferences_2021,damico_disengaging_2022} for some examples). Additionally, it has been suggested that information providers, be it media or peers on social media, may contribute to polarization by strategically selecting the information they supply \citep{pogorelskiy_news_2019,germano_crowding_2022,aina_tailored_2024}.} This polarization benefits the voter, however, as she can acquire her first-best information by splitting her attention appropriately between the two outlets. Hence, when media outlets care about the implemented policy mix, polarization occurs in equilibrium but helps, rather than harms, the voter. The voter thus prefers a polarized media duopoly to a monopoly of either media outlet (which is then more moderate in its coverage).

In the second extension, we introduce uncertainty about the decision-maker's preferences within a dual-self model. A single agent controls both learning and decision-making, but her preferences might change between the two stages due to factors like job loss, illness, or self-control issues.
%a (potential) change in agent's preferences between learning and decision-making affects their learning strategy. 
%For example, this change may result from a shift in the agent's circumstances, such as job loss, illness, or self-control issues, like succumbing to temptation. 
We study how this uncertainty shapes the learning strategy and compare welfare outcomes between na\"ive agents (unaware of potential preference changes) and 
 sophisticated agents (aware of potential changes)  as in \citet{odonoghue1999doing}.
%We model this as an interaction between multiple selves over time and explore both naive and sophisticated agents \citep*[][]{odonoghue1999doing}. 
%We study the implications of either mindset for the agent's welfare. To do so, we need to extend our baseline model to allow for uncertainty about the decision-maker's preference at the learning stage. 
We identify conditions under which the sophisticate may opt for \emph{strategic ignorance}, completely avoiding learning about an attribute to prevent excessive overreaction by the future self. Further, we delineate when the expected utilities of the naif and the sophisticate differ and in which direction. Notably, if the welfare criterion is measured by the changed utility, the naif can outperform the sophisticate.  This happens when the preference change increases the weight on the affected attribute, making the sophisticate's ``hedging'' strategy---effectively assigning less weight to that attribute due to possible overreaction by future self---counterproductive, regardless of whether preferences ultimately change.

	%==========================
	%==========================
	\section{Conclusion} \label{sec:conclusion}
	%==========================
	%==========================
	
	We present a model of delegated learning in a multi-attribute context. We explore the strategic  interaction between a decision-maker, who makes a decision influenced by the state of the world, and a researcher, who allocates a budget of tests to procure information about various attributes of the state. The players have different preferences for the final decision, as both seek to align the decision with the ``weighted state" but assign different weights to the attributes. Our insights shed light on how the researcher's preference misalignment affects the equilibrium learning strategy and the players' resulting payoffs.
		
	We characterize the equilibrium test allocation in this strategic interaction. We show it coincides with the solution to the single-player problem for certain attribute weights, which generally differ from both the decision-maker's and the researcher's weights. %(and may not even fall in between the two). 
    A key insight is that the researcher may forgo learning altogether if he anticipates an excessive  overreaction of the decision-maker to new information. More broadly, preference misalignment leads to reduced emphasis on attributes that spark greater disagreement between the two players. We also show that our characterization applies beyond our baseline model and extends to other models and preference formulations. Finally, while it is possible to include correlation between attributes, the economic insights remain unchanged in their essence while rendering the model substantially less tractable.

    The model provides a fruitful foundation for applied work, as demonstrated by the applications and extensions explored herein. Our framework is both rich and tractable, making it adaptable to a wide variety of settings. In addition to the problems considered in our applications and extensions, our model can be applied to many other problems. These include but are not limited to strategic management under uncertainty, biased portfolio advising, and optimal delegation (of learning and decision rights).

	\appendix
	
	%==============================
	%=============================
	\section{Main Proofs}
	%===============================
	%================================

	%=====================================================
	\subsection{Proofs for Section \ref{sec:single_player}: \nameref{sec:single_player}}
	%=====================================================
	
	%=========================================
	\subsubsection{Proof of Lemma \ref{lemma:single_player_maximization}}
	%=========================================
	
	Given the optimal decision strategy $d^*(s,\tau)$ given by equation \eqref{eq:optimal_decision}, the agent chooses $\tau \in \mathbb{R}^K_+$ with $\sum_k \tau_k \leq T$ to maximize
	    \begin{align*}
		V(\tau) &= \mathbb{E} \left[ - \left( \tilde{d}^*(\tau) - \tilde{b} \right)^2 \right] = -\mathbb{E} \left[ \left(\tilde{d}^*(\tau)\right)^2 \right] + 2\mathbb{E} \left[ \tilde{d}^*(\tau) \tilde{b} \right] - \mathbb{E} \left[ \tilde{b}^2 \right] \\
		&= - \left( var \left(\tilde{d}^*(\tau)\right) + \left(\mathbb{E} \left[\tilde{d}^*(\tau)\right] \right)^2 \right)
				+ 2 \left( \cov \left( \tilde{d}^*(\tau), \tilde{b} \right) + \mathbb{E} \left[\tilde{d}^*(\tau)\right]  \mathbb{E} \left[ \tilde{b} \right] \right) 
		\\ &\quad
		- \left( var \left( \tilde{b} \right) + \left(\mathbb{E} \left[\tilde{b}\right] \right)^2 \right) \\
		&= - (\psi(\tau) + (b_0)^2) + 2 (\psi(\tau)+ b_0 b_0) - (\sigma^2_0+(b_0)^2) = \psi(\tau) - \sigma^2_0 =  - \hat{\sigma}^2(\tau),
	\end{align*} 
	where $\hat{\sigma}^2(\tau)$ is given by equation \eqref{eq:bliss_posterior_var}.
    \footnote{
In the second equality, we used $\operatorname{var}(\tilde X)=\mathbb{E}[\tilde X^2]-(\mathbb{E}[\tilde X])^2$ and $\operatorname{cov}(\tilde X,\tilde Y)=\mathbb{E}[\tilde X\tilde Y]- \mathbb{E}[\tilde X]\mathbb{E}[\tilde Y]$. 
In the third equality, note that $\tilde d^*(\tau)=\mathbb{E}[\tilde b\mid s,\tau]$ is the conditional expectation of $\tilde b$ given the signals. Therefore, $\operatorname{cov}(\tilde d^*(\tau),\tilde b)=\operatorname{cov}(\tilde d^*(\tau),\tilde d^*(\tau)+\tilde\nu)=\operatorname{var}(\tilde d^*(\tau))=\psi(\tau)$, where $\tilde\nu:=\tilde b-\mathbb{E}[\tilde b\mid s,\tau]$ with $\mathbb{E}[\tilde\nu\mid s,\tau]=0$. Finally, we use the law of total variance: $\psi(\tau)=\sigma_0^2-\hat{\sigma}^2(\tau)$, with $\hat{\sigma}^2(\tau)$ given by \eqref{eq:bliss_posterior_var}.}

	\qed

	%=========================================
	\subsubsection{Proof of Theorem \ref{thm:benchmark}} \label{proof:LMS_theorem2}
	%=========================================
    We prove a general characterization of the optimal test allocation for $K\geq 2$, from which Theorem~\ref{thm:benchmark} immediately follows.
    \begin{theorem}\label{thm:benchmark_general}
		In the single-player case, without loss of generality, let $\alpha_1 \Sigma^0_{1} \geq \hdots \geq  \alpha_K \Sigma^0_{K}$.
		Then, the unique optimal test allocation $\tau^*$ is as follows:
		if $\alpha_k = 0$ for all $k$, then $\tau^* = (0, \ldots, 0)$;
		otherwise,\footnote{We apply the convention $\sum_{j=k}^{l} x_j = 0$ when $l<k$.}
		\begin{align*}
			\tau_k^* = \left( T - \bar{T}_{J} \right) \frac{\alpha_k}{\sum_{l=1}^{J} \alpha_l} + \sum_{j=k}^{J-1} \left( \left( \bar{T}_{j+1} - \bar{T}_{j} \right) \frac{\alpha_k}{\sum_{l=1}^{j} \alpha_l} \right)
		\end{align*}
		for all $k =1, \ldots, J$ and $\tau^*_k = 0$ for all $k =J+1, \ldots, K$,
		where the budget thresholds $\bar{T}_1 = 0 \leq \bar{T}_2 \leq \ldots \leq \bar{T}_{K+1}$ are given by
		\begin{align*}
			\bar{T}_k &:= \frac{\sum_{l=1}^{k-1}\alpha_l}{\alpha_{k}\Sigma_{k}^0} - \sum_{l=1}^{k-1} \frac{1}{\Sigma_l^{0}}
		\end{align*} 
		for all $k=1, \ldots, K$, $\bar{T}_{K+1} := +\infty$,
		and $J =1, \ldots, K+1$ is such that $T \in \left[\right. \bar{T}_{J}, \bar{T}_{J+1} \left.\right)$.
	\end{theorem}
    \begin{proof}
	By Lemma \ref{lemma:single_player_maximization}, the agent's objective function when choosing a test allocation $\tau \in \mathcal{T}$ is given by
	\begin{align} \label{eq:sp_objfn_as_sum}
		V(\tau)= - \hat{\sigma}^2(\tau) = - \sum_{k=1}^K \left[ \alpha_k^2 \left( \frac{1}{\Sigma^0_k} + \tau_k \right)^{-1} \right].
	\end{align}
	If $\alpha=(0,\hdots,0)$, then $V(\tau)=0$ for all $\tau$. Hence, any test allocation $\tau \in \mathcal{T}$ is optimal. By the assumption that in case of indifference, the agent abstains from learning, we get the first part of the statement.  
 
	For the remainder of the proof, assume $\alpha_1>0$.
	It is immediate that \eqref{eq:sp_objfn_as_sum} is concave in $\tau$ (and all constraints are convex in $\tau$), hence we can use the Karush-Kuhn-Tucker technique to find the solution. 
	Let $\xi$ denote the Lagrange multiplier on the budget constraint and $\eta_k$ denote the Lagrange multiplier on the non-negativity constraint for $\tau_k$, with $\eta := (\eta_1, \ldots, \eta_K)$. Then the KKT first order conditions can be written as follows:
	\begin{align} 
		\label{eq:sp_prb_foc1}
		\alpha^2_k \left( \frac{1}{\Sigma^0_k} + \tau_k \right)^{-2} &= \xi - \eta_k
		\ \text{ for all } k =1, ..., K,
		\\
		\label{eq:sp_prb_foc2}
		0 &= \eta_k \tau_k
		\ \text{ for all } k =1, ..., K,
		\\
		\label{eq:sp_prb_foc3}
		\eta_k &\geq 0,\ \tau_k \geq 0
		\ \text{ for all } k =1, ..., K,
		\\
		\label{eq:sp_prb_foc4}
		0 &= \xi \left( T - \sum_{k=1}^K \tau_k \right),
		\\
		\label{eq:sp_prb_foc5}
		T &\geq \sum_{k=1}^K \tau_k,\ \xi \geq 0.
	\end{align}
	
	Observe that the marginal benefit of increasing $\tau_1$ is $\alpha^2_1 \left( \frac{1}{\Sigma^0_1} + \tau_1 \right)^{-2}>0$ since $\alpha_1 > 0$. Hence, for $k=1$, the LHS of condition \eqref{eq:sp_prb_foc1} is strictly positive, and so must be the RHS, implying $\xi>0$ (since $\eta_1\geq 0$). Therefore, from conditions \eqref{eq:sp_prb_foc4} and \eqref{eq:sp_prb_foc5}, the budget constraint must be binding in the optimum: $T = \sum_{k=1}^K \tau_k$.

	Consider an arbitrary $\xi>0$. Then for every $k$ there exists a unique pair $(\tau_k(\xi),\eta_k(\xi))$ satisfying conditions \eqref{eq:sp_prb_foc1}-\eqref{eq:sp_prb_foc3}. This pair is given by
	\begin{align}\label{eq:sp_prb_partsoln}
		\tau_k(\xi) &= 
		\begin{cases}
			\frac{\alpha_k}{\sqrt{\xi}}-\frac{1}{\Sigma^0_k} & \text{if} \ \xi \leq \left(\alpha_k\Sigma^0_k\right)^2 \\
			0 & \text{otherwise} 
\		\end{cases},
		&
		\eta_k(\xi) &=
		\begin{cases}
			0 & \text{if} \  \xi \leq \left(\alpha_k\Sigma^0_k\right)^2 \\
			\xi - \left( \alpha_k \Sigma^0_k  \right)^2 & \text{otherwise}
		\end{cases}
	\end{align}
	To find a solution, next we need to find value $\xi^*>0$ under which also the budget constraint $T=\sum_k \tau_k (\xi^*)$ is satisfied. Note that 	
 	$\tau_k(\xi)$ are continuous and weakly decreasing in $\xi$ for all $k$, and strictly decreasing whenever $\tau_k(\xi) > 0$.

	Further, for all $k$ we have $\lim\limits_{\xi \to +\infty} \tau_k(\xi) = 0$ and $\lim\limits_{\xi \to 0} \tau_k(\xi) = +\infty$. The above imply that function $B(\xi) := \sum_{k=1}^K \tau_k(\xi)$ is continuous, strictly decreasing in $\xi$ whenever $B(\xi)>0$ (i.e. whenever there exist $k$ with $\tau_k(\xi)>0$), and has the same limit properties. By the intermediate value theorem, there exists unique $\xi^*>0$ such that $B(\xi^*) = T>0$ (the budget constraint binds). This value $\xi^*$ together with the corresponding pairs $\left( \tau_k(\xi^*), \eta_k(\xi^*) \right)$ for every $k$ then solve the entire FOC system \eqref{eq:sp_prb_foc1}--\eqref{eq:sp_prb_foc5}, and thus uniquely solve the player's problem from Lemma \ref{lemma:single_player_maximization}. %By the monotonicity of \eqref{eq:sp_prb_partsoln} in $\xi$, the full solution is unique.
	
	It remains to show that the solution $\tau^* := \left(\tau_1(\xi^*),\hdots,\tau_K(\xi^*)\right)$ satisfies the representation of the statement. Let us use the substitution $\zeta := \frac{1}{\sqrt{\xi}}$, and let us redefine the functions $\tau_k(\xi)$ and $B(\xi)$ to be functions of $\zeta$. Hence,  let $\hat{\tau}_k(\zeta) = \max \left\{ 0, \alpha_k \zeta - \frac{1}{\Sigma^0_k} \right\}$ and $\hat{B}(\zeta) = \sum_{k=1}^K \hat{\tau}_k(\zeta)$. Like before, let $\zeta^*=\frac{1}{\sqrt{\xi^*}}$ be the unique solution to $\hat{B}(\zeta^*) = T$. Let further $\zeta_k := \frac{1}{\alpha_k \Sigma^0_k}$ for all $k$, and note that $\hat{\tau}_k(\zeta) > 0$ if and only if $\zeta > \zeta_k$, as well as the fact that by assumption, $\zeta_1 \leq \ldots \leq \zeta_K$.
	
	To calculate $\tau^*$, we can use the fundamental theorem of calculus:
	\begin{align*}
		\tau^*_k = \hat{\tau}_k(\zeta^*) &= \int\limits^{\zeta^*}_{\min \left\{ \zeta_k, \zeta^* \right\}} \frac{d\hat{\tau}_k(\zeta)}{d\zeta} d\zeta 
		= \int\limits^{\zeta^*}_{\min \left\{ \zeta_k, \zeta^* \right\}} \alpha_k d\zeta 
		= \alpha_k \left( \zeta^* - \min \left\{ \zeta_k, \zeta^* \right\} \right).
	\end{align*}
	Let $J =1, \ldots, K+1$ be such that $\zeta^* \in \left[ \zeta_{J}, \zeta_{J+1} \right)$, with $\zeta_{K+1} := +\infty$.\footnote{Note that it must hold that $\zeta^*>\zeta_1$, which implies that there exists $J=1,\hdots,K $ with $\zeta^* \in \left[ \zeta_J, \zeta_{J+1} \right]$. Otherwise, if $\zeta^*\leq \zeta_1$, it would imply that $\hat{\tau}_k(\zeta^*)=0$ for all $k$, hence $\hat{B}(\zeta^*)=0$, which is a contradiction to $\hat{B}(\zeta^*)=T>0$.}
	Then for all $k =J+1, \ldots, K$ we have $\tau^*_k = 0$, and for all $k =1, \ldots, J$ we have:
	\begin{align}
		\nonumber
		\zeta^* - \zeta_k &= \left( \zeta^* - \zeta_J \right) + \left( \zeta_J - \zeta_{J-1} \right) + \ldots + \left( \zeta_{k+1} - \zeta_k \right)
		\\
		\label{eq:sp_opt_tau_part}
		\Rightarrow
		\tau^*_k &= \left( \zeta^* - \zeta_{J} \right) \alpha_k + \sum_{j=k}^{J-1} \left( \zeta_{j+1} - \zeta_j \right) \alpha_k .
	\end{align}
	
	To find $\left( \zeta^* - \zeta_{J} \right)$ and $\left(\zeta_{j+1} - \zeta_j \right)$ in this case, we again invoke the fundamental theorem of calculus: 
	\begin{align}
		\nonumber
		T = \hat{B}(\zeta^*) &= \int\limits^{\zeta^*}_{\zeta_1} \frac{d\hat{B}(\zeta)}{d\zeta} d\zeta 
		= \int\limits^{\zeta^*}_{\zeta_1} \sum_{k=1}^K \frac{d\hat{\tau}_k(\zeta)}{d\zeta} d\zeta \\ \label{eq:sp_FTC}
	&= \int\limits^{\zeta_2}_{\zeta_1} \alpha_1 d\zeta 
		+\int\limits^{\zeta_3}_{\zeta_2} \left( \alpha_1 + \alpha_2 \right) d\zeta
		+\ldots 
		+\int\limits^{\zeta^*}_{\zeta_{J}} \left( \sum_{l=1}^{J} \alpha_l \right) d\zeta
		%\\
		%&= \left( \hat{\zeta} - \bar{\zeta}_{J} \right) \left( \sum_{l=1}^{J} \alpha_l \right) + \sum_{j=1}^{J-1} \left( \left( \bar{\zeta}_{j+1} - \bar{\zeta}_j \right) \sum_{l=1}^j \alpha_l \right).
	= \sum_{l=1}^{J} \left( \zeta^* - \zeta_l \right) {\alpha}_l
	\end{align}
	Specifically, then, if we let
	\begin{align*}
		\bar{T}_k := \hat{B} \left( \zeta_k \right) 
		%&= \sum_{j=1}^{k-1} \left( \left( \bar{\zeta}_{j+1} - \bar{\zeta}_j \right) \sum_{l=1}^j \hat{\alpha}_l \right)
		%= \bar{\zeta}_{k} \left( \sum_{l=1}^{k-1} \hat{\alpha}_l \right) - \left( \sum_{l=1}^{k-1} \bar{\zeta}_l \hat{\alpha}_l \right) 
		%\\
		&= \sum_{l=1}^{k-1} \left( \zeta_k - \zeta_l \right) {\alpha}_l
		%\\
		= \frac{\sum_{l=1}^{k-1}{\alpha}_l}{{\alpha}_{k}\Sigma_{k}^0} - \sum_{l=1}^{k-1} \frac{1}{\Sigma_l^{0}},
	\end{align*}
	expression \eqref{eq:sp_FTC} implies that $T - \bar{T}_{J} = \left( \zeta^* - \zeta_{J} \right) \sum_{l=1}^{J} \alpha_l \iff \zeta^* - \zeta_{J} = \frac{T - \bar{T}_{J}}{\sum_{l=1}^{J} \alpha_l}$ and, similarly, $\bar{T}_{k+1} - \bar{T}_k = \left( \zeta_{k+1} - \zeta_k \right) \sum_{l=1}^k {\alpha}_l \iff \zeta_{k+1} - \zeta_k = \frac{\bar{T}_{k+1} - \bar{T}_k}{\sum_{l=1}^k {\alpha}_l}$ for any $k =1, \ldots, K-1$. Plugging these into \eqref{eq:sp_opt_tau_part} yields the solution presented in the statement.
	\end{proof} %\qed 	

	%=====================================================
	\subsection{Proofs for Section \ref{sec:strategic_players}: \nameref{sec:strategic_players}}
	%=====================================================
	
	%====================================
	\subsubsection{Proof of Lemma~\ref{lemma:strategic_player_maximization}}
	%=========================================
	
	%By following the same steps as in the proof of Theorem 4.1 in \citet{bardhi_attributes:_2024}, 
	Given the test allocation $\tau$ and the decision-maker's equilibrium decision strategy $d^*(s,\tau)=\mathbb{E} \left[ \tilde{b}^{D}(s,\tau) \right]$, the researcher's expected payoff is
	\begin{align*}
		V^R(\tau) 
        &= \mathbb{E}\left[-\left( \tilde{d}^*(\tau) - \tilde{b}^R  \right)^2 \right] 
        = 
        - \mathbb{E}\left[ \left(\tilde{d}^*(\tau)\right)^2 \right] + 2 \mathbb{E}\left[ \tilde{d}^*(\tau) \tilde{b}^R \right] - \mathbb{E}\left[\left( \tilde{b}^R \right)^2 \right] \\
        &=
        - \left( \psi^D(\tau) + \left( b^D_0 \right)^2 \right) + 2 \left( \cov\left[ \tilde{d}^*(\tau),\tilde{b}^R \right] + \mathbb{E}\left[\tilde{d}^*(\tau) \right] \mathbb{E}\left[ \tilde{b}^R\right]  \right)
        - \left( \sigma^{2,R}_0 + \left( b^R_0 \right)^2  \right) \\
        &=
         - \left( \psi^D(\tau) + \left( b^D_0 \right)^2 \right) + 2 \left( \cov\left[ \tilde{d}^*(\tau),\tilde{b}^R \right] + b^D_0 b^R_0  \right)
        - \left( \sigma^{2,R}_0 + \left( b^R_0 \right)^2  \right) \\
        &=       
        2\cov\left[ \tilde{d}^*(\tau), \tilde{b}^R \right] - \psi^{D}(\tau) + V^R(\emptyset)
	\end{align*}
	where $V^R(\emptyset) = - \left(b^{D}_0 -b^R_0 \right)^2 - \sigma^{2,R}_0$ is the researcher's expected utility under no learning.	To complete the proof, let us show that
	\begin{eqnarray*}
		\cov \left( \tilde{d}^*(\tau),\tilde{b}^{R} \right) = \sum_k \alpha^{D}_k \alpha^R_k \Sigma^0_k \hat{\Sigma}_k(\tau_k) \tau_k.
	\end{eqnarray*}
	The researcher's bliss point satisfies $\tilde{b}^R = \sum_k \alpha^R_k \tilde{\theta}_k$. The random variable capturing the equilibrium decision-maker's decision satisfies
	\begin{eqnarray*}
		\tilde{d}^*(\tau) &=&
		\sum_k \alpha^{D}_k \hat{\mu}_k (\tilde{s}_k,\tau_k) =
		\sum_k \alpha^{D}_k \left( \frac{1}{1+\tau_k\Sigma^0_k} \mu^0_k + \frac{\tau_k \Sigma^0_k}{1+\tau_k\Sigma^0_k} \tilde{s}_k \right) \\
		&=&
		\sum_k \alpha^{D}_k \left( \frac{1}{1+\tau_k\Sigma^0_k} \mu^0_k + \frac{\tau_k \Sigma^0_k}{1+\tau_k\Sigma^0_k} \left( \tilde{\theta}_k + \tilde{\varepsilon}_k  \right) \right)
	\end{eqnarray*}
	Hence,
	\begin{eqnarray*}
		\cov(\tilde{d}^*(\tau),\tilde{b}^{R}) &=&
		\cov \left( \sum_k \alpha^{D}_k \frac{\tau_k \Sigma^0_k}{1+\tau_k\Sigma^0_k} \tilde{\theta}_k, \sum_k \alpha^{R}_k \tilde{\theta}_k    \right)
		=
		\sum_k \alpha^{D} \alpha^R_k \tau_k \frac{\Sigma^0_k}{1+\tau_k\Sigma^0_k} var(\tilde{\theta}_k) \\
		&=& \sum_k \alpha^{D}_k \alpha^R_k \tau_k \hat{\Sigma}_k(\tau_k) \Sigma^0_k,
	\end{eqnarray*}
	where we used (i) $\cov(\tilde{\theta}_k,\tilde{\varepsilon}_j)=0$ for all $k$ and $j$, since the noise terms are independent of all other random variables; (ii) $\cov(\tilde{\theta}_k,\mu^0_j) = 0$ for all $k$ and $j$ since $\mu^0_j$ is a constant; and (iii) $\cov(\tilde{\theta}_k,\tilde{\theta}_j)=0$ for all $k\neq j$ since the attributes are independent.
	\qed

	%=====================================
	\subsubsection{Proof of Theorem \ref{thm:strategic_solution}}
	%======================================

    Recall that $\psi^{D}(\tau) = \sigma^{2,D}_0 - \hat{\sigma}^{2,D}(\tau)$	is the variance of the decision-maker's decision, where $\hat{\sigma}^{2,D}(\tau)$ is given by \eqref{eq:bliss_posterior_var}. Then, from Lemma~\ref{lemma:strategic_player_maximization}, %and dropping terms that are constant in $\tau$, 
    the researcher chooses the test allocation to maximize his expected payoff
    \begin{align*}
	V^{R}(\tau) 
	&= 2 \cov\left( \tilde{d}^*(\tau), \tilde{b}^R  \right) + \hat{\sigma}^{2,D}(\tau) + const_1 \\
	%&= 2\sum_k \alpha^{D}_k \alpha^R_k \underbrace{\Sigma^0_k \tau_k \hat{\Sigma}_k(\tau_k)}_{=\Sigma^0_k - \hat{\Sigma}_k(\tau_k)} + \sum_k \left( \alpha^{D}_k \right)^2 \hat{\Sigma}_k(\tau_k) +const_1 
	%\\
	&= -\sum_k \left( 2\alpha^D_k \alpha^R_k - \left(\alpha^D_k\right)^2  \right) \hat{\Sigma}_k(\tau_k) + \underbrace{\sum_k 2\alpha^D_k\alpha^R_k \Sigma^0_k}_{:=const_2} + const_1 \\
	&= -\sum_k \lambda_k \hat{\Sigma}_k(\tau_k) + const_1 + const_2
\end{align*}
subject to the non-negativity and the budget constraints. Note that $const_1$ and $const_2$ are constant in $\tau$, $\lambda_k := 2 \alpha^R_k \alpha^D_k - \left( \alpha^{D}_k \right)^2$ for every $k$, and in the second step we used the fact that $\Sigma^0_k\tau_k \hat{\Sigma}_k(\tau_k) = \Sigma^0_k \frac{\tau_k \Sigma^0_k}{1+\tau_k\Sigma^0_k}=\Sigma^0_k\frac{\tau_k \Sigma^0_k+1-1}{1+\tau_k\Sigma^0_k}=\Sigma^0_k-\frac{\Sigma^0_k}{1+\tau_k\Sigma^0_k}=\Sigma^0_k - \hat{\Sigma}_k(\tau_k)$. 

The partial derivatives of $V^{R}(\tau)$ with respect to $\tau_k \geq 0$ are
\begin{eqnarray} \label{eq:foc_strategic}
	\frac{\partial V^{R}(\tau)}{\partial \tau_k}
	= \lambda_k \left( \frac{\Sigma^0_k}{1+\tau_k\Sigma^0_k} \right)^2.
\end{eqnarray}
Hence, the sign of $\lambda_k$ determines whether the objective function $V^R(\tau)$ is increasing, decreasing or constant in $\tau_k\geq 0$. Note that when $\lambda_k \leq 0$ the solution dictates $\tau^*_k = 0$ for such $k$.\footnote{When $\lambda_k \leq 0$ for all $k$ and there exists at least one $j$ such that $\lambda_j=0$, we use the assumption that in case of indifference the researcher abstains from learning, thus yielding $\tau^*=(0,0,\hdots,0)$. Otherwise, setting $\tau^*_k=0$ whenever $\lambda_k\leq 0$ is uniquely optimal: (i) either there exists at least one $j$ with $\lambda_j>0$ (and thus the function $V^{R}(\tau)$ is strictly increasing in $\tau_j$), or (ii) $\lambda_k <0$ for all $k$ (and thus the function $V^{R}(\tau)$ is \emph{strictly} decreasing in every $\tau_k$).} Hence, the maximizers of function $V^{R}(\tau)$ are the same as the maximizers of an adjusted function $V^{R,adj}(\tau):=V^{R}(\tau)$ where we set $\alpha^{D}_k=\alpha^R_k=0$ for all $k$ with $\lambda_k \leq 0$ (and hence $V^{R,adj}(\tau)$ is independent of $\tau_k$ for such $k$). 	

For every $k$, denote $\hat{\alpha}_k = \sqrt{\max\{0,\lambda_k\}}$. Recall the single-player objective function \eqref{eq:sp_objfn_as_sum} for an agent with weights $\hat{\alpha}$ is given by
\begin{align*}
	V(\tau) = - \sum_k \hat{\alpha}^2_k \hat{\Sigma}_k(\tau_k).
\end{align*}
Then, we get $V^{R,adj}(\tau) = V(\tau)+ const$. Hence, the maximizers of both of these functions, and thus also the maximizer of $V^R(\tau)$, coincide.
\qed

	\subsection{Proofs for Section \ref{sec:organizations}: \nameref{sec:organizations}}\label{appendix:organizations}
	%=====================================================

       We  first state four lemmas (proved in Online Appendix \ref{OA:organizations}) that we use to prove both Propositions  \ref{prop:CS_manager_gamma} and \ref{prop:CS_manager_beta}. Lemma~\ref{lemma:characterization_learning_low} and Lemma~\ref{lemma:characterization_learning_high} characterize the equilibrium learning strategy $\tau^*(\beta,\gamma)$ as a function of $\gamma$ when $\beta\leq 1/2$ and $\beta>1/2$, respectively. Lemma \ref{lemma:binary_Vp_concavity} establishes that the decision-maker's interim expected payoff is strictly increasing and concave in $(\tau_1,\tau_2)$. Lemma \ref{lemma:optimal_allocation_gamma_zero} derives the researcher's equilibrium test allocation for $\gamma=0$.		
	
	\begin{lemma}\label{lemma:characterization_learning_low}
		Let $\alpha^{D}_1\geq \alpha^{D}_2>0$ and assume the analyst's weight vector $\alpha^{R}$ is given by decomposition \eqref{eq:decomposition_binary}. For any $\beta \in [0, 1/2]$ there exist unique $\gamma_1(\beta),\gamma_2(\beta) \in \Gamma(\beta)$ such that $\gamma_1(\beta) \leq 0 \leq \gamma_2(\beta)$ and:
		\begin{eqnarray}
			\tau^*(\beta,\gamma) =
			\begin{cases}
				(0,T) & \text{if}\ \gamma>\gamma_2(\beta), \\
				(0,0) & \text{if}\ \gamma \in [\gamma_1(\beta),\gamma_2(\beta)], \\
				(T,0) & \text{if}\ \gamma< \gamma_1(\beta).
			\end{cases}
		\end{eqnarray}
		Further, there exist $\beta_1,\beta_2 \in \mathbb{R}$ with $0 < \beta_2 \leq \beta_1 < 1/2$ such that $\gamma_1(\beta) \in \text{int}\left(\Gamma(\beta)\right)$ if and only if $\beta\in (\beta_1,1/2]$, and $\gamma_2(\beta) \in \text{int}\left(\Gamma(\beta)\right)$ if and only if $\beta\in (\beta_2,1/2]$.
	\end{lemma}

	\begin{lemma}\label{lemma:characterization_learning_high}
		%Fix the test budget $T>0$ and 
    Fix the manager's weight vector $\alpha^{D} = (\alpha^{D}_1,\alpha^{D}_2)$ with $\alpha^{D}_k >0$ for both $k=1,2$. Let the analyst's weight vector $\alpha^{R}(\beta,\gamma)$ be given by the decomposition \eqref{eq:decomposition_binary}. Fix $\beta>1/2$. Then there exist $\gamma_I(\beta), \gamma_{II}(\beta) \in \text{int}(\Gamma(\beta))$ with $\gamma_I(\beta)< \gamma_{II}(\beta)$ such that the analyst's equilibrium test allocation is
		\begin{eqnarray}
			\tau^*(\beta,\gamma) =
			%		 (\tau^*_1(\beta,\gamma),\tau^*_2(\beta,\gamma)) =
			\begin{cases}
				(T,0)	& \text{if}\ \gamma \leq \gamma_I(\beta) \\
				\left( \tilde{\tau}_1(\beta,\gamma), \tilde{\tau}_2(\beta,\gamma)  \right)		
				& \text{if} \ \gamma \in \left( \gamma_I(\beta), \gamma_{II}(\beta)  \right) \\
				(0,T)	& \text{if} \ \gamma \geq \gamma_{II}(\beta)
			\end{cases}
		\end{eqnarray}
		where 
		\begin{equation}\label{eq:tilde_tau_gamma}
		\begin{aligned}
			\tilde{\tau}_1(\beta,\gamma) 
			:=& \frac{\hat{\alpha}_1(\beta,\gamma) \Sigma^0_{1}-\hat{\alpha}_2(\beta,\gamma) \Sigma^0_{2}}{\Sigma^0_{1} \Sigma^0_{2} (\hat{\alpha}_1(\beta,\gamma) + \hat{\alpha}_2(\beta,\gamma))} + \frac{\hat{\alpha}_1(\beta,\gamma)}{\hat{\alpha}_1(\beta,\gamma)+\hat{\alpha}_2(\beta,\gamma)} T \\ 
			\tilde{\tau}_2(\beta,\gamma) 
			:=& T-\tilde{\tau}_1(\beta,\gamma)
		\end{aligned}
		\end{equation}
		and
		\begin{equation}\label{eq:tilde_alpha_gamma}
		\begin{aligned}
			\hat{\alpha}_1 (\beta,\gamma) 
			&:=& \sqrt{ \max \left\{ (2\beta-1)\left(\alpha^{D}_1\right)^2-2\gamma \alpha^{D}_1 \alpha^{D}_2, 0 \right\} } \\
			\hat{\alpha}_2 (\beta,\gamma)
			&:=& \sqrt{ \max \left\{ (2\beta-1)\left( \alpha^{D}_2 \right)^2 + 2\gamma \alpha^{D}_1 \alpha^{D}_2, 0 \right\} }
		\end{aligned}
		\end{equation}
		Moreover, $\tilde{\tau}_1(\beta,\gamma)$ ($\tilde{\tau}_2(\beta,\gamma)$) is strictly decreasing (strictly increasing) in $\gamma$ for $\gamma \in \left(\gamma_I(\beta), \gamma_{II}(\beta) \right)$.
	\end{lemma}

	\begin{lemma}\label{lemma:binary_Vp_concavity}
		The decision-maker's interim expected equilibrium payoff is a strictly increasing and strictly concave function of $(\tau_1,\tau_2)$.
	\end{lemma}
\begin{lemma}\label{lemma:optimal_allocation_gamma_zero}
		%	Suppose Assumption~\ref{assump:ties} holds.
		Fix the test budget $T>0$ and the decision-maker's weight vector $\alpha^{D} \in \mathbb{R}^2_{++}$.
		Suppose $\gamma=0$ (the agent is not distorted).
		\begin{itemize}
			\item[(i)] If $\beta\leq 1/2$ (the agent is too insensitive), then the agent optimally chooses no testing: $\tau^*(\beta,0) = (0,0)$.
			\item[(ii)] If $\beta>1/2$ (the agent is sufficiently sensitive), then the agent optimally chooses the decision-maker's most preferred test allocation: $\tau^*(\beta,0) = \tau^*(1,0)$. 
		\end{itemize}
	\end{lemma}

	%============================================
	\subsubsection{Proof of Proposition \ref{prop:CS_manager_gamma}}
	%=======================================

	Suppose the premise of Proposition~\ref{prop:CS_manager_gamma} holds. Proposition~\ref{prop:CS_manager_gamma}(i) is a direct implication of Lemma~\ref{lemma:characterization_learning_low}. The manager's interim expected payoff is strictly increasing in $\tau_1$ and $\tau_2$ by Lemma~\ref{lemma:binary_Vp_concavity}. Hence, the manager always strictly prefers any positive test allocation $\tau \neq (0,0)$ to abstaining from learning.
	
	Now take $\beta>1/2$. Lemma~\ref{lemma:optimal_allocation_gamma_zero} shows that (given the constraint $\tau\in \mathcal{T}$) the manager's interim expected payoff $V^{D}(\tau)$ is maximized when the agent's relative bias is $\gamma=0$. Furthermore, since $V^{D}(\tau)$ is strictly increasing in $\tau_1$ and $\tau_2$, the constraint $\tau_1+\tau_2\leq T$ must be binding at the manager's most preferred test allocation. Furthermore, Lemma~\ref{lemma:characterization_learning_high} shows when $\beta>1/2$, the agent's equilibrium test allocation $\tau^*(\beta,\gamma)$ is such that $\tau^*_1(\beta,\gamma)$ and $\tau^*_2(\beta,\gamma)$ are weakly decreasing and weakly increasing in $\gamma$, respectively, and $\tau^*_1(\beta,\gamma) + \tau^*_2(\beta,\gamma) = T$.  The claim of Proposition~\ref{prop:CS_manager_gamma}(ii) then follows from the strict concavity of $V^{D}(\tau)$ (shown in Lemma~\ref{lemma:binary_Vp_concavity}).
	\qed

	%============================================
	\subsubsection{Proof of Proposition \ref{prop:CS_manager_beta}}
	%=======================================
 
	Suppose $\gamma=0$. Then, Lemma \ref{lemma:optimal_allocation_gamma_zero} shows that the researcher's learning strategy is to not learn for $\beta\leq 1/2$ and then to implement the decision-maker's preferred learning strategy. As the decision-maker's interim expected payoff $V^{D}(\tau)$ is strictly increasing and concave in $\tau_1$ and $\tau_2$ (Lemma \ref{lemma:binary_Vp_concavity}), this implies that $V^{D}(\tau)$ is weakly increasing in $\beta$ for $\gamma=0$.

	Suppose from this point onwards that $\gamma>0$ (case $\gamma<0$ is completely analogous). 
	Then, for the researcher's weights to be weakly positive, we must have $\beta \geq \underline{\beta}:= \gamma \frac{\alpha_2^{D}}{\alpha_1^{D}}$.
	The equilibrium test allocation is given by Theorems~\ref{thm:benchmark} and ~\ref{thm:strategic_solution}, %Theorem~\ref{thm:strategic_solution} and Corollary~\ref{cor:binary_solution_single_player}, 
    where we can express $\hat{\alpha}_k$ for $k=1,2$ as \eqref{eq:tilde_alpha_gamma},\footnote{While Lemma \ref{lemma:characterization_learning_high} is only stated for $\beta > 1/2$, expression \eqref{eq:tilde_alpha_gamma} is well-defined for all $\beta > \underline{\beta}$.}
	and so
	\begin{align*}
		\lambda_1 &= (2\beta - 1) (\alpha^{D}_1)^2 - 2 \gamma \alpha^{D}_1 \alpha^{D}_2,
		\\
		\lambda_2 &= (2\beta - 1) (\alpha^{D}_2)^2 + 2 \gamma \alpha^{D}_1 \alpha^{D}_2.
	\end{align*}
	Note that $\lambda_k$ for $k=1,2$ is strictly increasing in $\beta$. Therefore, there exist $\beta_1$ and $\beta_2$ such that $\hat{\alpha}_k = 0$ for $\beta \leq \beta_k$, and $\hat{\alpha}_k$ is strictly positive and strictly increasing for $\beta > \beta_k$. 
	These values are given by the respective roots of $\lambda_k = 0$:\footnote{Note that these $\beta_k$ are different from the ones defined in the proof of Lemma \ref{lemma:characterization_learning_low}.}
	\begin{align*}
		\beta_1 := \frac{1}{2} + \gamma \frac{\alpha^{D}_2}{\alpha^{D}_1}, 
		&& 
		\beta_2 := \frac{1}{2} - \gamma \frac{\alpha^{D}_1}{\alpha^{D}_2}.
	\end{align*}
	Observe that $\beta_2 < 1/2 < \beta_1$ and $\underline{\beta} < \beta_1$. 
	Theorems~\ref{thm:benchmark} and \ref{thm:strategic_solution} %Corollary~\ref{cor:binary_solution_single_player} 
    then imply that if $\beta > \beta_2$, then $\exists k: \tau^*_k (\beta,\gamma) > 0$ and $\tau^*_1 (\beta,\gamma) + \tau^*_2 (\beta,\gamma) = T$ (since at least one of the auxiliary weights $\hat{\alpha}_k$ is then strictly positive). 
	Therefore, we can conclude that if $\underline{\beta} > \beta_2$, then $\tau^*(\beta, \gamma) = (0,T)$ for all $\beta \in [\underline{\beta}, \beta_1]$, and if $\underline{\beta} < \beta_2$, then
	\begin{align*}
		\tau^*(\beta, \gamma) = \begin{cases}
			(0,0) & \text{ for } \beta \in [\underline{\beta}, \beta_2],
			\\
			(0,T) & \text{ for } \beta \in (\beta_2, \beta_1].
		\end{cases}
	\end{align*}
	As $V^{D}(\tau)$ is strictly increasing in $\tau^*_2$, it follows that $V^{D}(\tau)$ is weakly increasing in $\beta$ in either of the two cases.
	
	Consider now the case when $\beta > \beta_1$, so we have $\hat{\alpha}_1,\hat{\alpha}_2 > 0$. 
	Theorems~\ref{thm:benchmark} and \ref{thm:strategic_solution}imply that the equilibrium test allocation is given by
	\begin{equation}
		\begin{aligned}
			\tau^*_1 
			&= \max \left\{ 0, \min \left\{ \frac{ \frac{\hat{\alpha}_1}{\hat{\alpha}_2} \Sigma^0_{1} - \Sigma^0_{2}}{\Sigma^0_{1}\Sigma^0_{2}( \frac{\hat{\alpha}_1}{\hat{\alpha}_2} + 1 )} + 
			\frac{ \frac{\hat{\alpha}_1}{\hat{\alpha}_2} }{ \frac{\hat{\alpha}_1}{\hat{\alpha}_2} + 1 } T, T  \right\}    \right\}, \\
			\tau^*_2
			&= T- \tau^*_1.
		\end{aligned}
	\end{equation}
	Note that the expression above (and, hence, $V^{D}(\tau)$) only depends on $\beta$ through $\frac{\hat{\alpha}_1}{\hat{\alpha}_2}$. 
	Observe that $\lim\limits_{\beta \to +\infty} \frac{\hat{\alpha}_1}{\hat{\alpha}_2} = \frac{\alpha^{D}_1}{\alpha^{D}_2}$, so that $\lim \limits_{\beta \to +\infty} \tau^*(\gamma,\beta) = \tau^*(\alpha^{D})$, so the equilibrium test allocation converges to the decision-maker's optimal (payoff-maximizing) test allocation as $\beta \to \infty$. 	Further, $\frac{\hat{\alpha}_1}{\hat{\alpha}_2}$ is strictly increasing in $\beta$ in the case considered:
	\begin{align*}
		\frac{\partial}{\partial \beta} \frac{\hat{\alpha}_1}{\hat{\alpha}_2} 
		= \frac{1}{\hat{\alpha}_2} \frac{\partial \hat{\alpha}_1}{\partial \beta} - \frac{\hat{\alpha}_1}{\hat{\alpha}_2^2} \frac{\partial \hat{\alpha}_2}{\partial \beta}
		= \frac{\hat{\alpha}_1}{\hat{\alpha}_2} \left( \left( \frac{\alpha^{D}_1}{\hat{\alpha}_1} \right)^2 - \left( \frac{\alpha^{D}_2}{\hat{\alpha}_2} \right)^2 \right) > 0,
	\end{align*}
	where the inequality follows from $\hat{\alpha}_k = \sqrt{ \lambda_k }$ and
	\begin{align*}
		\left( \frac{\alpha^{D}_1}{\alpha^{D}_2} \right)^2 > \left( \frac{\hat{\alpha}_1}{\hat{\alpha}_2} \right)^2
		&= \frac{ (2\beta - 1) (\alpha^{D}_1)^2 - 2 \gamma \alpha^{D}_1 \alpha^{D}_2 }{ (2\beta - 1) (\alpha^{D}_2)^2 + 2 \gamma \alpha^{D}_1 \alpha^{D}_2 }= \frac{ (\alpha^{D}_1)^2 - \frac{2 \gamma}{2 \beta - 1} \alpha^{D}_1 \alpha^{D}_2 }{ (\alpha^{D}_2)^2 + \frac{2 \gamma}{2 \beta - 1} \alpha^{D}_1 \alpha^{D}_2 }.
	\end{align*}
	Together with the limit result above, this implies that $\frac{\hat{\alpha}_1}{\hat{\alpha}_2}$ and $\tau^*(\gamma,\beta)$ \emph{monotonically} converge to $\frac{\alpha^{D}_1}{\alpha^{D}_2}$ and $\tau^*(\alpha^{D})$, respectively, as $\beta \to \infty$. 
	As $V^{D}(\tau)$ is concave in $\tau$ (see Lemma \ref{lemma:binary_Vp_concavity}) and maximized by $\tau^*(\alpha^{D})$, it must be weakly increasing in $\beta \in (\beta_1, +\infty)$. 
	\qed

	%=====================================================
	\subsection{Proofs for Section \ref{sec:discrimination}: \nameref{sec:discrimination}}\label{appendix:discrimination}
	%=====================================================

	%=========================================
	\subsubsection{Proof of Lemma \ref{prop:discrimination_impartial}}
	%=========================================
    For the first part of the statement, observe that the maximization problem of an impartial advisor is
    \begin{eqnarray*} 
    	\max_{\tau \in \mathcal{T}} \frac{1}{2} \mathbb{E}\left[u_1 \left( \tilde{d}^{*}(\tau),\tilde{\theta}_1 \right) + u_2 \left( \tilde{d}^{*}(\tau),\tilde{\theta}_2 \right) \right]
    	.
    \end{eqnarray*}
    By the definition of welfare, an impartial advisor thus chooses the welfare-maximizing test allocation in equilibrium, since the objectives are scaled versions of each other.

    Let us turn to the second part of lemma.
	Fix $\delta \in (0,1)$. Using Theorem \ref{thm:strategic_solution}, we can solve the discrimination model as a solution to a single-agent model with weights
	\begin{eqnarray} \label{app:alpha_1_tilde_p_d}
		\hat{\alpha}_1(p,\delta) &:=& 
		\begin{cases}
			\frac{1}{2} \sqrt{(1+\delta)(1-2p-\delta)} & \text{if}\ p<\frac{1-\delta}{2} \\ 
			0 & \text{otherwise} 
		\end{cases} \\ \label{app:alpha_2_tilde_p_d}
		\hat{\alpha}_2(p,\delta) &:=&
		\begin{cases}
			\frac{1}{2} \sqrt{(1-\delta)(1+2p+\delta)} & \text{if}\ p > - \frac{1+\delta}{2} \\
			0 & \text{otherwise}
		\end{cases}
	\end{eqnarray}
	Let $p=0$. Then $\hat{\alpha}_1(0,\delta) = \hat{\alpha}_2(0,\delta) = \frac{1}{2} \sqrt{1-\delta^2}$. The equilibrium test allocation depends on the weights only through their ratios $\frac{\hat{\alpha}_1(0,\delta)}{\hat{\alpha_1}(0,\delta)+\hat{\alpha}_2(0,\delta)} = \frac{1}{2}$ and $\frac{\hat{\alpha}_2(0,\delta)}{\hat{\alpha_1}(0,\delta)+\hat{\alpha}_2(0,\delta)} = \frac{1}{2}$. These ratios are independent of the level of discrimination $\delta$, which completes the proof.\qed 

	%=========================================
	\subsubsection{Proof of Lemma \ref{prop:discrimination_equalizing_bias}}
	%=========================================
We first define several new notions.  Given the politician's discrimination $\delta \in (0,1)$ and the advisor's partiality $p\geq 0$, let $\tau^*(p,\delta)=\left( \tau^*_1(p,\delta), \tau^*_2(p,\delta)  \right)$ denote the respective equilibrium test allocation. Note it holds $\tau^*_1(p,\delta)+\tau^*_2(p,\delta)=1$, i.e., the budget is always exhausted in equilibrium.\footnote{This result follows from the assumption that the sums of the weights of each player are equal: $\sum_k \alpha^{R}_k(p)=\sum_k \alpha^{D}_k(\delta)$. Thus, it cannot simultaneously hold $\alpha_k^R(p)\leq \alpha_k^{D}$ for both $k$ (which is necessary and sufficient for $\hat{\alpha}_1(p,\delta)=\hat{\alpha}_2(p,\delta)=0$), and there is thus always learning in equilibrium.}

Given a particular test allocation $\tau$ and equilibrium decision strategy, %$d^{*D}(s,\tau;\delta)$, 
let $V^k(\tau,\delta)$ denote the ex ante expected utility of group $k$, i.e.,
 \begin{align*}
		V^k(\tau,\delta) =&\ \mathbb{E}\left[ -(\tilde{d}^{*D}(\tau;\delta) - \tilde{\theta}_k)^2 \right].
\end{align*}

Under the additional constraint that the budget is exhausted, and so the test allocation satisfies $\tau_1=1-\tau_2$, let us define the difference between the expected utilities of group 1 from group 2 as
\begin{align}\label{eq:discrimination_Cap_Delta_fn}
		\Delta(\tau_2,\delta) := V^1((1-\tau_2,\tau_2),\delta) - V^2((1-\tau_2,\tau_2),\delta).
	\end{align}
 Note that inequality is then given by $I(p,\delta) = |\Delta(\tau_2^*(p,\delta),\delta) |$.

We now state two intermediary results (proved in Online Appendix \ref{OA:discrimination}), used to prove Lemma~\ref{prop:discrimination_equalizing_bias}.

 \begin{lemma} \label{lemma:discrimination_exists_partiality}
     For each $\delta \in (0,1)$, there exists a threshold partiality level $p^{aux}(\delta)\in \left(0,\frac{1-\delta}{2}\right)$ such that the equilibrium learning strategy satisfies $\tau^*_2(p,\delta)=1$ for advisors with $p \geq p^{aux}(\delta)$ and 
     \begin{align*}
    \tau^*_2(p,\delta)=\tau^{aux}_2(p,\delta):= \frac{2\hat{\alpha}_2(p,\delta) - \hat{\alpha}_1(p,\delta)}{\hat{\alpha}_1(p,\delta)+\hat{\alpha}_2(p,\delta)}\ \ \ \ \text{for advisors with }\ p \in [0,p^{aux}(\delta)],
     \end{align*}
     where $\hat{\alpha}_1(p,\delta)$ and $\hat{\alpha}_2(p,\delta)$ are given by equations \eqref{app:alpha_1_tilde_p_d} and \eqref{app:alpha_2_tilde_p_d}.
\end{lemma}

\begin{lemma}\label{lem:discrimination_restore_equality}
For each $\delta$ there exists a unique value of $\hat{\tau}_2(\delta) \in \left( 1/2, 1 \right)$ for which it holds $\Delta(\hat{\tau}_2(\delta);\delta)=0$. 
\end{lemma}

 Given the politician's level of discrimination $\delta$ and her equilibrium decision strategy, the value $\hat{\tau}_2(\delta)$ is the amount of tests allocated to learn about group 2 that guarantees that the expected utilities of both social groups are the same. From the proof of Lemma \ref{lem:discrimination_restore_equality}, we obtain  
 \begin{align*}
  \Delta(\tau_2,\delta) =  \frac{(1+\delta)(1-\tau_2)}{1+1-\tau_2} - \frac{(1-\delta)\tau_2}{1+\tau_2}.
 \end{align*}
 Solving $\Delta(\hat{\tau}_2(\delta),\delta)=0$ for $\hat{\tau}_2(\delta)\geq 0$, we obtain
	\begin{align} \label{eq:equality_restoring_tau2}
		\hat{\tau}_2(\delta)
		&=
		\frac{-(1-\delta)+ \sqrt{1+3\delta^2}}{2\delta}.
	\end{align}
	
	From Lemma~\ref{lemma:discrimination_exists_partiality}, the equilibrium test allocation satisfies: $\tau^*_2(p,\delta)=1/2$ at $p=0$, $\tau^*_2(p,\delta)=1$ at $p \geq p^{aux}(\delta)$, $\tau^*_2(p,\delta)$ is continuous on $p\in [0,1]$ and it is strictly increasing in $p \in [0,p^{aux}(\delta))$. Since $\hat{\tau}_2(\delta) \in (1/2,1)$, we thus have that for any level of discrimination $\delta \in (0,1)$ there exists a unique partiality level $\hat{p}(\delta) \in (0,p^{aux}(\delta))$ such that %\footnote{The constraint $\hat{b}(\delta) \in (0,1/2)$ is implied by the condition $\hat{b}(\delta)< (1-\delta)/2$. When $b\geq 1/2$, then there exists no discrimination $\delta \in (0,1)$ under which the constraint would be satisfied. When $b=0$, then $\tau^*_2(b,\delta)=1/2 \neq \hat{\tau}_2(\delta)$ as $\hat{\tau}_2(\delta) > 1/2$.}
	$\hat{\tau}_2(\delta) = \tau^*_2(\hat{p}(\delta),\delta)$.
	Solving this equation, we obtain
	\begin{eqnarray*}
		\hat{p}(\delta)
		&=& \delta \left( \frac{1}{\sqrt{1+3\delta^2}} - \frac{1}{2}  \right),
	\end{eqnarray*}
    which is continuous on $\delta\in (0,1)$. Taking the first derivative, we get
	\begin{eqnarray*}
		\hat{p}'(\delta)
		&=&
		\frac{1}{\left( 1+3\delta^2 \right)^{3/2}} - \frac{1}{2}\ \ 
		\begin{cases}
			>0 & \text{ if }\  \delta < \overline{\delta} \\
			=0 & \text{ if }\  \delta = \overline{\delta} \\
			<0 & \text{ if }\  \delta > \overline{\delta}
		\end{cases}
	\end{eqnarray*}
	where $\overline{\delta} := \sqrt{\frac{2^{2/3}-1}{3}}.$ %\frac{1}{3} \left( 1 - 2^{1/2} + 2^{2/3} \right)$.
 	\qed

%================================================
\subsubsection{Proof of Proposition~\ref{prop:discrimination_frontier}}
%=================================================
The first part of the statement follows directly from the following Lemma~\ref{lemma:pareto_frontier}, which is proved in Online Appendix \ref{OA:discrimination}.
\begin{lemma} \label{lemma:pareto_frontier}
Welfare $W(p,\delta)$ and inequality $I(p,\delta)$ are both (i) continuous in $p$; and (ii) strictly decreasing on $p \in (0,\hat{p}(\delta))$, where $\hat{p}(\delta)$ restores equality. Furthermore, welfare strictly decreases and inequality strictly increases on $p \in (\hat{p}(\delta),p^{aux}(\delta))$, where $p^{aux}(\delta)$ is characterized in Lemma \ref{lemma:discrimination_exists_partiality}; and they are both constant on $p \geq p^{aux}(\delta)$. 
\end{lemma}
Hence, the advisor with $\hat{p}(\delta)$ strictly improves both welfare and inequality outcomes compared to any advisor with $p>\hat{p}(\delta)$, and hence the advisors with $p>\hat{p}(\delta)$ do not constitute the Pareto frontier. 

For the second part of the statement, we build on the following Lemma~\ref{lemma:discr_tests_unchecked_equalizing}, which is proved in Online Appendix \ref{OA:discrimination}. For a given $\delta$, let $(1-\bar{\tau}_2(\delta),\bar{\tau}_2(\delta))$ denote the optimal learning strategy of the politician under unchecked discrimination, and let $(1-\hat{\tau}_2(\delta),\hat{\tau}_2(\delta))$, where $\hat{\tau}_2(\delta)$ is given by equation \eqref{eq:equality_restoring_tau2}, denote the equilibrium test allocation chosen by the advisor $\hat{p}(\delta)$.

\begin{lemma} \label{lemma:discr_tests_unchecked_equalizing}
Fix $\delta$. Then it holds that
\begin{align*}
    | \overline{\tau}_2(\delta) - 1/2|> |\hat{\tau}_2(\delta) - 1/2|. 
\end{align*}
\end{lemma}

Let
\begin{align*} 
	\omega(\tau_2,\delta) &:= V^1((1-\tau_2,\tau_2),\delta) + V^2((1-\tau_2,\tau_2),\delta) \\
	%&= -3 - \delta^2 + \frac{1/2(1+\delta^2) + (1+\delta)(1-\tau_2)}{2-\tau_2} + \frac{1/2(1+\delta^2) + (1-\delta)\tau_2}{1+\tau_2}
	&= -3 - \delta^2 + \frac{1/2(1+\delta)^2 + (1+\delta)(1-\tau_2)}{2-\tau_2} + \frac{1/2(1-\delta)^2 + (1-\delta)\tau_2}{1+\tau_2}
\end{align*}
denote the sum of the expected utilities of the two groups as a function of test allocation $\tau = (\tau_1,\tau_2)$ under the constraint that the budget is fully used, $\tau_1=1-\tau_2$, and where $V^k(\tau,\delta)$ is given by equation \eqref{eq:discrimination_group_exp_utility}. Since the budget is exhausted in equilibrium (as shown in the proof of Lemma~\ref{prop:discrimination_equalizing_bias}), welfare is thus given by $W(p,\delta) = \omega (\tau_2^*(p,\delta),\delta)$.

Observe the following symmetry property: $\omega(\tau_2,\delta)=\omega(1-\tau_2,\delta)$. Moreover, for any $\delta\in (0,1)$, the function $\omega(\tau_2,\delta)$ is strictly concave and maximized at $\tau_2=1/2$, when the test budget is split equally between the two groups. Therefore, the function $\omega(\tau_2,\delta)$ is strictly decreasing in $|\tau_2-1/2|$. Lemma~\ref{lemma:discr_tests_unchecked_equalizing} then implies that welfare with an advisor $\hat{p}(\delta)$ is strictly higher than welfare under unchecked discrimination.
\qed

	\bibliographystyle{abbrvnat}
	\bibliography{testalloc}

	\newpage
    
	%==============================
	%=============================
\setcounter{section}{1}
\renewcommand{\thesection}{\Alph{section}}
\renewcommand{\thesubsection}{\Alph{section}.\arabic{subsection}}
\setcounter{equation}{0}
\numberwithin{equation}{subsection}
\setcounter{proposition}{0}
\numberwithin{proposition}{subsection}
\setcounter{corollary}{0}
\numberwithin{corollary}{subsection}
\setcounter{page}{1}
%

 	%==============================
	%=============================
	{\centering \huge Online Appendix to\\ ``Strategic Attribute Learning''\\ \vspace{0.3cm}
    \Large Jean-Michel Benkert, Ludmila Matyskov\'{a}, Egor Starkov\\}
    \vspace{1cm}
	%===============================
	%================================
	\section{Extensions} \label{sec:online_app_B}
	%===============================
	%================================

	%==========================
	\subsection{Multiple researchers: media polarization} \label{sec:media}
	%==========================
	
	Our first extension explores a model of a media market, in which media polarization can emerge, but---in contrast to conventional wisdom---in a way that does not inherently reduce welfare. Specifically, we consider a voter with some priorities over policy issues and two media outlets with different priorities, which attempt to influence the voter's choice via their choice of coverage. We show that in a wide range of scenarios, the competition between media outlets polarizes their coverage in equilibrium. However, this polarization is advantageous for the voter, as it offers access to a broader and more diverse range of information. This diversity allows voter to make better-informed decisions compared to relying on a single, moderate outlet that is not fully aligned with the voter.
		
	We adopt Framework A from Section~\ref{sec:equivalent_models} and extend it by introducing two media outlets, $A$ and $B$, each aiming to influence a decision of voter $V$ that affects two policy issues $k=1,2$.\footnote{
		Our model features a single voter for simplicity. One can think of her as the median voter in a more general model with multiple voters who have different weights on the policy issues. Since the decision is one-dimensional and voters' preferences are single-peaked, standard logic dictates that the median voter (one with median weights $\alpha^i$) is pivotal and would thus be the one targeted by the media outlets aiming to influence the vote outcome.} 
	For $i=A,B,V$, player $i$'s utility is
	\begin{eqnarray} \label{eq:media_utility_defn}
		u^i(d,\theta) = - \sum_{k=1}^{2} \alpha^i_k \left( d - \theta_k \right)^2,
	\end{eqnarray}
	where $d \in \mathbb{R}$ is the voter's decision (the policy platform of the elected candidate), $\alpha^i_k>0$ are player $i$'s weights on respective policy issues $k=1,2$ with $\alpha^i_1+\alpha^i_2=1$, and $\theta_k$ is the optimal decision for issue $k$ (i.e., the platform of the candidate offering the best policy on issue $k$ given state of the world $\theta$).\footnote{In the model, all players agree, conditional on $\theta$, which candidate would enact an optimal policy on a given issue $k$. Tension arises because the single elected candidate will choose policies relevant to many issues, and preferences of the voter and the media outlets differ in terms of weights they assign to different issues.} Assume w.l.o.g. that media outlet $A$ assigns a higher weight to issue $k=1$ than outlet $B$: $\alpha^{A}_1 > \alpha^B_1$.
		
	Each media outlet $m=A,B$ chooses coverage ${q}^m=({q}^m_1,{q}^m_2) \in \mathbb{R}^2_+$ such that ${q}^m_1+{q}^m_2 = 1$, where ${q}^m_k$ describes the fraction of the coverage of media outlet $m$ devoted to issue $k$.\footnote{
		In our setting, the media cannot lie, withhold facts, distort information, or provide fake news. This assumption reflects the recent crackdown on fake news by companies and authorities \citep[see, e.g.,][]{facebook_misinformation,un_countering_disinformation}. Therefore, the strategic choice of which topics to cover remains one of the few options available to media editorial rooms.}
	The voter has a total budget of attention $T>0$ and chooses how to allocate it between the two outlets: $t = ({t}^A,{t}^B) \in \mathbb{R}^2_+$ such that ${t}^A + {t}^B \leq T$. 
	Given ${q}^m$ and ${t}^m$ for $m 
    =A,B$, for each issue $k$ the voter observes a realization of signal
	\begin{align*}
		\tilde{s}^m_k({q}^m_k,{t}^m) = \tilde{\theta}_k + \tilde{\varepsilon}^m_k,
		\text{ where }
		\tilde{\varepsilon}^m_k \sim \mathcal{N} \left(0,\frac{1}{{q}^m_k{t}^m} \right),
	\end{align*}
	and $\tilde{\varepsilon}^m_k$ is independent of all other random variables. Let $s=(s^A_1,s^A_2,s^B_1,s^B_2)$ denote the collection of signal realizations from both outlets on both issues, and $\tilde{s}$ its respective distribution.
	
	To curb the equilibrium multiplicity (though it cannot be avoided entirely, see the discussion below), assume that the voter chooses her attention allocation $t$ first. Media outlets then observe ${t}$ and simultaneously choose their coverage ${q}^m$. The rest of the game proceeds as in the baseline model: nature draws state realization $\theta$ and signal realizations $s$; the voter observes ${q}^m$ and $s$, updates her beliefs and takes decision $d$; afterwards, payoffs are realized. The equilibrium concept is analogous to that in the baseline model, and is formally stated in Section \ref
 {sec:media_proofs} below.\footnote{
		It is straightforward to show that an equilibrium of such a sequential-move game would be outcome-equivalent to an equilibrium of a simultaneous-move game, where the media choose $q$ contemporaneously with the voter's choice of ${t}$. However, the latter game may also admit other equilibria.} 

    We say that \emph{the voter achieves her best information in equilibrium} when the attention and coverages chosen in equilibrium jointly maximize the voter's ex ante expected payoff, given the voter's equilibrium decision strategy. We can state our main result of this section as follows.\footnote{The proofs of this extension can be found below.}
    
	\begin{proposition} \label{prop:media}
		If $\alpha^{A}_1 > \alpha^{V}_1 > \alpha^B_1$ and $T$ is large enough, then in the unique equilibrium the media outlets are polarized, ${q}^A=(1,0)$ and ${q}^B=(0,1)$, and the voter achieves her best information in equilibrium.
	\end{proposition}

	The proposition states that if media outlets are slanted in the opposite directions relative to the voter (and the voter's attention budget is large enough), even small differences in preferences between the voter and the media cause each outlet to focus exclusively on its more favoured issue (as compared to other players), neglecting the other. This occurs even if both media outlets care about both issues (i.e., $\alpha^m_k \in (0,1)$ for all $m,k$). The voter’s attention allocation and the coverage of each media jointly determine how much  information the voter receives on each policy issue. Once the voter chooses her attention, the media outlets engage in a tug-of-war, each seeking to increase the amount of information the voter receives on the outlet's preferred issue. This competition ultimately causes each outlet to cover exclusively a single issue. This polarization and the resulting diverse coverage allow the voter to achieve her best information in equilibrium.
 
	To set the benchmark for the result above, suppose there was only one media outlet. Then, we would be back in the setting of Section~\ref{sec:strategic_players}. In such a setting, the media outlet can shape the total amount of information the voter receives on both issues according to its own preferences and away from the voter's best information. Thus, the voter always weakly prefers a polarized media duopoly over a media monopoly, however moderate the latter may be. We capture this observation in the following statement, which follows from Theorem~\ref{thm:strategic_solution} and Proposition~\ref{prop:media}.\footnote{
		Both Proposition~\ref{prop:media} and Corollary~\ref{cor:media_dislike_monopoly} require that $\alpha^{A}_1 > \alpha^{V}_1 > \alpha^B_1$ and $T$ large enough. In Section \ref{sec:media_proofs}, we generalize the results to all $\alpha$ and $T$. While equilibrium multiplicity becomes an issue, we show that all equilibria are payoff-equivalent. In a representative equilibrium, at least one media outlet focuses exclusively on its relatively more preferred issue (while the other outlet either does the same or is ignored by the voter). Corollary~\ref{cor:media_dislike_monopoly} also continues to hold in general, though with weak preference instead of strict.}
	
	\begin{corollary} \label{cor:media_dislike_monopoly}
		If $\alpha^{A}_1 > \alpha^{V}_1 > \alpha^B_1$ and $T$ is large enough, then the voter strictly prefers the equilibrium when both media outlets $m =A,B$ are available to the equilibrium when only one outlet $m =A,B$ is present in the market.
	\end{corollary}
	
	Finally, we note that the results in this section rely on the assumption that the media outlets are directly interested in the voter's final decision. This is applicable in scenarios where editors have strong ideological stances, or policy-interested outlet owners can affect editorial policy (e.g., Fox News, Russia Today). We can also think of settings other than media markets in which similar outcomes could also arise, such as think tanks or advisors influencing policymakers.
	In contrast, independent media outlets may instead prioritize engagement and aim to attract attention from consumers in order to boost ad revenue or subscription fees. It is quite intuitive that if attention-seeking is the primary motivation for the media outlets, then they would give the voter the information she wants to see.\footnote{The classic Hotelling model of spatial competition suggests that both media outlets would likely converge on the median voter's ``position'' (offer the median voter's preferred information mix) even if there are many voters with heterogeneous preferences.} We would thus observe less media polarization with attention-seeking than with partisan media, but Proposition~\ref{prop:media} suggests that this difference is likely superficial and irrelevant for the actual voting decisions.

	%=========================================
	\subsubsection{Proofs for the extension on media polarization} \label{sec:media_proofs}
	%=========================================
	We begin by formally stating the equilibrium concept for the media model alongside two supplementary results. Lemma~\ref{lemma:media_tau_belief} then establishes that given $(t,q^A,q^B)$, the posterior distribution of $\tilde{\theta}|s$ can be related to that in the baseline model by considering \emph{aggregate attention allocation} $\tau$ given by $\tau_k := {t}^A {q}^A_k + {t}^B {q}^B_k$ for each $k$.
	Further, Lemma~\ref{lemma:media_tau_payoff} shows that the ex ante payoffs can be expressed as functions of $\tau$, and so the voter's attention allocation problem and the media outlets' coverage choice problems can be seen as a joint choice of $\tau$, with every player having some power over it. Then, we provide a general equilibrium characterization, with Proposition~\ref{prop:media2} covering the cases excluded in Proposition~\ref{prop:media}, and Corollary~\ref{cor:media_dislike_monopoly2} being a more general, albeit slightly weaker analog to Corollary~\ref{cor:media_dislike_monopoly}. We provide a unified proof for Propositions \ref{prop:media} and \ref{prop:media2} before proving Corollaries \ref{cor:media_dislike_monopoly} and \ref{cor:media_dislike_monopoly2}.
 
	Let $\mathcal{R} := \left\{ t \in \mathbb{R}^2_+ \mid t_1 + t_2 \leq T \right\}$ and $\mathcal{Q} := \left\{ q \in \mathbb{R}^2_+ \mid q_1 + q_2 \leq 1 \right\}$ be the domains of the players' strategies.

	\paragraph*{Equilibrium concept.} A weak Perfect Bayesian Equilibrium of the game introduced in Section \ref{sec:media} (referred to in the relevant proofs as ``equilibrium'') is a tuple
	\begin{align*}
		\left( {t}^*,\, {{q}^A}^*({t}),\, {{q}^B}^*({t}),\, d^* \left( s,  {t}, {q}^A, {q}^B \right),\, \tilde{\theta} \left( s,  {t}, {q}^A, {q}^B \right) \right)
	\end{align*}
	such that:
	\begin{enumerate}
		\item the coverage choice ${{q}^m}^*(t): \mathcal{R} \to \mathcal{Q}$ of each outlet $m =A,B$ maximizes its value \eqref{eq:media_value} given equilibrium strategies and belief updating rule $\left( {{q}^{j}}^*({t}),\, d^* \left( s,  {t}, {q}^A, {q}^B \right),\, \tilde{\theta} \left( s,  {t}, {q}^A, {q}^B \right) \right)$ for $j \neq m$ and realized attention choice ${t}$;
		
		\item the voter's attention allocation ${t}^* \in \mathcal{R}$ and decision strategy $d^* \left( s,  {t}, {q}^A, {q}^B \right): \mathbb{R}^2 \times \mathcal{R} \times \mathcal{Q}^2 \to \mathbb{R}$ maximize her value \eqref{eq:media_value} given equilibrium strategies and belief updating rule $\left( {{q}^A}^*({t}),\, {{q}^B}^*({t}),\, \tilde{\theta} \left( s,  {t}, {q}^A, {q}^B \right) \right)$; 
		
		\item the voter's posterior belief $\tilde{\theta} \left( s,  {t}, {q}^A, {q}^B \right): \mathbb{R}^2 \times \mathcal{R} \times \mathcal{Q}^2 \to \varDelta(\mathbb{R}^2)$ about $\theta$ is obtained via Bayes' rule given signal realizations $s$ and realized attention allocation and coverage choices $\left( {t}, {{q}^A}, {{q}^B} \right)$.
	\end{enumerate}
	Next, we show that the induced distribution of posterior beliefs in this extension coincides with that in the main model.
	\begin{lemma} \label{lemma:media_tau_belief}
		In the media model, the distribution of the voter's posterior belief $\tilde{\theta} \left( \tilde{s},{t},{q}^A, {q}^B \right)$ is the same as the one induced by \eqref{eq:signal_defn} and \eqref{eq:posterior_attributes}--\eqref{eq:posterior_var_attributes}, with $\tau = \tau \left( {t}, {q}^A, {q}^B \right)$ being the aggregate attention allocation defined as $\tau = (\tau_1,\tau_2)$ with $\tau_k := {t}^A {q}^A_k + {t}^B {q}^B_k$ for each $k$.
	\end{lemma}
	
	\begin{proof}
		Throughout the proof, we fix some attention allocation and coverages $({t}, {q}^A, {q}^B)$ and the corresponding aggregate attention allocation $\tau_k = {t}^A {q}^A_k + {t}^B {q}^B_k$, $k=1,2$. For the sake of brevity, we will omit $({t}, {q}^A, {q}^B)$ and $\tau$ from the set of conditioning variables. 
		
		The collection $({t}, {q}^A, {q}^B)$ induces a distribution of the signal vector $\tilde{s}$. In turn, each signal realization is associated with a posterior belief $\tilde{\theta}(s)$. Ex ante, therefore, the collection $(t, q^A,q^B)$ induces a distribution over posterior beliefs, $\tilde{\theta}(\tilde{s}) \in \varDelta \left( \varDelta (\mathbb{R}^2) \right)$.
		Our goal is to establish that, given the aggregate attention allocation $\tau$ associated with the collection $(t,q^A,q^B)$, the induced distribution coincides with the one described in Section \ref{sec:model} induced by the same $\tau$, where \eqref{eq:posterior_attributes}--\eqref{eq:posterior_var_attributes} describe the posterior distribution conditional on $s$, and \eqref{eq:signal_defn} implies that $\tilde{s}_k \sim \mathcal{N} \left( \mu^0_k, \Sigma^0_k + \frac{1}{\tau_k} \right)$ for all $k$.
		
		Since $\tilde{\theta}_1 \perp \tilde{\theta}_2$ and $\tilde{\theta}_k \perp \tilde{\varepsilon}_j^m$ for all $j,k =1,2$ and $m =A,B$, the posterior belief about each attribute is independent of everything related to the other attribute. Hence, we fix some $k=1,2$ and work with the (marginal) distributions of $\tilde{\theta}_k$ and $\tilde{s}^m_k$, which is sufficient to prove the result.
		
		We first characterize $\tilde{\theta}_k(s)$, the voter's posterior belief about $\theta_k$ given some signal realizations $s^A_k, s^B_k$.
		It is immediate to verify directly from normal p.d.f.s that if $\tilde{X} \sim \mathcal{N}(\mu,\sigma^2_X)$, $\tilde{\varepsilon} \sim \mathcal{N}(0,\sigma^2_\varepsilon)$ and $Y = X + \varepsilon$, then $\tilde{X}|Y \sim \mathcal{N} \left( \hat{\mu}, \hat{\sigma}^2 \right)$, where 
		\begin{align*}
			\hat{\mu} &= \frac{ \frac{1}{\sigma^2_X} }{ \frac{1}{\sigma^2_X} + \frac{1}{\sigma^2_\varepsilon} } \mu + \frac{ \frac{1}{\sigma^2_\varepsilon} }{ \frac{1}{\sigma^2_X} + \frac{1}{\sigma^2_\varepsilon} } Y, 
			&
			\hat{\sigma}^2 &= \frac{ 1 }{ \frac{1}{\sigma^2_X} + \frac{1}{\sigma^2_\varepsilon} }.
		\end{align*}
		This directly implies characterization \eqref{eq:posterior_attributes}--\eqref{eq:posterior_var_attributes} in Section \ref{sec:prelim}. In the context of this proof, this gives first that $\tilde{\theta}_k | s^A_k \sim \mathcal{N} (\hat{\mu}_{k,A}, \hat{\Sigma}_{k,A})$ with
		\begin{align*}
			\hat{\mu}_{k,A} &= \frac{ \frac{1}{\Sigma^0_k} }{ \frac{1}{\Sigma^0_k} + t^A q^A_k } \mu^0_k + \frac{ t^A q^A_k }{ \frac{1}{\Sigma^0_k} + t^A q^A_k } s^A_k, 
			&
			\hat{\Sigma}_{k,A} &= \frac{ 1 }{ \frac{1}{\Sigma^0_k} + t^A q^A_k },
		\end{align*}
		and then, subsequently, that $(\tilde{\theta}_k | s^A_k) | s^B_k \sim \mathcal{N} (\hat{\mu}_k, \hat{\Sigma}_k)$, where
		\begin{align}
			\hat{\mu}_{k} &= \frac{ \frac{1}{\hat{\Sigma}_{k,A}} }{ \frac{1}{\hat{\Sigma}_{k,A}} + t^B q^B_k } \hat{\mu}_{k,A} + \frac{ t^B q^B_k }{ \frac{1}{\hat{\Sigma}_{k,A}} + t^B q^B_k } s^B_k \notag
			= \frac{ \frac{1}{\Sigma^0_k} \mu_k^0 + t^A q^A_k s^A_k + t^B q^B_k s^B_k }{ \frac{1}{\Sigma^0_{k}} + t^A q^A_k + t^B q^B_k }\notag
			\\
			&= \frac{ \frac{1}{\Sigma^0_k} }{ \frac{1}{\Sigma^0_{k}} + \tau_k } \mu^0_k + \frac{1}{ \frac{1}{\Sigma^0_{k}} + \tau_k } \left( t^A q^A_k s^A_k + t^B q^B_k s^B_k \right),\label{eq:fict_rv}
			\\
			\hat{\Sigma}_{k} &= \frac{ 1 }{ \frac{1}{\hat{\Sigma}_{k,A}} + t^B q^B_k }
			= \frac{ 1 }{ \frac{1}{\Sigma^0_{k}} + t^A q^A_k + t^B q^B_k } 
			= \frac{ 1 }{ \frac{1}{\Sigma^0_{k}} + \tau_k }.\label{eq:sigma_hat_modified}
		\end{align}
		Comparing \eqref{eq:sigma_hat_modified} with its counterpart \eqref{eq:posterior_var_attributes} in the baseline model, we note that they coincide.
		Next, compare \eqref{eq:fict_rv} with its counterpart \eqref{eq:posterior_mean_attributes} in the baseline model. In order to conclude that the distribution of posteriors is the same as in the baseline model, we need to show that the random variable defined as
		\begin{align*}
			\tilde{s}^f_k &:= 
			\frac{1}{\tau_k} \left( t^A q^A_k \tilde{s}^A_k + t^B q^B_k \tilde{s}^B_k \right) 
				= \tilde{\theta}_k + \frac{1}{\tau_k} \left( t^A q^A_k \tilde{\varepsilon}^A_k + t^B q^B_k \tilde{\varepsilon}^B_k \right)
		\end{align*}
		has the same distribution as $\tilde{s}_k$ from the baseline model.
		Note that $\tilde{s}^f_k|\theta_k \sim \mathcal{N} \left( \theta_k,\frac{1}{\tau_k} \right)$ since $\mathbb{E}[\tilde{\varepsilon}^m_k]=0$ and $var\left( \frac{t^m q^m_k}{\tau_k} \tilde{\varepsilon}^m_k \right) = \frac{1}{\tau_k}$, and then $\tilde{s}^f_k \sim \mathcal{N} \left( \mu^0_k,\Sigma^0_k + \frac{1}{\tau_k} \right)$. For a given $\tau$, these coincide with the distributions of $\tilde{s}_k | \theta_k$ and $\tilde{s}_k$ from the baseline model, respectively.
		Therefore, the distribution of $\hat{\mu}_k$ above coincides with that of  \eqref{eq:posterior_attributes}, so the distribution of $\tilde{\theta}_k({s}^f_k)$ in the media model coincides with the distribution of $\tilde{\theta}_k(s_k)$ in the baseline model when ${s}^f_k = s_k$.
		Together, the two latter facts imply that the distribution of $\tilde{\theta}_k(\tilde{s}^f_k)$ in the media model coincides with the distribution of $\tilde{\theta}_k(\tilde{s}_k)$ in the baseline model.
	\end{proof}
	
	Our next lemma below show that all players' respective ex ante expected utilities only depend on $({t},{q}^A,{q}^B)$ through $\tau = \tau({t},{q}^A,{q}^B)$. Specifically, define player $i$'s value function as
	\begin{align} \label{eq:media_value}
		V^i(\tau) := \mathbb{E} \left[ - \sum_{k=1}^{2} \alpha^i_k \left( d^*(\tilde{s}^f,\tau) - \tilde{\theta}_k(\tilde{s}^f,\tau) \right)^2 \right]
	\end{align}
	for $i=A,B,V$, where $d^*(\tilde{s}^f,\tau)$ is the voter's equilibrium decision strategy given $\tau = \tau(t,q^A,q^B)$, and ${s}^f = ({s}^f_1, {s}^f_2)$ is a vector of fictitious ``aggregate'' signals introduced in the proof of Lemma \ref{lemma:media_tau_belief}.
	%We further let 
	%\begin{align*}
	%	{\tau^i}^* := \arg \max_{\tau \in \mathcal{T}} V^i(\tau)
	%\end{align*}
	%denote the \emph{optimal} aggregate attention allocation for player $i$.
	%Note that ${\tau^i}^*$ is characterized by Theorem~\ref{thm:strategic_solution} for all players, with the voter playing the role of the decision-maker. Specifically, the voter's optimum ${\tau^V}^*$ coincides with her single-player optimum characterized in Corollary~\ref{cor:binary_solution_single_player}.
	
	\begin{lemma} \label{lemma:media_tau_payoff}
		In the media model, given $({t},{q}^A,{q}^B)$, the expectations of all players' utilities \eqref{eq:media_utility_defn} are given by values \eqref{eq:media_value} and only depend on $\tau \left( {t},{q}^A,{q}^B \right)$: 
		\begin{align*}
			\mathbb{E} \left[ u^i \left( d^* \left( \tilde{s}, {t}, {q}^A, {q}^B \right),\, \tilde{\theta} \left( \tilde{s}, {t}, {q}^A, {q}^B \right) \right) \right]
			= V^i \left( \tau \left( {t},{q}^A,{q}^B \right) \right)
		\end{align*}
	\end{lemma}

	\begin{proof}
		Rewriting the expected utility (the left-hand side of the equality in the lemma) using definition \eqref{eq:media_utility_defn}, we get
		\begin{align} \label{eq:media_util_expanded1}
			\mathbb{E} \left[ - \sum_{k=1}^{2} \alpha^i_k \left( d^* \left( \tilde{s}, {t}, {q}^A, {q}^B \right) - \tilde{\theta}_k \left( \tilde{s}, {t}, {q}^A, {q}^B \right) \right)^2 \right].
		\end{align}
		The voter's equilibrium decision for any signal realization $s$ is given by
		\begin{align*}
			d^* \left( s, {t}, {q}^A, {q}^B \right) 
			%&= \mathbb{E} \left[ \sum_{k=1}^2 \alpha^V_k \tilde{\theta}_k \mid s, {t}, {q}^A, {q}^B \right]
			%\\
			&= \sum_{k=1}^2 \alpha^V_k \mathbb{E} \left[ \tilde{\theta}_k \left(s, {t}, {q}^A, {q}^B \right) \right].
		\end{align*} 
		By Lemma~\ref{lemma:media_tau_belief} we know that $\tilde{\theta}_k \left(\tilde{s}, {t}, {q}^A, {q}^B \right) \sim \tilde{\theta}_k \left(\tilde{s}^f, \tau \right)$ (the two have the same distributions), where $\tilde{s}^f = (\tilde{s}^f_1, \tilde{s}^f_2)$ with $\tilde{s}^f_k = \frac{1}{\tau_k} \left( t^A q^A_k \tilde{s}^A_k + t^B q^B_k \tilde{s}^B_k \right)$. This implies that $d^* \left( \tilde{s}, {t}, {q}^A, {q}^B \right) \sim d^*(\tilde{s}^f,\tau)$.
		Then we can write \eqref{eq:media_util_expanded1} as
			$$\mathbb{E} \left[ - \sum_{k=1}^{2} \alpha^i_k \left( d^*(\tilde{s}^f,\tau) - \tilde{\theta}_k \left( \tilde{s}^f, \tau \right) \right)^2 \right],$$
		which is exactly \eqref{eq:media_value}, the definition of $V^i \left( \tau \right)$.
	\end{proof}

	We now make a statement about the media model equilibria in cases not covered in Proposition \ref{prop:media}, and prove both statements simultaneously.
	
	\begin{proposition} \label{prop:media2}
		In each of the cases listed below, all equilibria are payoff-equivalent to the following:
		\begin{enumerate}
			%\item if $\alpha^{A}_1 \geq \alpha^{V}_1 \geq \alpha^B_1$ and $T$ is large enough: media outlets are polarized, ${q}^A=(1,0)$ and ${q}^B=(0,1)$, and the voter achieves her optimal aggregate attention allocation ${\tau^V}^*$;
			%
			\item if $\alpha^{A}_1 > \alpha^{V}_1 > \alpha^B_1$ and $T$ is not large enough: both media outlets choose the same extreme coverage, either ${q}^A = {q}^B =(1,0)$, or ${q}^A = {q}^B=(0,1)$, and the voter achieves her best information;
			
			\item if $\alpha^{V}_1 \geq \alpha^{A}_1$: the voter only follows outlet $A$, $t=(T,0)$, and outlet $A$ achieves its best information;\footnote{We define outlet $m$'s best information by analogy with the voter's best information in the text, i.e., $(t,q^A,q^B)$ must jointly maximize $m$'s expected payoff given the voter's equilibrium decision strategy.}
			
			\item if $\alpha^B_1 \geq \alpha^{V}_1$: the voter only follows outlet $B$, $t=(0,T)$, and outlet $B$ achieves its best information.
		\end{enumerate}
	\end{proposition}

	\begin{proof}[Proof of Propositions \ref{prop:media} and \ref{prop:media2}.]
	%In this proof, we characterize the equilibria of the media game for all parameter values. This proves Proposition~\ref{prop:media} from the main text and Proposition~\ref{prop:media2} stated above.
	%
	Lemmas~\ref{lemma:media_tau_belief} and \ref{lemma:media_tau_payoff} imply that the media setting with four signals $s=(s^A_1,s^A_2,s^B_1,s^B_2)$ can, without loss, be replaced by the baseline model with fictitious aggregate signals ${s}^f = ({s}^f_1, {s}^f_2)$, which have the same distribution conditional on aggregate attention allocation $\tau \left( {t}, {q}^A, {q}^B \right)$ as the signals in our main model of Section~\ref{sec:model}. For brevity, we often drop the arguments of the aggregate attention allocation and write it as simply $\tau$.
	
	Proceeding by backwards induction, we first note that the voter's equilibrium decision $d^* \left( {s}^f, \tau \right)$ is analogous to the single-player and strategic versions of our model, see the proof of Lemma~\ref{lemma:media_tau_payoff}. In the remainder of this proof, we assume the voter follows this equilibrium decision strategy $d^* \left( {s}^f, \tau \right)$.
	
	Proceeding further, Lemma~\ref{lemma:media_tau_payoff} implies that the outlets' problem of choosing coverage ${q}^m$ and the voter's problem of choosing attention allocation ${t}$ reduce to the problem of maximizing value function \eqref{eq:media_value} over $\tau \left( {t}, {q}^A, {q}^B \right)$.
	Fix some $t$ and consider the outlets' problem. Let $T_V(t) := t_1 + t_2$ denote the total amount of attention devoted to the two outlets by the voter. Let ${\tau^i}^*(t) := \arg \max_{\tau \in \mathbb{R}^2_+} V^i(\tau)$ s.t. $\tau_1+\tau_2 \leq T_V(t)$ for all $i =A,B,V$.
	Note that player $i$'s best information, by definition, is given by $\max_{t \in \mathcal{R}} {\tau^i}^*(t)$.
	For both $m =A,B$, ${\tau^m}^*(t)$ is then described by Lemma \ref{lemma:strategic_player_maximization} and Theorem \ref{thm:strategic_solution} with $T=T_V(t)$. By Lemma \ref{lemma:strategic_player_maximization}, if $m$'s expected payoff \eqref{eq:strategic_player_maximization} is weakly increasing in $\tau_k$, it is also strictly concave in $\tau_k$.\footnote{
		This is most evident from \eqref{eq:foc_strategic} in the proof of Theorem~\ref{thm:strategic_solution}: one can readily see by analogy that both monotonicity and concavity in this problem depend on whether $\lambda^i_k := \alpha^{V}_k \left( 2 \alpha^m_k - \alpha^{V}_k \right) \gtrless 0$.} 
	From Theorems~\ref{thm:benchmark} and \ref{thm:strategic_solution}, we obtain a closed-form solution for ${\tau^i}^*(t)$:
	\begin{equation} \label{eq:strat_taustar_K2}
		\begin{aligned}
			{\tau^i_1}^*(t) 
			&= \max \left\{ 0,\  \min \left\{ \frac{\hat{\alpha}^i_1 \Sigma^0_{1} - \hat{\alpha}^i_2 \Sigma^0_{2}}{\Sigma^0_{1}\Sigma^0_{2}(\hat{\alpha}^i_1+\hat{\alpha}^i_2)} + 
			\frac{\hat{\alpha}^i_1}{\hat{\alpha}^i_1+\hat{\alpha}^i_2} T_V(t),\ T_V(t) \right\}    \right\}, \\
			{\tau^i_2}^*(t)
			&= T_V(t) - {\tau^i_1}^*(t),
		\end{aligned}
	\end{equation}
	where $\hat{\alpha}^i_k := \sqrt{ \max\left\{ 0, (\alpha^i_k)^2 - (\alpha^i_k - \alpha^V_k)^2 \right\} }$ for all $k =1,2$ and $i =A,B,V$ (note that $\hat{\alpha}^V_k = \alpha^V_k$, so for $i=V$, \eqref{eq:strat_taustar_K2} coincides with the solution in Theorem \ref{thm:benchmark}). Note that ${\tau^i_1}^*(t)$ is weakly increasing in $\alpha^i_1$ (subject to the $\alpha^i_1+\alpha^i_2 = 1$ constraint).
	This characterization implies that since $\alpha^i_1 + \alpha^i_2 = 1$ for all $i =A,B,V$, the auxiliary weights $\hat{\alpha}^m_k$ cannot simultaneously be zero for both $k =1,2$ for any $m =A,B$, so ${\tau^m_1}^*(t) + {\tau^m_2}^*(t) = T_V(t)$ for all $t$ (the media want no attention wasted). 
	We note that subject to this constraint, value $V^i(\tau_1, T_V(t) - \tau_1)$ is either strictly monotone, or strictly concave in $\tau_1$ for all $i =A,B,V$.
	
	The characterization in Theorem~\ref{thm:strategic_solution} also implies that given $t$, if $\hat{\alpha}^m_1 \Sigma^0_1 > \hat{\alpha}^m_2 \Sigma^0_2$ and $T_V(t) \leq \frac{\hat{\alpha}^m_1}{\hat{\alpha}^m_2 \Sigma^0_2} - \frac{1}{\Sigma^0_1}$ (or for all $T_V(t)$ if $\hat{\alpha}_2 = 0$), then ${\tau^m}^*(t) = \left( T_V(t), 0 \right)$, and vice versa: if $\hat{\alpha}^m_2 \Sigma^0_2 > \hat{\alpha}^m_1 \Sigma^0_1$ and $T_V(t) \leq \frac{\hat{\alpha}^m_2}{\hat{\alpha}^m_1 \Sigma^0_1} - \frac{1}{\Sigma^0_2}$ (or for all $T_V(t)$ if $\hat{\alpha}_1 = 0$), then ${\tau^m}^*(t) = \left( 0, T_V(t) \right)$. If none of these conditions apply, then ${\tau^m}^*(t)$ is interior (and strictly monotone in ${\alpha}^m_1$ subject to the $\alpha^m_1 + \alpha^m_2=1$ constraint). Therefore, if we let
	\begin{align*}
		\hat{T} := \max \left\{ 0, \min_{m \in \{A,B\}} \left\{ \frac{\hat{\alpha}^m_1}{\hat{\alpha}^m_2 \Sigma^0_2} - \frac{1}{\Sigma^0_1} \right\}, \min_{m \in \{A,B\}} \left\{ \frac{\hat{\alpha}^m_2}{\hat{\alpha}^m_1 \Sigma^0_1} - \frac{1}{\Sigma^0_2} \right\} \right\},
	\end{align*}
	then whenever $T_V(t) \leq \hat{T}$, both media outlets optimally prefer the same extreme coverage: ${\tau^A_1}^*(t) = {\tau^B_1}^*(t) \in \{ 0, T_V(t) \}$. Note that if $\hat{\alpha}^A_1 \Sigma^0_1 > \hat{\alpha}^A_2 \Sigma^0_2$ and $\hat{\alpha}^B_1 \Sigma^0_1 < \hat{\alpha}^B_2 \Sigma^0_2$, then $\hat{T} = 0$, and such a preference profile (${\tau^A_1}^*(t) = {\tau^B_1}^*(t) \in \{ 0, T_V(t) \}$) cannot arise. Further, if $T_V(t) > \hat{T}$, then ${\tau^m}^*(t)$ is interior for at least one outlet $m$, and since $\alpha^{A}_1 > \alpha^B_1$ (and $\alpha^{A}_2 < \alpha^B_2$), we then have ${\tau^A_1}^*(t) > {\tau^B_1}^*(t)$.
	
	Next, we claim that if $\tau \neq {\tau^m}^*(t)$ for some outlet $m$ in equilibrium, then $m$ is either polarized (i.e., ${q}^m \in \{((1,0), (0,1)\}$), or ignored (${t}^m = 0$).
	To see this, proceed by contradiction and consider a strategy profile $\left( {t}, {q}^A, {q}^B \right)$ such that $\tau_k \left( {t}, {q}^A, {q}^B \right) < {\tau^m_k}^*(t)$, ${q}^m_k < 1$, and $t^m > 0$ for some $k =1,2$ and $m =A,B$. Then increasing ${q}^m_k$ (while decreasing the other weight ${q}^m_{-k}$) brings the aggregate test allocation $\tau$ closer to $m$'s optimum ${\tau^m}^*(t)$, which strictly increases $m$'s expected payoff \eqref{eq:strategic_player_maximization}.\footnote{
		\label{foot:closeness}
		``Closer'' here is used in the sense of the new aggregate attention allocation $\bar{\tau}$ being a convex combination of the original $\tau$ and the optimal ${\tau^m}^*(t)$. The payoff function being concave on the relevant interval, and ${\tau^m}^*(t)$ being its maximizer then directly imply that $\bar{\tau}$ yields a higher payoff than the original $\tau$.} 
	This is a profitable deviation for outlet $m$, and the original strategy profile thus cannot be an equilibrium. The same is true of the case when $\tau_k > {\tau^m_k}^*(t)$ and ${q}^m_k > 0$ and ${t}^m > 0$ for some $k=1,2$, $m =A,B$.
	
	The logic above also implies that as $\alpha^{A}_1 > \alpha^B_1$, there is no equilibrium with $\tau$ such that $\tau_1 > {\tau^A_1}^*(t)$ or ${\tau^B_1}^*(t) > \tau_1$, since in the former case both outlets $m$ want to decrease ${q}^m_1$ and increase ${q}^m_2$ (and at least one has the scope for such a deviation), and in the latter they want to do the opposite. Therefore, in equilibrium, the aggregate attention allocation $\tau = \tau \left( {t}^*, {{q}^A}^*, {{q}^B}^* \right)$ must be such that $\tau_1 \in \left[ {\tau^B_1}^*(t^*), {\tau^A_1}^*(t^*) \right]$ and $\tau_2 = T_V(t^*) - \tau_1 \in \left[ {\tau^A_2}^*(t^*), {\tau^B_2}^*(t^*) \right]$. 
	
	It remains to consider the voter's attention allocation problem. Since ${\tau^m}^*(t)$ only depend on $t$ through $T_V(t)$, the voter's problem can be split into first choosing the total amount of attention to devote to media, $T_V$, and then choosing how to split it by way of choosing $\tau$ such that $\tau_1 \in \left[ {\tau^B_1}^*(T_V), {\tau^A_1}^*(T_V) \right]$ and $\tau_2 = T_V - \tau_1$.
	Note that in the second step, all aggregate allocations in this interval are available to the voter for a given $T_V$. By setting ${t}^m=T_V$ for some $m$ and ${t}^{-m}=0$, she can induce $\tau = {\tau^m}^*(T_V)$. %\footnote{If $\tau = {\tau^m}^*(T_V)$ happens to be the equilibrium outcome, then the equilibrium may not be unique, since there exist other strategy profiles $\left({t}, {q}^A, {q}^B \right)$ that can induce $\tau = {\tau^m}^*$, and some of them can be sustained in equilibrium.}
	Any allocation $\tau$ in the interior of the interval can be achieved by setting ${t}^A = \tau_1$ and ${t}^B = \tau_2$, in which case the unique best response for outlet $A$ is ${q}^A = (1,0)$, and the unique best response for outlet $B$ is ${q}^B = (0,1)$, as argued above. We shall thus refer to $\tau$ that satisfy these requirements as \emph{feasible} given some $T_V$.
	
	The voter never wants any attention wasted, since her payoff is strictly increasing in $\tau_k$ for both $k$. From \eqref{eq:strat_taustar_K2} we can see that $\frac{d}{dT_V} {\tau^m_k}^*(T_V) \in [0,1]$ for all $m,k$, hence for any $T_V$, any $\tau^a$ that is feasible given $T_V$, and any $\varepsilon > 0$, there exists $\tau^b$ that is feasible given $T_V + \varepsilon$ and is such that $\tau^b_k \geq \tau^a_k$ for both $k$, with at least one inequality being strict.\footnote{A way to see this is to notice that an increase in $T_V$ implies an increase in the strong set order of the intervals of feasible in $\tau_1$ and $\tau_2$.} Therefore, increasing $T_V$ always offers a strict improvement the voter, so in equilibrium, $T_V(t^*) = T$.
	
	We now move on to the second stage of the voter's problem. If $T \leq \hat{T}$, then ${\tau^A_1}^*(T) = {\tau^B_1}^*(T) \in \{0,1\}$, so the set of available $\tau$ is a singleton, and the voter is then indifferent between all attention allocations.
	Suppose now that $T > \hat{T}$.
	If $\alpha^{A}_1 > \alpha^{V}_1 > \alpha^B_1$, then ${\tau^V_1}^*(T) \in \left[ {\tau^B_1}^*(T), {\tau^A_1}^*(T) \right]$, so the voter can attain her optimum ${\tau^V}^*(T)$ by following the strategy described above. It is immediate from \eqref{eq:strat_taustar_K2} that in this case there exists $\hat{T}_V \geq \hat{T}$ such that ${\tau^V}^*(T)$ is interior if and only if $T > \hat{T}_V$. For $T > \hat{T}_V$, ${\tau^V}^*(T)$ is uniquely optimal for the voter, hence ${t} = ({\tau^V_1}^*(T), {\tau^V_2}^*(T))$, ${q}^A = (1,0)$, and ${q}^B = (0,1)$ is the unique equilibrium, proving Proposition~\ref{prop:media}. For $T \leq \hat{T}_V$, as argued above, the voter can still achieve ${\tau^V}^*(T)$, but the equilibrium may not be unique. This proves case 1 of Proposition~\ref{prop:media2}. 
	In case $\alpha^{V}_1 \notin \left( \alpha^{A}_1, \alpha^{B}_1 \right)$, since the voter's value function \eqref{eq:single_player_maximization2} is strictly concave in $\tau$, she will optimally choose the available allocation closest (in the sense of Footnote \ref{foot:closeness}) to ${\tau^V}^*(T)$. In particular, if $\alpha^B_1 \geq \alpha^{V}_1$, then ${\tau^V_1}^*(T) \leq {\tau^B_1}^*(T)$, so the voter will choose ${t} = (0,T)$, resulting in $\tau = \tau^*_{B}(T)$. And if $\alpha^V_1 \geq \alpha^A_1$, then ${\tau^V_1}^*(T) \geq {\tau^A_1}^*(T)$, and the voter will choose $t = (T,0)$, resulting in $\tau = \tau^*_{A}(T)$. This proves parts 2 and 3 of Proposition~\ref{prop:media2} and completes this proof.
	\end{proof}

 	Finally, we state the following general corollary before proving it along with Corollary \ref{cor:media_dislike_monopoly}.
	\begin{corollary} \label{cor:media_dislike_monopoly2}
		The voter weakly prefers the equilibrium when both media outlets $m =A,B$ are available to the equilibrium when only one outlet $m =A,B$ is present in the market.
	\end{corollary}

	\begin{proof}[Proof of Corollaries \ref{cor:media_dislike_monopoly} and \ref{cor:media_dislike_monopoly2}.]
    	
	In the monopoly case, if only outlet $m$ is available, then the voter always lends it her full attention, ${t}^m = T$, since her value \eqref{eq:media_value} reduces to \eqref{eq:single_player_maximization2}, which is strictly increasing and strictly concave in $\tau$. Therefore,  $m$ can choose aggregate attention allocation $\tau$ freely subject to $\tau_1 + \tau_2 = T$. Since $\alpha^i_1+\alpha^i_2 = 1$ for all $i$, Theorem \ref{thm:strategic_solution} implies that $m$ does not want to waste any of the voter's attention. The unique equilibrium aggregate attention allocation is, therefore, ${\tau^m}^*(T)$, as defined in the proof of Propositions \ref{prop:media} and \ref{prop:media2}, which is $m$'s best information.
	
	Propositions \ref{prop:media} and \ref{prop:media2} above shows that under media duopoly, the equilibrium aggregate attention allocation $\tau^*$ is given by ${\tau^V}^*(T)$ if $\alpha^{A}_1 > \alpha^{V}_1 > \alpha^B_1$, by ${\tau^A}^*(T)$ if $\alpha^{A}_1 \leq \alpha^{V}_1$, and by ${\tau^B}^*(T)$ if $\alpha^{V}_1 \leq \alpha^B_1$. Note that in the two latter cases, it is the voter's preferred media outlet that has the monopoly power, which proves Corollary \ref{cor:media_dislike_monopoly2} for those cases. 
	In the former case, ${\tau^V}^*(T)$ being its maximizer of the voter's value function imply that the voter prefers duopoly, concluding the proof of Corollary \ref{cor:media_dislike_monopoly2}. If ${\tau^V}^*(T) \in \left( {\tau^B}^*(T), {\tau^A}^*(T) \right)$, which is the case if and only if $T > \hat{T}_V$, as defined in the proof of Propositions \ref{prop:media} and \ref{prop:media2}, then it is also the unique maximizer (since the value function is strictly concave on that interval in that case), which proves Corollary \ref{cor:media_dislike_monopoly}.
	\end{proof}

	%=============================
	\subsection{Changing Preferences} \label{sec:multiple_selves}
	%===============================
	
	In this second extension, we study how a (potential) change in an agent's preferences between learning and decision-making affects his learning strategy. Following the literature on time-inconsistent preferences \citep[e.g.,][]{odonoghue1999doing}, we consider sophisticated and naive agents and investigate the impact of such a change in preferences on the learning strategy and on the agent's welfare.
	
	We model this similarly to our baseline model by considering a single agent facing a two-stage problem. In the first stage, the agent acquires information about unknown attributes. In the second stage, the agent takes a decision based on the acquired information. The agent's preference may change between the two stages. For instance, this change may occur due to a change in the agent's circumstances, such as losing the job or falling ill. Alternatively, the change in preferences may be due to the agent's self-control problems, such as succumbing to temptation.

	In the first stage, when acquiring information, the agent has utility function 
	\begin{align}\label{eq:first_stage}
		u^R(d,\theta) = - \left( d - \alpha^R_1 \theta_1 - \alpha^R_2 \theta_2 \right)^2, 
	\end{align}
	where we assume $\alpha^{R}_k>0$ for $k=1,2$. Afterwards, the agent's preference may change, which we model as a change in the weight of attribute $\tilde{\theta}_2$ occurring with (ex ante) probability $p\in (0,1)$. In the second stage, when making the decision, the agent has utility function
	\begin{align*}
		u^{D}(d,\theta) = \begin{cases}
			-\left( d - \alpha^R_1 \theta_1 - \alpha^R_2 \theta_2 \right)^2 & \text{with probability } 1-p, \\
			-\left( d - \alpha^R_1\theta_1 - c\alpha^R_2 \theta_2 \right)^2 & \text{with probability } p, 
		\end{cases}
	\end{align*}
	where $c> 0$ and $c\neq 1$.\footnote{This setting extends the baseline model by allowing for uncertainty about the decision-maker's preferences. An alternative interpretation is that there are multiple potential decision-makers, and the researcher is uncertain about who will make the decision. We extend our equilibrium concept to this setting in Section~\ref{appendix:multiple_selves}.} The naive agent believes her preferences will stay the same across both stages, while the sophisticated agent understands they might change. Throughout this section, we simplify the analysis by assuming $\mu^0 = (0,0)$, so that, absent any new information, the researcher and both types of decision-makers agree on the optimal decision $d^\ast=0$. Additionally, we fix $\Sigma_1^0 = \Sigma_2^0= T=1$. 

	As the naive agent incorrectly assumes no possibility of change, the naif's learning problem coincides with that of a single player in Section \ref{sec:single_player}, and the equilibrium test allocation follows from Theorem \ref{thm:benchmark} with weights $\alpha^R$. In contrast, for the sophisticate, a strategic game similar to that in Section \ref{sec:strategic_players} ensues. However, Theorem \ref{thm:strategic_solution} does not apply directly, as the decision-maker's preferences are stochastic. By appropriately adapting Lemma \ref{lemma:strategic_player_maximization} and then applying the same logic as in the proof of Theorem \ref{thm:strategic_solution}, we find that the sophisticate's equilibrium test allocation can nevertheless be expressed as the solution to a single-player problem with auxiliary weights\footnote{The formal steps are contained in the proof of Proposition \ref{prop:strat_ignorance}.}
	\begin{align*}
		\left( \hat{\alpha}_1, \hat{\alpha}_2 \right) := \left( \alpha_1^R,\alpha_2^R\sqrt{\max\{1-p  + p( c(2 -c)) , 0\}} \right).
	\end{align*}
	Thus, the potential change in preferences reduces the weight put on attribute $\tilde{\theta}_2$, as the sophisticate anticipates the misalignment between his current and future self. Notably, the next result shows that this reduction may even induce the sophisticate to engage in ``strategic ignorance'' \citep[see, e.g., ][]{carrillo2000strategic}.
 
	\begin{proposition}\label{prop:strat_ignorance}
		The sophisticate avoids learning about attribute $\tilde{\theta}_2$ for any $T>0$ if and only if $\sqrt{p}(c-1)\geq 1$.
	\end{proposition}
	In the presence of a potential shift in preferences, represented by $c$, the sophisticate may choose strategic ignorance regarding attribute $\tilde{\theta}_2$. If the change reduces the weight on $\tilde{\theta}_2$ or does not increase it too much, the agent still learns about it, but with less intensity ($\hat{\alpha}_2 < \alpha_2^R$). In contrast, the sophisticate avoids learning about attribute $\tilde{\theta}_2$ if the change is sufficiently high ($c>2$) and the probability of change is not too small. In short, when the sophisticate expects the future self to overreact to information about $\tilde{\theta}_2$, she may opt for strategic ignorance.

	Having established the differences between the naif and the sophisticate's equilibrium learning strategies, we want to understand the implications for the agent's welfare. In this context, the choice of the appropriate welfare criterion is unclear. For instance, if the change in preferences is triggered by succumbing to temptation, the initial utility function $u^R$ seems most appropriate. However, if the change in preferences is due to altered  circumstances, such as illness or parenthood, then the changed utility function $u^{D}$ seems more appropriate. We remain agnostic about the welfare criterion and consider both alternatives. 
	\begin{proposition}\label{prop:selves_welfare}
	If $\alpha_1^R \notin (\hat{\alpha}_2/2, 2\alpha_2^R)$, the expected payoff of the naif and the sophisticate coincide for both welfare criteria. Otherwise, if the welfare criterion is the initial utility, $u^R$, the welfare of the sophisticate exceeds that of the naif. If the welfare criterion is the changed utility, $u^{D}$, the naif's welfare exceeds that of the sophisticate for $c>1$ and vice versa for $c<1$. 
	\end{proposition}
	The interesting case arises when a potential change in preferences influences equilibrium learning strategies. As shown in the appendix, the sophisticate allocates (at least weakly) more tests to attribute $\tilde{\theta}_1$ than the naif. Anticipating a shift in preferences, the sophisticate reduces emphasis on attribute $\tilde{\theta}_2$ to avoid misalignment with the future self's choices. When initial utility is used as the welfare criterion, the sophisticate’s expected payoff is higher because they optimize learning, aware of possible preference changes, while the naif makes suboptimal choices. However, when the welfare criterion is the changed utility, the naif can outperform the sophisticate. If $c>1$, the naif learns more about $\tilde{\theta}_2$, benefiting whether preferences change or not: If a change occurs, the naif has learned more about the now more important attribute $\tilde{\theta}_2$ than the sophisticate (in extreme case, the sophisticate engages in strategic ignorance and learns nothing about $\tilde{\theta}_2$); if no change occurs, the naif has implemented the optimal test allocation due to her naivete. In contrast, the sophisticate’s ``hedging" learning strategy against the potential change, which underweighs $\tilde{\theta}_2$, is less effective. However, this is reversed when $c<1$. Then, the naif's learning strategy is still optimal when no change occurs but (substantially) worse when a change occurs.
	
	The result in Proposition \ref{prop:selves_welfare} has a parallel to \citeauthor{odonoghue1999doing}'s (\citeyear{odonoghue1999doing}) findings in their ``doing it now or later'' framework. They find that in case of immediate costs the sophisticate is always better off than the naif (as in our case with the initial utility). Conversely, in case of immediate gratification, either type of the agent can be better off (as in our case with changed utility). In both, \citet{odonoghue1999doing} and our model, the sophisticate may engage in a form of overcompensation. In our case it happens by putting less weight on attribute $\tilde{\theta}_2$ to hedge against preference changes, which leads to a worse outcome \emph{irrespective} of whether preferences actually change. Summarizing, both models highlight that naive strategies can sometimes lead to better outcomes if future preferences or circumstances validate those strategies. This underscores the complexity of modeling changing preferences and how different assumptions about future behavior or the reason for change can impact welfare.

	%=====================================================
	\subsubsection{Proofs for the extension on changing preferences}\label{appendix:multiple_selves}
	%=====================================================
	
	\paragraph*{Equilibrium concept.} 
	A weak Perfect Bayesian Equilibrium of the game with potentially changing preferences (hereinafter referred to simply as ``equilibrium'') is a tuple $(\tau^*, d^*_c(s,\tau),d^*_u(s,\tau), \tilde{\theta}(s,\tau))$ such that:
	\begin{enumerate}
		\item the researcher's test allocation strategy $\tau^* \in \mathcal{T}$  maximizes his expected payoff given the decision-maker's strategies $d^*_c(s,\tau)$ and $d^*_u(s,\tau)$;
		%, and whenever the researcher is indifferent between some $\tau$ and abstaining from learning, he chooses the latter;
		\item the decision-maker's decision strategies $d^*_c(s,\tau): \mathbb{R}^2 \times \mathcal{T} \to \mathbb{R}$ and $d^*_u(s,\tau): \mathbb{R}^2 \times \mathcal{T} \to \mathbb{R}$ maximize her expected payoff given changed and unchanged preferences, respectively, and her posterior beliefs $\tilde{\theta}(s,\tau)$;
		\item the decision-maker's posterior beliefs $\tilde{\theta}(s,\tau): \mathbb{R}^2 \times \mathcal{T} \to \varDelta(\mathbb{R}^2)$ are obtained via Bayes' rule given the signal realizations $s$ and the researcher's choice $\tau$.
	\end{enumerate}

	%=========================================
	%\subsubsection{Proof of Proposition \ref{prop:strat_ignorance}}
	\begin{proof}[Proof of Proposition \ref{prop:strat_ignorance}]
	%=========================================
	
	Observe that the decision-maker's equilibrium decision is given by
	\begin{align*}
			d_u^\ast(s, \tau)&= \alpha^R_1 \hat{\mu}_1(s,\tau) + \alpha_2^R \hat{\mu}_2(s,\tau) \quad\;  \text{with probability } 1-p, \\
			d_c^\ast(s, \tau)&= \alpha^R_1 \hat{\mu}_1(s,\tau) + c\alpha_2^R \hat{\mu}_2(s,\tau) \quad \text{with probability } p.	
	\end{align*}
	Thus, for any test allocation $\tau$, the uncertainty about the decision-maker's preferences translates accordingly into uncertainty about the decision-maker's decision. Hence, we can write the researcher's value as
	\begin{align*}
		V^R(\tau) =& - \mathbb{E}\left[p(\alpha^R_1 \hat{\mu}_1(s,\tau) + c\alpha_2^R \hat{\mu}_2(s,\tau)  -\tilde{b}^R)^2 + (1-p)((\alpha^R_1 \hat{\mu}_1(s,\tau) + \alpha_2^R \hat{\mu}_2(s,\tau))- \tilde{b}^R)^2\right].
	\end{align*}
	By appropriately adapting Lemma \ref{lemma:strategic_player_maximization}, we obtain that the equilibrium test allocation is found by solving
	\begin{align*}
		\max_{\tau \in \mathcal{T}}  \;- p\left(  \frac{\left(\alpha_1^R\right)^2\Sigma^0_1}{1 + \tau_1 \Sigma^0_1} +\frac{\left(\alpha_2^R\right)^2\Sigma^0_2}{1 + \tau_2\Sigma^0_2} \left(c^2 + 2 c \Sigma^0_2 \tau_2\right)\right) - (1-p)\left( \sum_{j=1}^2 \frac{\left(\alpha_j^R\right)^2\Sigma^0_j}{1 + \tau_j \Sigma^0_j} \right).
	\end{align*}
	
	Proceeding as in the benchmark model, we can obtain the FOCs for the learning problem. These are given by
	\begin{align*}
	&	\left(\frac{ \Sigma_1^0}{1 + \Sigma_1^0 \tau_1^\ast}\right)^2(\alpha_1^R)^2 = \lambda,\\
	&\left(\frac{ \Sigma_2^0}{1 + \Sigma_2^0 \tau_2^\ast}\right)^2(\alpha_2^R)^2\left(1-p  + p(c(2 - c)) \right)= \lambda.
	\end{align*}
	By the same logic as in the proof of Theorem \ref{thm:strategic_solution}, the result follows.
	\end{proof}

	%=========================================
	%\subsubsection{Proof of Proposition \ref{prop:selves_welfare}}
	\begin{proof}[Proof of Proposition \ref{prop:selves_welfare}]
	%=========================================
	The proof proceeds in three steps.
	\paragraph*{Step 1.} We begin by observing that $\alpha_2^R \geq \hat{\alpha}_2 = \alpha_2^R\sqrt{\max\{1-p  + p( c(2 -c)) , 0\}},$ where $1-p  + p( c(2 -c))$ is a concave function of $c$ with a maximum at $c=1$. Thus, $\alpha_2^R \geq \hat{\alpha}_2^S$. Further, the sophisticate's and the naif's learning strategies $\tau_1^S$ and $\tau_ N$, respectively, read 
	\begin{align*}
		\tau^S_1 
		&= \max \left\{ 0, \min \left\{ \frac{{2\alpha}_1^R  - \hat{\alpha}_2 }{ {\alpha}_1^R+\hat{\alpha}_2} , 1  \right\}    \right\}, \quad \tau^N_1 
		= \max \left\{ 0, \min \left\{ \frac{{2\alpha}_1^R  - {\alpha}_2^R }{{\alpha}_1^R+{\alpha}_2^R}  , 1  \right\}    \right\}.
	\end{align*}
	Hence, $ \tau^S_1\geq \tau^N_1$, as the sophisticate puts less weight on the second attribute than the naif and they both put the same weight on the first. Thus, if naif learns only about attribute $\tilde{\theta}_1$, the sophisticate will do the same. The condition for that is $\frac{{2\alpha}_1^R  - {\alpha}_2^R }{{\alpha}_1^R+{\alpha}_2^R} \geq 1 \Leftrightarrow \alpha_1^R \geq 2 \alpha_2^R.$
	Analogously, when the sophisticate learns only about attribute $\tilde{\theta}_2$, the naif will do the same. The condition for that is $\frac{{2\alpha}_1^R  - \hat{\alpha}_2 }{ {\alpha}_1^R+\hat{\alpha}_2} \leq 0\Leftrightarrow\alpha_1^R\leq \frac{\hat{\alpha}_2}{2}.$
	Hence, when $\alpha_1^R \notin (\hat{\alpha}_2/2, 2\alpha_2^R)$ the two learning strategies coincide so that the expected payoffs of the naif and the sophisticate coincide. This proves the first part of the proposition. 
	
	\paragraph*{Step 2.} For the rest of the proof, let $\alpha^R_1 \in (\hat{\alpha}_2/2, 2\alpha_2^R)$. If the relevant welfare criterion is the initial utility (given the decision-maker's equilibrium decision rule), it coincides with the objective function the sophisticate maximizes. Hence, the second part of the proposition follows trivially as the learning strategies for the naif and the sophisticate differ when $\alpha^R_1 \in (\hat{\alpha}_2/2, 2\alpha_2^R)$.
	
	\paragraph*{Step 3.} If the relevant welfare criterion is the decision-maker's utility, the welfare function (given the decision-maker's equilibrium decision rule) is 
	\begin{align*}
		V^{D}(\tau) 
		&= p \left[-\frac{ \left(\alpha^R_1\right)^2  }{1+ \tau_1 } - \frac{ \left(c \alpha^R_2\right)^2  }{1+ \tau_2 }\right]  + (1-p)\left[-\frac{ \left(\alpha^R_1\right)^2  }{1+ \tau_1 } - \frac{ \left(\alpha^R_2\right)^2  }{1+ \tau_2 }\right] \\
		&= -\frac{ \left(\alpha^R_1\right)^2  }{1+ \tau_1 } - \frac{\left(\alpha^R_2\right)^2}{2-\tau_1 }\left(1-p + pc^2\right).
	\end{align*}
	Let us find the value of $\tau_1\in [0,1]$ that maximizes this function. The first and second order conditions show that $V^{D}(\tau)$ is a strictly concave function for $\tau_1\in [0,1]$ reaching a unique maximum at
	\begin{align*}
		\tau^{D}_1 = \max\left\{0, \min \left\{ \frac{2\alpha^R_1 -\alpha^R_2 \sqrt{1-p+pc^2}}{\alpha^R_1 + \alpha^R_2 \sqrt{1-p+pc^2}},1\right\} \right\}.
	\end{align*}
	Since $\alpha^R_1 \in (\hat{\alpha}_2/2, 2\alpha_2^R)$, the learning strategies of the naif and sophisticate differ. Furthermore, from Step 1, we then get that the sophisticate always learns strictly more about attribute $\tilde{\theta}_1$: $\tau^S_1 > \tau^N_1$.
	
	Note that a function $f(x):= \frac{2\alpha^R_1-x}{\alpha^R_1+x}$ is strictly decreasing in $x \in \mathbb{R}$. Furthermore, since $\alpha^R_1 \in (\hat{\alpha}_2/2, 2\alpha_2^R)$, we have 
	\begin{align*}
		\tau^N_1 &= \max\{f(\alpha^R_2),0\},
		\\
		\tau^S_1 &= \min \{ f(\hat{\alpha}_2),1\},
		\\
		\tau^{D}_1 &= \max \{0, \min \{f(\alpha^R_2\sqrt{p+(1-p)c^2}),1\}\}.
	\end{align*}
	Let $c>1$. Then $\tau^S_1>\tau^N_1 \geq \tau^{D}_1$ since  $\sqrt{1-p+pc^2}>\sqrt{p+(1-p)}=1$ and thus $\alpha^R_2\sqrt{1-p+pc^2}>\alpha^R_2$. Hence the naif's learning strategy is strictly closer to the maximizer. Due to the strict concavity of $V^{D}(\tau)$, the naif is strictly better off compared to the sophisticate.
	
	Let $c<1$. Then $\tau^{D}_1 \geq \tau^S_1 > \tau^N_1$. To see that, we show that $\hat{\alpha}_2> \alpha^R_2\sqrt{1-p+pc^2}$. Note that we have $\hat{\alpha}_2=\alpha^R_2\sqrt{1-p+pc(2-c)}$, which follows from $c<1$ and $p>0$ and the calculations in Step 1. Hence, the auxiliary weight of the sophisticate on attribute $\tilde{\theta}_2$ is not zero. The inequality $\hat{\alpha}_2> \alpha^R_2\sqrt{1-p+pc^2}$ then follows by noting that $c^2=c\times c<c(2-c)$ for $c<1$. Hence the sophisticate's learning strategy is strictly closer to the maximizer. Due to the strict concavity of $V^{D}(\tau)$, we get that the sophisticate is strictly better off compared to the naif.
	\end{proof}

\newpage

	\section{Additional Proofs} \label{sec:online_app_C}
 	%=====================================
	\subsection{Proof of Proposition \ref{prop:equivalent_models}} \label{OA:frameworks}
	%=====================================
	
	The Baseline Model and Framework A are equivalent up to realized payoffs, since the utility functions \eqref{eq:util_spec_baseline} and \eqref{eq:util_spec_weighted} are equivalent up to an exogenous random term. To see this, expand the squares in both utility functions, using the assumption that $\sum_k \alpha^i_k=1$:
	\begin{align*}
		u^i(d,\theta) 
		&= - \left( d - \sum_k \alpha^i_k \theta_k  \right)^2
		= - d^2 + 2 \sum_k d \alpha^i_k \theta_k - \left( \sum_k \alpha^i_k \theta_k \right)^2,
		\\
		u^i_A(d,\theta) 
		&= - \sum_k \alpha^i_k \left( d - \theta_k  \right)^2
		= - d^2 + 2 \sum_k d \alpha^i_k \theta_k - \sum_k  \alpha^i_k \theta_k^2.
	\end{align*}
	Observe that the difference $u^i(d,\theta) - u^i_A(d,\theta)$ is independent of the decision $d$. Therefore, the decision-maker's expected utilities \eqref{eq:util_interim_DM} only differ in the two settings by a constant independent of $d$, and the researcher's values \eqref{eq:util_exante_R} only differ by a constant.
	
	Now consider Framework B. The decision-maker's equilibrium decision $d^*(s,\tau)$ in this setting satisfies 
	\begin{align*}
		d^*_k(s,\tau) = \mathbb{E} \left[ \alpha^{D}_k \tilde{\theta}_k \right] 
		= \mathbb{E} \left[ \alpha^{D}_k \tilde{\theta}_k (s_k, \tau_k) \right] 
		= \alpha^{D}_k \hat{\mu}_k(s_k, \tau_k)
	\end{align*}
	for all $k$, where $\tilde{\theta}_k (s_k,\tau_k)$ is the decision-maker's posterior belief and $\hat{\mu}_k(s_k, \tau_k)$ is its posterior expectation, given by \eqref{eq:posterior_attributes} and \eqref{eq:posterior_mean_attributes}, respectively. 
	Let $\tilde{d}_k^*(\tau) := d^*_k(\tilde{s},\tau)$ for all $k$ denote the random variable that captures the distribution of the decision-maker's equilibrium decision given $\tau$.
	The researcher's value is then given by
	\begin{align*}
		V^R_B(\tau) = \mathbb{E} \left[ u^R_B \left( \tilde{d}^*(\tau), \tilde{\theta} \right) \right] 
		&= - \mathbb{E} \left[ \sum_k \left( \tilde{d}_k^*(\tau) - \alpha^{R}_k \tilde{\theta}_k \right)^2 \right]
        \\
		%&= \sum_k \left[ -\mathbb{E} \left[ \left(\tilde{d}_k^*(\tau) \right)^2 \right] + 2 \mathbb{E} \left[ \alpha^R_k \tilde{\theta}_k \tilde{d}_k^*(\tau) \right] - \mathbb{E} \left[ \left( \alpha^R_k \tilde{\theta}_k \right)^2 \right] \right] \\
		&= \sum_k \left[ - Var \left( \tilde{d}_k^*(\tau) \right) + 2 \cov\left( \tilde{d}_k^*(\tau), \alpha^{R}_k \tilde{\theta}_k \right) \right] + C,
	\end{align*}
	where $C$ is some constant that does not depend on $\tau$ (note, specifically, that $\mathbb{E} \left[ \tilde{d}_k^*(\tau) \right] = \alpha^{D}_k \mu_k^0$ does not depend on $\tau$). This is analogous (though not equivalent) to Lemma \ref{lemma:strategic_player_maximization}.
	
	From \eqref{eq:posterior_mean_attributes} and \eqref{eq:signal_defn} we have that 
	\begin{align*}
		Var \left( \tilde{d}_k^*(\tau) \right) 
		&= Var \left( \alpha^{D}_k \frac{\tau_k \Sigma^0_k}{1 + \tau_k\Sigma^0_k} \tilde{s}_k \right)
		%= \left( \frac{\tau_k \Sigma^0_k}{1 + \tau_k\Sigma^0_k} \alpha^{D}_k \right)^2 \left( \Sigma^0_k + \frac{1}{\tau_k} \right) 
		= \frac{\tau_k (\Sigma^0_k)^2}{1 + \tau_k\Sigma^0_k} \left( \alpha^{D}_k \right)^2,
		\\
		\cov\left( \tilde{d}_k^*(\tau), \alpha^{R}_k \tilde{\theta}_k \right)
		&= \cov \left( \alpha^{D}_k \frac{\tau_k \Sigma^0_k}{1 + \tau_k\Sigma^0_k} \tilde{s}_k,\, \alpha^{R}_k \tilde{\theta}_k \right)
		= \frac{\tau_k (\Sigma^0_k)^2}{1 + \tau_k\Sigma^0_k} \alpha^{D}_k \alpha^R_k,
	\end{align*}
	hence the objective function above reduces to
	\begin{align*}
		V^R_B(\tau) = \sum_k \frac{\tau_k (\Sigma^0_k)^2}{1 + \tau_k\Sigma^0_k} \left( 2\alpha^R_k - \alpha^{D}_k \right) \alpha^{D}_k  + C
	\end{align*}
	and its partial derivative w.r.t. $\tau_k$ is given by
	\begin{align*}
		\frac{\partial V^R_B(\tau)}{\partial \tau_k} = \left( 2\alpha^R_k - \alpha^{D}_k \right) \alpha^{D}_k \left( \frac{1}{\Sigma^0_k} + \tau_k \right)^{-2},
	\end{align*}
	which is equivalent to \eqref{eq:foc_strategic}. The remainder of the proof of Theorem \ref{thm:strategic_solution} then applies to Framework B, and the statement of Theorem \ref{thm:strategic_solution} holds. The equilibrium test allocation is, therefore, the same as in the Baseline Model.
	\qed

	\subsection{Details of the Proofs in Appendix \ref{appendix:organizations} }\label{OA:organizations}
	%=====================================================

	%============================================
	\subsubsection{Proof of Lemma \ref{lemma:characterization_learning_low}}
	%=======================================

        Fix the manager's weight vector $\alpha^{D}=(\alpha^{D}_1,\alpha^{D}_2)$ with $\alpha^{D}_1\geq \alpha^{D}_2>0$. Let $\alpha^{R}(\beta,\gamma)$ 
		%	\begin{eqnarray*}
			%		\alpha^{R}(\gamma,\beta) := \beta \alpha^{D} + \gamma \bar{\alpha}^{D}
			%	\end{eqnarray*}
		denote the analyst's weight vector as a function of $\beta$ and $\gamma$ for a given $\alpha^{D}$  as defined by \eqref{eq:decomposition_binary}.\footnote{For simplicity, we omit the dependence on $\alpha^{D}$ in the notation.} We thus have
		\begin{eqnarray}\label{eq:alpha1a_alpha2a_gamma_beta}
			\alpha^{R}_1(\gamma,\beta) = \beta \alpha^{D}_1 - \gamma \alpha^{D}_2,\ \ \  
			\alpha^{R}_2(\gamma,\beta) = \beta \alpha^{D}_2 + \gamma \alpha^{D}_1.
		\end{eqnarray}
		For $\beta\in \mathbb{R}_+$, let $\gamma_1(\beta)\in \mathbb{R}$ and $\gamma_2(\beta)\in \mathbb{R}$ denote the values of $\gamma$ defined by equations $\alpha^{R}_1(\beta,\gamma_1(\beta)) = 1/2 \alpha^{D}_1,$ and $\alpha^{R}_2(\beta,\gamma_2(\beta)) = 1/2 \alpha^{D}_2.$	From equations \eqref{eq:alpha1a_alpha2a_gamma_beta}, we get
		\begin{eqnarray} \label{eq:gamma_beta_1_2}
			\gamma_1(\beta) 
			=  - \left(\frac{1}{2} - \beta \right) \frac{\alpha^{D}_1}{\alpha^{D}_2}, \ \ \
			\gamma_2(\beta) 
			= \left( \frac{1}{2} - \beta \right) \frac{\alpha^{D}_2}{\alpha^{D}_1}
		\end{eqnarray}
		Since $\alpha^{R}_1(\beta,\gamma)$ ($\alpha^{R}_2(\beta,\gamma)$) is strictly decreasing (strictly increasing) in $\gamma$, the following results hold: Given $\beta \in \mathbb{R}_+$ (a) $\alpha^{R}_1(\beta,\gamma) \leq 1/2 \alpha^{D}_1$ if and only if $\gamma \geq \gamma_1(\beta)$; and (b)  $\alpha^{R}_2(\beta,\gamma) \leq 1/2 \alpha^{D}_2$ if and only if $\gamma \leq \gamma_2(\beta)$.
		
		Note that $\gamma_1(\beta)$ ($\gamma_2(\beta)$) is strictly increasing (strictly decreasing) in $\beta$. Furthermore, whenever $\beta < 1/2$ we have $\gamma_1(\beta) < 0 < \gamma_2(\beta)$ and for $\beta=1/2$ we have $\gamma_1(\beta)=0=\gamma_2(\beta)$. Taking $\beta\leq 1/2$, combining results (a) and (b) %and Assumption~\ref{assump:ties} to resolve multiplicity of the analyst's equilibrium behavior\footnote{The multiplicity of the analyst's equilibrium behavior occurs in the nongeneric cases of $(\beta,\gamma)$ for which the analyst's weights of the two attributes $k\neq l$ satisfy $\alpha^{R}_k(\beta,\gamma)=1/2\alpha^{D}_k$ and $\alpha^{R}_l(\beta,\gamma)\leq 1/2 \alpha^{D}_l$.}, 
		we can then apply Theorem~\ref{thm:strategic_solution} and obtain the analyst's equilibrium test allocation as
		\begin{eqnarray}\label{eq:optimal_testing_low}
			\tau^*(\beta,\gamma) = (\tau_1^*(\beta,\gamma), \tau^*_2(\beta,\gamma)) =
			\begin{cases}
				(T,0)	& \text{if} \ \gamma<\gamma_1(\beta) \\
				(0,0)	& \text{if} \ \gamma \in \left[ \gamma_1(\beta), \gamma_2(\beta)  \right] \\
				(0,T)	& \text{if} \ \gamma > \gamma_2(\beta)
			\end{cases}.
		\end{eqnarray}
		For each $\beta$, however, $\gamma$ also needs to satisfy the constraint $\gamma \in \Gamma(\beta)$. For $\beta\in \mathbb{R}_+$, let 
		\begin{eqnarray}\label{eq:underline_overline_gamma}
			\underline{\gamma}(\beta) 
			:= - \beta \frac{\alpha^{D}_2}{\alpha^{D}_1}, \ \ \ 
			\overline{\gamma}(\beta)
			:= \beta \frac{\alpha^{D}_1}{\alpha^{D}_2}
		\end{eqnarray}
		denote the end points of the constraint $\Gamma(\beta)$. Let $\beta_1,\beta_2 \in \mathbb{R}_+$ denote the values of $\beta$ defined by equations $\gamma_1(\beta_1) = \underline{\gamma}(\beta_1)$ and $\gamma_2(\beta_2) = \overline{\gamma}(\beta_2).$	
		Solving these equations, we obtain
		\begin{eqnarray} \nonumber
			\beta_1 
			&=& \frac{1}{2} \frac{ \left( \alpha^{D}_1 \right)^2 }{ \left( \alpha^{D}_1  \right)^2 + \left(  \alpha^{D}_2  \right)^2 }, \ \ \ 
			\beta_2 
			= \frac{1}{2} \frac{ \left( \alpha^{D}_2 \right)^2 }{ \left( \alpha^{D}_1  \right)^2 + \left(  \alpha^{D}_2  \right)^2 }
		\end{eqnarray}
		Since $\alpha^{D}_1\geq \alpha^{D}_2 >0$, we get $0<\beta_2 \leq \beta_1 < 1/2$. Furthermore, since $\gamma_1(\beta)$ ($\gamma_2(\beta)$) is strictly increasing (strictly decreasing) in $\beta$ and since $\underline{\gamma}(\beta)$ ($\overline{\gamma}(\beta)$) is strictly decreasing (strictly increasing) we get the following results: (i) for all $\gamma \in \Gamma(\beta)$ we have $\gamma\geq \gamma_1(\beta)$ if and only if $\beta \leq \beta_1$; and (ii) for all $\gamma \in \Gamma(\beta)$ we have $\gamma \leq \gamma_2(\beta)$ if and only if $\beta \leq \beta_2$. Combining points (i) and (ii) with the equilibrium test allocation defined in equation \eqref{eq:optimal_testing_low}, we obtain Lemma~\ref{lemma:characterization_learning_low}.	
	\qed

	%============================================
	\subsubsection{Proof of Lemma \ref{lemma:characterization_learning_high}}
	%=======================================

		%Fix the test budget $T>0$ and 
     Fix the manager's weight vector $\alpha^{D} = (\alpha^{D}_1,\alpha^{D}_2)$ with $\alpha^{D}_k >0$ for both $k=1,2$. Let $\alpha^{R}_1(\beta,\gamma)$, $\alpha^{R}_2(\beta,\gamma)$,  $\gamma_1(\beta)$, $\gamma_2(\beta)$, $\underline{\gamma}(\beta)$, and $\overline{\gamma}(\beta)$ be given by equations \eqref{eq:alpha1a_alpha2a_gamma_beta}, \eqref{eq:gamma_beta_1_2}, and \eqref{eq:underline_overline_gamma}  in the proof of Lemma~\ref{lemma:characterization_learning_low}. 
		
		Fix $\beta>1/2$. Then we have $\underline{\gamma}(\beta)<\gamma_2(\beta) < 0 < \gamma_1(\beta) < \overline{\gamma}(\beta)$.	As shown in the proof of Lemma~\ref{lemma:characterization_learning_low}, it holds that $\alpha^{R}_1(\beta,\gamma) \leq 1/2 \alpha^{D}_1$ if and only if $\gamma \geq \gamma_1(\beta)$; and $\alpha^{D}_2(\beta,\gamma) \leq 1/2 \alpha^{D}_2$ if and only if $\gamma \leq \gamma_2(\beta)$. By Theorem~\ref{thm:strategic_solution}, the analyst's equilibrium test allocation satisfies %\footnote{Note that at $\gamma=\gamma_1(\beta)$ there is no multiplicity of the analyst's optimal behavior since $\alpha^{R}_1(\beta,\gamma_1(\beta))=1/2 \alpha^{D}_1$ and $\alpha^{R}_2(\beta,\gamma_1(\beta))>1/2\alpha^{D}_2$. Similarly for $\gamma=\gamma_2(\beta)$. Hence, when $\beta>1/2$, Assumption~\ref{assump:ties} is not needed since the analyst's optimal behavior is unique for each $\gamma\in \Gamma(\beta)$, unlike in Lemma~\ref{lemma:characterization_learning_low}.
		\begin{eqnarray}\label{eq:tau_star_gamma_ends}
			\tau^*(\beta,\gamma) = \left( \tau^*_1(\beta,\gamma), \tau^*_2(\beta,\gamma) \right)
			&=&
			\begin{cases}
				(T,0)	& \text{if}\ \gamma \leq \gamma_2(\beta), \\
				(0,T)	& \text{if}\ \gamma \geq \gamma_1(\beta).
			\end{cases}
		\end{eqnarray}
		Now suppose $\gamma \in \left[\gamma_2(\beta),\gamma_1(\beta) \right]$. Let $\tilde{\tau}_k(\beta,\gamma)$  and $\hat{\alpha}_k(\beta,\gamma)$ for $k = 1,2$ be given by equations \eqref{eq:tilde_tau_gamma} and \eqref{eq:tilde_alpha_gamma}. Note that expressions in \eqref{eq:tilde_alpha_gamma} are well-defined for $\gamma \in \left[\gamma_2(\beta),\gamma_1(\beta) \right]$. 
		
		Furthermore, we have
		\begin{eqnarray*}
			\tilde{\tau}_1(\beta,\gamma) = 
			\begin{cases}
				\frac{1}{\Sigma^0_{2}} + T >T & \text{at}\ \gamma=\gamma_2(\beta), \\
				-\frac{1}{\Sigma^0_{1}} < 0 & \text{at}\ \gamma=\gamma_1(\beta),
			\end{cases}
		\end{eqnarray*}
		since $\hat{\alpha}_2(\beta,\gamma_2(\beta))=0$, $\hat{\alpha}_1(\beta,\gamma_2(\beta))>0$, and $\hat{\alpha}_1(\beta,\gamma_1(\beta)) = 0$, $\hat{\alpha}_2(\beta,\gamma_1(\beta))>0$. Further, $\tilde{\tau}_1(\beta,\gamma)$ is continuous and strictly decreasing in $\gamma$. To see the strict monotonicity, let us take the partial derivative of $\tilde{\tau}_1(\gamma)$ with respect to $\gamma$. For ease of notation, we suppress the dependence on $\beta$ and $\gamma$ from the  right-hand side of the expression. We get
		\begin{eqnarray*}
			\frac{ \partial \tilde{\tau}_1 (\beta,\gamma)}{\partial \gamma}
			= \frac{(\Sigma^0_{1}+\Sigma^0_{2}) \left( \hat{\alpha}_2 \frac{\partial \hat{\alpha}_1}{\partial \gamma}  - \hat{\alpha}_1 \frac{\partial \hat{\alpha}_2}{\partial \gamma}   \right)}{ \Sigma^0_{1} \Sigma^0_{2} \left( \hat{\alpha}_1 + \hat{\alpha}_2  \right)^2}		
			+ \frac{\hat{\alpha}_2 \frac{\partial \hat{\alpha}_1}{\partial \gamma} - \hat{\alpha}_1 \frac{\partial \hat{\alpha}_2}{\partial \gamma}    }{ \left( \hat{\alpha}_1 + \hat{\alpha}_2\right)^2}T <0,
		\end{eqnarray*}
		where we used $\frac{\partial \hat{\alpha}_1(\beta,\gamma)}{\partial \gamma} = -\frac{\alpha^{D}_2}{\sqrt{\hat{\alpha}_1(\beta,\gamma)}}<0$ and $\frac{\partial \hat{\alpha}_2(\gamma)}{\partial \gamma} = \frac{\alpha^{D}_1}{\sqrt{\hat{\alpha}_2(\beta,\gamma)}}>0$ to determine the sign. By the intermediate value theorem, there exist $\gamma_I(\beta), \gamma_{II}(\beta) \in \left( \gamma_2(\beta), \gamma_1(\beta) \right)$ such that 
		\begin{eqnarray}\label{eq:gamma_I_II}
			\tilde{\tau}_1(\beta, \gamma) = 
			\begin{cases}
				T	& \text{at}\ \gamma=\gamma_{I}(\beta), \\
				0	& \text{at}\ \gamma = \gamma_{II}(\beta).
			\end{cases}
		\end{eqnarray}
		From the strict monotonicity of $\tilde{\tau}_1(\beta,\gamma)$, it further holds that
		\begin{eqnarray}\label{eq:tilde_tau_1_gamma_cutoffs}
			\tilde{\tau}_1(\beta,\gamma) 
			\begin{cases}
				> T	& \text{if and only if}\ \gamma < \gamma_{I}(\beta), \\
				< 0 & \text{if and only if}\ \gamma > \gamma_{II}(\beta).
			\end{cases}
		\end{eqnarray}
		Finally, note that $\hat{\alpha}_k(\beta,\gamma)$ as given by \eqref{eq:tilde_alpha_gamma} are equivalent to \eqref{eq:alpha_n_lambda} and \eqref{eq:auxiliary_weights} given decomposition \eqref{eq:decomposition_binary}. 
        By Theorems~\ref{thm:benchmark} and \ref{thm:strategic_solution}, 
		%By Theorem~\ref{thm:strategic_solution} and Corollary~\ref{cor:binary_solution_single_player}, 
        the analyst's equilibrium test allocation for $\gamma \in \left[\gamma_2(\beta),\gamma_1(\beta)\right]$ is
		\begin{eqnarray} \label{eq:tau_star_1_gamma_learning}
			\tau^*_1(\beta,\gamma)
			&=&
			\max \left\{ 0, \min \left\{ \tilde{\tau}_1(\beta,\gamma) , T  \right\}    \right\} \\
			\tau^*_2(\beta,\gamma) \label{eq:tau_star_2_gamma_learning}
			&=&
			\max \left\{ 0, \min \left\{ \tilde{\tau}_2(\beta,\gamma) , T  \right\}    \right\}		
		\end{eqnarray}
		Combining the results \eqref{eq:tau_star_1_gamma_learning}, \eqref{eq:tau_star_2_gamma_learning} with \eqref{eq:tilde_tau_1_gamma_cutoffs} and \eqref{eq:tau_star_gamma_ends}, together with the strict monotonicity result of $\tilde{\tau}_1(\beta,\gamma)$, we obtain the claim of Lemma~\ref{lemma:characterization_learning_high}.	
	\qed

	%============================================
	\subsubsection{Proof of Lemma \ref{lemma:binary_Vp_concavity}}
	%=======================================

		Given a test allocation $\tau=(\tau_1,\tau_2)$, the decision-maker's interim expected payoff is 
		\begin{eqnarray}\label{eq:VP_tau}
			V^{D}(\tau) = -\frac{ \left(\alpha^{D}_1\right)^2 \Sigma^0_{1} }{1+ \tau_1 \Sigma^0_{1}} - \frac{ \left(\alpha^{D}_2\right)^2 \Sigma^0_{2} }{1+ \tau_2 \Sigma^0_{2}}.
		\end{eqnarray}
		The first, the second, and the cross derivates of the decision-maker's interim expected payoff with respect to $\tau_k \geq 0$ and $\tau_l\neq \tau_k$ are
        $\frac{\partial V^{D}(\tau)}{ \partial \tau_k}=\left( \frac{ \alpha^{D}_k \Sigma^0_{k}  }{1+\tau_k \Sigma^0_{k}} \right)^2 >0$, $\frac{\partial^2 V^{D}(\tau)}{\partial \tau_k^2}=- 2 \left( \alpha^{D}_k \right)^2 \left( \frac{ \Sigma^0_{k}  }{ 1+\tau_k \Sigma^0_{k}  }  \right)^3 <0 $, and $\frac{\partial^2 V^{D}(\tau)}{\partial \tau_k \partial \tau_l}=0$.
   % \begin{eqnarray*}
%			\frac{\partial V^{D}(\tau)}{ \partial \tau_k}
%			&=& 
%			\left( \frac{ \alpha^{D}_k \Sigma^0_{k}  }{1+\tau_k \Sigma^0_{k}} \right)^2 >0 \\
%			\frac{\partial^2 V^{D}(\tau)}{\partial \tau_k^2}
%			&=&
%			- 2 \left( \alpha^{D}_k \right)^2 \left( \frac{ \Sigma^0_{k}  }{ 1+\tau_k \Sigma^0_{k}  }  \right)^3 <0 \\
%			\frac{\partial^2 V^{D}(\tau)}{\partial \tau_k \partial \tau_l}
%			&=& 0
%		\end{eqnarray*}
		Hence, $V^{D}(\tau)$ is strictly increasing and strictly concave function.
\qed 
	
	%============================================
	\subsubsection{Proof of Lemma \ref{lemma:optimal_allocation_gamma_zero}}
	%=======================================

		%	Suppose Assumption~\ref{assump:ties} holds.
		Fix the test budget $T>0$ and the decision-maker's weight vector $\alpha^{D} \in \mathbb{R}^2_{++}$.
		Suppose $\gamma=0$ (the agent is not distorted).
		\begin{itemize}
			\item[(i)] If $\beta\leq 1/2$ (the agent is too insensitive), then the agent optimally chooses no testing: $\tau^*(\beta,0) = (0,0)$.
			\item[(ii)] If $\beta>1/2$ (the agent is sufficiently sensitive), then the agent optimally chooses the decision-maker's most preferred test allocation: $\tau^*(\beta,0) = \tau^*(1,0)$. 
		\end{itemize}

		Suppose the premise of Lemma~\ref{lemma:optimal_allocation_gamma_zero} holds. 
		Part (i) immediately follows from Lemma~\ref{lemma:characterization_learning_low}. 
		For part (ii) take $\beta>1/2$ and let $\tilde{\tau}_1(\beta,\gamma)$ and $\tilde{\tau}_2(\beta,\gamma)$ be given by \eqref{eq:tilde_tau_gamma}. At $\gamma=0$, we get
		\begin{eqnarray*}
			\tilde{\tau}_1(\beta,0)
			&=&
			%	\frac{\hat{\alpha}_1(0)\Sigma^0_{1} - \hat{\alpha}_2(0) \Sigma^0_{2}}{\Sigma^0_{1} \Sigma^0_{2} \left( \hat{\alpha}_1(0) + \hat{\alpha}_2(0)  \right)} + 
			%	\frac{ \hat{\alpha}_1(0)   }{ \hat{\alpha}_1(0)  + \hat{\alpha}_2(0)     } T
			%	= 
			\frac{ \sqrt{2\beta - 1} \alpha^{D}_1 \Sigma^0_{1} - \sqrt{2\beta-1}\alpha^{D}_2 \Sigma^0_{2}    }{   \Sigma^0_{1} \Sigma^0_{2} \sqrt{2\beta-1} \left( \alpha^{D}_1 + \alpha^{D}_2 \right)    }
			+ \frac{ \sqrt{2\beta - 1} \alpha^{D}_1  }{ \sqrt{2\beta - 1}  \left( \alpha^{D}_1 +\alpha^{D}_2   \right)} T \\
			&=& 
			\frac{\alpha^{D}_1 \Sigma^0_{1} - \alpha^{D}_2 \Sigma^0_{2}}{\Sigma^0_{1}\Sigma^0_{2}(\alpha^{D}_1+\alpha^{D}_2)} + 
			\frac{\alpha^{D}_1}{\alpha^{D}_1+\alpha^{D}_2} T
		\end{eqnarray*}
		Lemma~\ref{lemma:characterization_learning_high} and 
        Theorem~\ref{thm:benchmark} %Corollary~\ref{cor:binary_solution_single_player} 
        then imply that the agent's equilibrium test allocation coincides with the optimal test allocation of a single player with weight vector $\alpha = \alpha^{D}$.
	\qed

 %=====================================================
	\subsection{Details of the Proofs in Appendix \ref{appendix:discrimination}}\label{OA:discrimination}
	%=====================================================
%============================================
	\subsubsection{Proof of Lemma \ref{lemma:discrimination_exists_partiality}}
	%=======================================

Let us set $T=1$ and $\Sigma^0_1=\Sigma^0_2=1$.
	From Theorems~\ref{thm:benchmark} and \ref{thm:strategic_solution}, %and Corollary~\ref{cor:binary_solution_single_player}, 
 the equilibrium test allocation satisfies $\tau^*_1(p,\delta)=1-\tau^*_2(p,\delta)$ and
	\begin{align*}
		\tau^*_2(p,\delta)
		= \min \left\{ \frac{\hat{\alpha}_2(p,\delta) - \hat{\alpha}_1(p,\delta)}{\hat{\alpha}_1(p,\delta)+\hat{\alpha}_2(p,\delta)} +
		\frac{\hat{\alpha}_2(p,\delta)}{\hat{\alpha}_1(p,\delta)+\hat{\alpha}_2(p,\delta)}, 1 \right\} 
        =
		\min \biggl\{ \underbrace{\frac{2\hat{\alpha}_2(p,\delta) - \hat{\alpha}_1(p,\delta)}{\hat{\alpha}_1(p,\delta)+\hat{\alpha}_2(p,\delta)}}_{ \tau^{aux}_2(p,\delta):=}, 1 \biggr\}
	\end{align*}
	where $\hat{\alpha}_1(p,\delta)$ and $\hat{\alpha}_2(p,\delta)$ are given by equations \eqref{app:alpha_1_tilde_p_d} and \eqref{app:alpha_2_tilde_p_d}, respectively. The function $\tau_2^{aux}(p,\delta)$ is well-defined, since $\hat{\alpha}_1(p,\delta)+\hat{\alpha}_2(p,\delta) \neq 0$ for all values of $p$ and $\delta$. Note that
	$\tau^{aux}_2(p,\delta)=1/2$ when $p=0$ and $\tau^{aux}_2(p,\delta)=2$ when $p \geq \frac{1-\delta}{2}$. Furthermore, when $p\in \left[0,\frac{1-\delta}{2} \right)$ we have
	\begin{eqnarray} \nonumber
		\frac{\partial \tau^{aux}_2(p,\delta)}{\partial p}
		&=&
		\frac{\left(2 \frac{\partial \hat{\alpha}_2}{\partial p} - \frac{\partial \hat{\alpha}_1}{\partial p}  \right) \left( \hat{\alpha}_1 +\hat{\alpha}_2 \right) 
			- \left( 2\hat{\alpha}_2 - \hat{\alpha}_1 \right) 
		\left( \frac{\partial \hat{\alpha}_1}{\partial p} + \frac{\partial \hat{\alpha}_2}{\partial p}    \right)}{
		\left( \hat{\alpha}_1 + \hat{\alpha}_2   \right)^2}	\\ \label{eq:discrimination_tau2_aux}
		&=& \frac{3}{4} \frac{ (1-\delta^2) }{\hat{\alpha}_1 \hat{\alpha}_2 \left( \hat{\alpha}_1 + \hat{\alpha}_2 \right)^2} >0. 
	\end{eqnarray}
	%where we used $\frac{\partial \hat{\alpha}_1}{\partial p} = - \frac{1+\delta}{4\hat{\alpha}_1}$ and $\frac{\partial \hat{\alpha}_2}{\partial p} = \frac{1-\delta}{4 \hat{\alpha}_2}$.
	Note that $\tau^{aux}_2(p,\delta)$ is continuous in $p$ for $p \geq 0$.\footnote{Since $\hat{\alpha}_1(p,\delta)$ and $\hat{\alpha}_2(p,\delta)$ are both continuous in $p$, the only possible point of discontinuity of $\tau_2^{aux}(p,\delta)$ is if $\hat{\alpha}_1(p,\delta)+\hat{\alpha}_2(p,\delta) = 0$. Equations \eqref{app:alpha_1_tilde_p_d} and \eqref{app:alpha_2_tilde_p_d} imply this never happens for all feasible values of $p$ and $\delta$. Hence, $\tau^{aux}_2(p,\delta)$ is continuous on $p\geq 0$.}
	Furthermore, $\tau^{aux}_2(p,\delta)$ is strictly increasing in $p$ on $p \in [0,(1-\delta)/2)$. At $p=0$ we have $ \tau^{aux}_2(p,\delta)=1/2<1$ and at $p \geq (1-\delta)/2$ we have $\tau^{aux}_2(p,\delta)=2>1$. Hence, for each value of discrimination $\delta \in (0,1)$ there exists a threshold partiality level of the advisor $p^{aux}(\delta)\in \left(0,\frac{1-\delta}{2}\right)$, defined implicitly by the equation $\tau^{aux}_2(p^{aux}(\delta),\delta)=1$, such that the equilibrium testing strategy satisfies $\tau^*_2(p,\delta)=\tau^{aux}_2(p,\delta)$ for $p \in [0,p^{aux}(\delta)]$ and $\tau^*_2(p,\delta)=1$ for $p \geq p^{aux}(\delta)$.
\qed

	%============================================
	\subsubsection{Proof of Lemma \ref{lem:discrimination_restore_equality}}
	%=======================================
The ex ante expected utility of group $k$, $V^k(\tau,\delta)$, can be derived from the expected payoff of a researcher $V^R(\tau)$ in the baseline model in the proof of Lemma~\ref{lemma:strategic_player_maximization}, where we set the decision-maker's weights to $\alpha^{D}_1(\delta)$, $\alpha^{D}_2(\delta)$ and the researcher's weights (now group $k$ weights) to $\alpha^{R}_k =1$ and $\alpha^{R}_{-k}=0$. In other words, given allocation $\tau$ and the politician's equilibrium strategy $d^*(s,\tau;\delta)$, the ex ante expected payoff of group $k$ is the same as that of the researcher in the baseline model who solely cares about attribute $\tilde{\theta}_k$ (with weight $\alpha^R_k=1$) and not about the other attribute (with weight $\alpha^R_{-k}=0$).
We get
 \begin{align*}
 V^k(\tau,\delta) =&\ \mathbb{E}\left[ -(\tilde{d}^{*D}(\tau;\delta) - \tilde{\theta}_k)^2 \right] 
		\\ 
		=& - \left(v^{D}_0(\delta) - \mu_k^0   \right)^2 - \left( \sigma^{2,D}_0(\delta) + \Sigma^0_k  \right) + \hat{\sigma}^{2,D}(\tau;\delta) + 2 \cov\left(\tilde{d}^{*D}(\tau;\delta),\tilde{\theta}_k  \right) 
		\\ 
		=& -\left( \alpha^{D}_1(\delta)\mu_1^0 + \alpha^{D}_2(\delta)\mu_2^0  - \mu_k^0 \right)^2 - \left( \left(\alpha^{D}_1(\delta)\right)^2\Sigma^0_1 + \left(\alpha^{D}_2(\delta)\right)^2\Sigma^0_2 + \Sigma^0_k  \right) 
		\\
		&+ \frac{\Sigma^0_k}{1+\tau_k \Sigma^0_k} \left( \left(\alpha^{D}_k(\delta)\right)^2 + 2\alpha^{D}_k(\delta) \Sigma^0_k \tau_k \right) + \frac{\Sigma^0_{-k}}{1+\tau_{-k} \Sigma^0_{-k}} \left(\alpha^{D}_{-k}(\delta)\right)^2
	\end{align*}
    Using $\mu^0_1=\mu^0_2=0$, $\Sigma^0_1=\Sigma^0_2=1$, we get
	\begin{align} \nonumber
		V^k(\tau,\delta)
		=& - \left(  \left(\alpha^{D}_1(\delta)\right)^2 + \left(\alpha^{D}_2(\delta)\right)^2 + 1 \right)\\ \label{eq:discrimination_group_exp_utility}
		&	+  \frac{1}{1+\tau_k} \left( \left(\alpha^{D}_k(\delta)\right)^2 + 2\alpha^{D}_k(\delta) \tau_k \right) + \frac{1}{1+\tau_{-k}} \left(\alpha^{D}_{-k}(\delta)\right)^2.
	\end{align}
    The difference between the expected utilities of the two groups is
	\begin{eqnarray*}
		V^1(\tau,\delta) - V^2(\tau,\delta) = \frac{2\alpha^P_1(\delta)\tau_1}{1+\tau_1} - \frac{2\alpha^P_2(\delta)\tau_2}{1+\tau_2} = \frac{(1+\delta)\tau_1}{1+\tau_1} - \frac{(1-\delta)\tau_2}{1+\tau_2}
	\end{eqnarray*}
	where we used $\alpha^P_1(\delta)=\frac{1}{2}(1+\delta)$ and $\alpha^P_2(\delta)=\frac{1}{2}(1-\delta)$. %In any equilibrium (for any sets of parameters $\delta$ and $p$), the available budget of resources $T$ is fully used. Let us thus constrain ourselves to analyze 
    The difference between the expected utilities of the two groups under the additional constraint that $\tau_1+\tau_2=1$ is
    \begin{align*}
		\Delta(\tau_2,\delta) &= V^1((1-\tau_2,\tau_2),\delta) - V^2((1-\tau_2,\tau_2),\delta)
		&= \frac{(1+\delta)(1-\tau_2)}{1+1-\tau_2} - \frac{(1-\delta)\tau_2}{1+\tau_2}.
	\end{align*}
   % where $\Delta(\tau_2,\delta)$ is a reformulated version of inequality $I(p,\delta)$, expressed as a function of test allocation $(\tau_1=T-\tau_2,\tau_2)$ rather than as a function of partiality $p$.

  	The	function $\Delta(\tau_2,\delta)$ is continuous and strictly decreasing in $\tau_2$ since $\delta \in (0,1)$ and thus
	$
	\frac{\partial \Delta(\tau_2,\delta)}{\partial \tau_2}
	= - \frac{(1+\delta)}{(1+1-\tau_2)^2} - \frac{(1-\delta)}{(1+\tau_2)^2} <0
	$.
	Furthermore, we have $\Delta(1/2,\delta)>0$ and $\Delta(1,\delta)<0$. Hence, for each $\delta$ there exists a unique value of $\hat{\tau}_2(\delta) \in \left( 1/2, 1 \right)$ for which it holds $\Delta(\hat{\tau}_2(\delta);\delta)=0$. 
\qed

	%============================================
	\subsubsection{Proof of Lemma \ref{lemma:pareto_frontier}}
	%============================================
 Let
\begin{align*} 
	\omega(\tau_2,\delta) &:= V^1((1-\tau_2,\tau_2),\delta) + V^2((1-\tau_2,\tau_2),\delta) \\
	%&= -3 - \delta^2 + \frac{1/2(1+\delta^2) + (1+\delta)(1-\tau_2)}{2-\tau_2} + \frac{1/2(1+\delta^2) + (1-\delta)\tau_2}{1+\tau_2}
	&= -3 - \delta^2 + \frac{1/2(1+\delta)^2 + (1+\delta)(1-\tau_2)}{2-\tau_2} + \frac{1/2(1-\delta)^2 + (1-\delta)\tau_2}{1+\tau_2}
\end{align*}
denote the sum of the expected utilities of the two groups as a function of test allocation $\tau = (\tau_1,\tau_2)$ under the constraint that the budget is fully used, $\tau_1=T-\tau_2$ with $T=1$, and where $V^k(\tau,\delta)$ is given by equation \eqref{eq:discrimination_group_exp_utility}. Since the budget is exhausted in equilibrium (as shown in the proof of Lemma~\ref{prop:discrimination_equalizing_bias}), welfare is thus given by $W(p,\delta) = \omega (\tau_2^*(p,\delta),\delta)$. The partial derivative of welfare with respect to $p>0$, whenever differentiable, is then
\begin{align*}
	\frac{\partial W(p,\delta)}{\partial p}
	&=
	\frac{\partial \tau^*_2}{\partial p} \frac{\partial \omega(\tau_2^*,\delta)}{\partial \tau_2} 
	= \frac{\partial \tau^*_2}{\partial p}
	\biggl(
		\frac{1/2(1+\delta)^2 - (1+\delta)}{(2-\tau_2^*)^2}
	- 
	\frac{1/2(1-\delta)^2 -(1-\delta)}{(1+\tau_2^*)^2}
	\biggr)
\end{align*}
where we suppressed the dependence of $\tau^*_2(p,\delta)$ on $(p,\delta)$ from the RHS for brevity.
Note that 
$1/2(1+\delta^2) - (1+\delta) = 1/2(1+\delta^2) - (1-\delta) = \frac{1-\delta^2}{2}$ and that $(2-\tau_2^*)^2<(1+\tau_2^*)^2$, since $\tau^*_2(p,\delta)>1/2$ for $p>0$. Hence, $\frac{\partial \omega (\tau_2^*,\delta)}{\partial \tau_2}<0$. We thus have $sign\left( \frac{\partial W(p,\delta)}{\partial p} \right) = - sign \left( \frac{\partial \tau^*_2}{\partial p} \right)$. 

Let $\tau^{aux}_2(p,\delta)$, $\hat{p}(\delta)$ and $p^{aux}(\delta)$ with $\hat{p}(\delta)<p^{aux}(\delta)$, be the variables defined in the proof of Lemma~\ref{prop:discrimination_equalizing_bias}. 
From equation \eqref{eq:discrimination_tau2_aux}, it follows
\begin{align*}
    sign\left( \frac{\partial W(p,\delta)}{\partial p} \right) 
    &= - sign \left( \frac{\partial \tau^*_2(p,\delta)}{\partial p} \right) = 
    \begin{cases}
        - sign \left(\frac{\tau_2^{aux}(p,\delta)}{\partial p} \right) < 0 & p \in (0,p^{aux}(\delta)) \\
        0 & p>p^{aux}(\delta)
    \end{cases},
\end{align*}
and the derivative does not exist at $p=p^{aux}(\delta)$. We thus have that welfare $W(p,\delta)$ is (i) continuous in $p$ (since $V^k(\tau^*(p,\delta),\delta)$ is continuous in $p$); (ii) strictly decreasing on $p \in (0,p^{aux}(\delta))$; and (iii) constant on $p\geq p^{aux}(\delta)$, where all the advisors use the entire budget to learn exclusively about group 2. Furthermore, the equalizing partiality level $\hat{p}(\delta)<p^{aux}(\delta)$.

Next, inequality is given by
\begin{align*}
    I(p,\delta) &= \left|\Delta(\tau_2^*(p,\delta),\delta) \right|
   % = \left| \frac{(1+\delta)(1-\tau_2^*(p,\delta))}{1+1-\tau_2^*(p,\delta)} - \frac{(1-\delta)\tau_2^*(p,\delta)}{1+\tau_2^*(p,\delta)} \right|
\end{align*}

where $\Delta(\tau_2^*(p,\delta),\delta)$ is given by equation \eqref{eq:discrimination_Cap_Delta_fn}.
As we have shown in the proof of Lemma~\ref{prop:discrimination_equalizing_bias}, the function $\Delta(\tau_2^*(p,\delta),\delta)$ is continuous in $\tau_2^*(p,\delta)$, strictly decreasing in $\tau_2^*(p,\delta)$, and is zero at $\tau_2^*(\hat{p}(\delta),\delta)$, where $\tau_2^*(\hat{p}(\delta),\delta)<T$ (the advisor who restores equality, $\hat{p}(\delta)$, does not use the entire test budget to learn exclusively about group 2). As shown above, we have:  the equilibrium allocation of tests to group 2,  $\tau_2^*(p,\delta)$, is continuous in $p$, strictly increasing on $p<p^{aux}(\delta)$, and constant on $p\geq p^{aux}(\delta)$, where the advisors use the entire test budget to learn exclusively about group 2: $\tau_2^*(p,\delta)=T$. Hence, inequality $I(p,\delta)$ is (i) continuous in $p$; (ii) strictly decreasing on $p \in \left(0,\hat{p}(\delta)\right)$ (iii) zero at $p=\hat{p}(\delta)$; 
(iv) strictly increasing on $p \in \left(\hat{p}(\delta), p^{aux}(\delta)\right)$; and (v) constant on $p \geq p^{aux}(\delta)$.

Therefore, welfare $W(p,\delta)$ and inequality $I(p,\delta)$ are both (i) continuous in $p$; and (ii) strictly decreasing on $p \in (0,\hat{p}(\delta))$. Furthermore, welfare strictly decreases and inequality strictly increases on $p \in (\hat{p}(\delta),p^{aux}(\delta))$; and they are both constant on $p \geq p^{aux}(\delta)$.

	%============================================
	\subsubsection{Proof of Lemma \ref{lemma:discr_tests_unchecked_equalizing}}
	%=======================================
By Theorem~\ref{thm:benchmark}, %Corollary~\ref{cor:binary_solution_single_player}, 
the politician's optimal learning strategy under unchecked discrimination yields
\begin{align*}
	\bar{\tau}_2(\delta) = \begin{cases}
		\frac{1-3\delta}{2} & \text{if } \delta < \frac{1}{3},\\
		0 & \text{if } \delta \geq \frac{1}{3}.
	\end{cases}
\end{align*}
From the proof of Lemma \ref{prop:discrimination_equalizing_bias}, equation \eqref{eq:equality_restoring_tau2}, the test allocation restoring equality $\hat{\tau_2}(\delta) = \frac{-(1-\delta)+\sqrt{1+3\delta^2}}{2\delta}$. 

Fix $\delta\in [1/3,1)$. Then $|\overline{\tau}_2(\delta) - 1/2|>|\hat{\tau}_2(\delta)-1/2|$, since $\hat{\tau}_2(\delta)\in (1/2,1)$, as shown in the proof of Lemma~\ref{prop:discrimination_equalizing_bias}. Next, fix $\delta\in (0,1/3)$. Then we have, equivalently,
\begin{align*}
	| \overline{\tau}_2(\delta) - 1/2| = \frac{1}{2} - \frac{1-3\delta}{2} &> 
	\frac{-1+\delta + \sqrt{1+3\delta^2}}{2\delta} - \frac{1}{2} = |\hat{\tau}_2(\delta) - 1/2|  \\
 % \frac{3\delta}{2} &> \frac{-1 + \sqrt{1+3\delta^2}}{2\delta} \\
1+3\delta^2 &> \sqrt{1+3\delta^2}
\end{align*}
which holds for any $\delta \in (0,1/3)$.

\qed
\end{document}